\documentclass[lettersize,journal]{IEEEtran}
\usepackage{amsmath,amsfonts}

\usepackage{subfig}

\usepackage{cite}
\newcommand{\myparatight}[1]{\smallskip\noindent{\bf {#1}:}~}

\usepackage{amsmath,amsfonts}
\usepackage{algorithm}
\usepackage{algorithmicx}
\usepackage{array}
\usepackage[amsmath,thmmarks]{ntheorem}
\usepackage{textcomp}
\usepackage{stfloats}
\usepackage{url}
\usepackage{verbatim}
\usepackage{graphicx}

\usepackage{tikz}
\usepackage{filecontents}
\usepackage{algpseudocode}
\usepackage{setspace}
\usepackage{booktabs}
\usepackage{multirow}
\usepackage{bm}
\usepackage{diagbox}

\usepackage{tcolorbox}
\usepackage{makecell}

\newtheorem{theorem}{Theorem}
\newtheorem{proof}{Proof}
\newtheorem{assumption}{Assumption}

\usepackage{tcolorbox}

\usepackage[hidelinks]{hyperref}

\begin{document}

\title{Practical Framework for Privacy-Preserving and Byzantine-robust Federated Learning}

\author{Baolei~Zhang,~Minghong Fang,~Zhuqing Liu, ~Biao~Yi,~Peizhao Zhou,~Yuan Wang,~Tong~Li\IEEEauthorrefmark{1},~and~Zheli~Liu
\thanks{
Baolei~Zhang,~Biao~Yi,~and~Peizhao Zhou are with College of Computer Science, Nankai University, China.~Yuan~Wang is with School of Mathematical Sciences, Nankai University, China.~Tong~Li~and~Zheli~liu are with College of Cyber Science, Nankai University, China. E-mail: zhangbaolei@mail.nankai.edu.cn,~yibiao199801@163.com,~zhoupz@mail.nankai.edu.cn, 2120220067@mail.nankai.edu.cn,~tongli@nankai.edu.cn,~liuzheli@nankai.edu.cn.

Minghong Fang is with the Department of Computer Science and Engineering, University of Louisville, USA. E-mail: minghong.fang@louisville.edu.
  
Zhuqing Liu is with the Department of Computer Science and Engineering, University of North Texas, USA. E-mail: zhuqing.liu@unt.edu.
		
\IEEEauthorrefmark{1} Corresponding Author.
}
}

\markboth{Journal of \LaTeX\ Class Files,~Vol.~14, No.~8, August~2021}%
{Shell \MakeLowercase{\textit{et al.}}: A Sample Article Using IEEEtran.cls for IEEE Journals}

\maketitle

\begin{abstract}

Federated Learning (FL) allows multiple clients to collaboratively train a model without sharing their private data. However, FL is vulnerable to \emph{Byzantine attacks}, where adversaries manipulate client models to compromise the federated model, and \emph{privacy inference attacks}, where adversaries exploit client models to infer private data. Existing defenses against both backdoor and privacy inference attacks introduce significant computational and communication overhead, creating a gap between theory and practice. To address this, we propose ABBR, a practical framework for Byzantine-robust and privacy-preserving FL. We are the first to utilize dimensionality reduction to speed up the private computation of complex filtering rules in privacy-preserving FL. Additionally, we analyze the accuracy loss of vector-wise filtering in low-dimensional space and introduce an adaptive tuning strategy to minimize the impact of malicious models that bypass filtering on the global model. We implement ABBR with state-of-the-art Byzantine-robust aggregation rules and evaluate it on public datasets, showing that it runs significantly faster, has minimal communication overhead, and maintains nearly the same Byzantine-resilience as the baselines.

\end{abstract}

\begin{IEEEkeywords}
Federated learning, Byzantine-robust defense, Privacy-preserving defense.
\end{IEEEkeywords}


\section{Introduction}

\IEEEPARstart{F}{ederated} learning (FL)~\cite{konevcny2016federated,mcmahan2017communication,yang2019federated,li2020federated,kairouz2021advances} is a machine learning framework designed to address data silo issues and privacy risks in applications like image recognition and intrusion detection. FL enables multiple clients to collaboratively train a model with a central server without sharing their private data. The process involves multiple iterations where 1) clients train local models on their data and send them to the server, and 2) the server aggregates these models to update the global model, which is then sent back to clients for further training. However, the decentralized nature of FL allows adversaries to manipulate client models, making FL vulnerable to Byzantine attacks~\cite{blanchard2017machine,gu2017badnets,DBLP:journals/csur/WangMWHQR23,DBLP:journals/corr/abs-2302-13520, bagdasaryan2020backdoor,fang2020local,DBLP:conf/sp/LiYHLWFS23,DBLP:journals/tdsc/ZhaoHWJSLH21}. These attacks involve adversaries controlling clients to corrupt the global model and degrade performance. To counter this, Byzantine-robust aggregation rules like FLAME \cite{nguyen2022flame}, Multi-Krum \cite{blanchard2017machine}, FoolsGold \cite{fung2018mitigating}, and FABA \cite{xia2019faba} have been developed.

When designing an FL scheme that can withstand Byzantine clients, it is also crucial to prioritize data privacy. Simply sharing model parameters without revealing private data is not enough to guarantee data confidentiality in FL. Specifically, an untrusted central server can launch inference attacks~\cite{DBLP:conf/sp/NasrSH19,melis2019exploiting,pyrgelis2017knock,DBLP:conf/uss/Fu0JCWG0L022,DBLP:journals/tdsc/GaoHGLZCL23} or reconstruction attacks~\cite{DBLP:conf/nips/ZhuLH19,DBLP:journals/corr/abs-2001-02610} to extract sensitive information from local models. Consequently, researchers~\cite{so2020byzantine,nguyen2022flame} suggest that Byzantine-robust aggregation rules must be combined with secure computation protocols, such as secure aggregation~\cite{bonawitz2017practical,bell2020secure} and secure multi-party computation (MPC) protocols, to prevent the exposure of local models in plaintext on the server.

However, there is a significant gap between the proposed schemes and practical implementation. Previous works~\cite{so2020byzantine,nguyen2022flame,he2020secure,DBLP:conf/esorics/DongCLWZ21} focus on implementing vector-wise filtering rules (a major category of Byzantine-robust aggregation rules, detailed in Section \ref{subsec:bar}) in privacy-preserving FL, leading to substantial computational and communication overhead. A large portion of this cost arises from the private computation of pairwise distances between local models, as vector-wise filtering rules use distance measurements to identify Byzantine clients. Therefore, the key challenge in making these schemes practical is improving the efficiency of these private distance computations.

In this paper, we propose ABBR, a practical framework for Byzantine-robust and privacy-preserving FL. The overview of ABBR is outlined in Figure~\ref{fig:overview}. We are the first to propose using dimensionality reduction to accelerate the private computation of vector-wise filtering rules in privacy-preserving FL. Meanwhile, we analyze the accuracy loss of vector-wise filtering rules in low-dimensional space and propose an adaptive tuning strategy to minimize the perturbation of unfiltered malicious local models on the global model.  Once malicious local models are filtered out, the remaining models are aggregated in the original-dimensional space. This approach significantly reduces computational and communication overhead, addressing the gap mentioned earlier.  ABBR uses the secure two-party computation (STPC) as the private computing backend, and consists of two following key components.

\myparatight{Dimensionality reduction}
In the STPC architecture, since the server only holds secret shares of local models, expensive operations (such as private multiplication, MUL, explained in Section~\ref{subsec:aby}) are needed to compute distances between local models. For example, calculating the Euclidean distance between two models requires $O(d)$ MULs, where $d$ is the number of model parameters. Drawing inspiration from using dimensionality reduction to lower overhead in data retrieval, ABBR introduces a dimensionality reduction component (detailed in Section~\ref{subsec:dimensionality_reduction}) to project secret shares of local models from high-dimensional to low-dimensional space. Unlike traditional methods, which would introduce costly operations, ABBR uses random projection and shares the projection matrix between the two servers, thus avoiding additional expensive operators in the dimensionality reduction process.

\myparatight{Adaptive tuning}
The dimensionality reduction component inevitably leads to a loss of parameter information from the local model, which introduces some error in the calculated distance. In Byzantine-robust aggregation rules, this error could cause malicious local models to be mistakenly classified as benign. To address this, we first analyze the error bound for the distance and discuss the upper bound on the potential filtering loss. Additionally, we propose an adaptive tuning strategy that minimizes the directional bias of malicious models by adaptively clipping the norm of local model updates in the original-dimensional space.

\myparatight{Our contributions} 
In summary, this paper makes the following contributions: 

\begin{itemize}
    \item We propose ABBR, a practical framework for Byzantine-robust and privacy-preserving FL. This is the first approach to use dimensionality reduction to expedite the private computation of complex filtering rules in privacy-preserving FL. With ABBR, vector-wise filtering rules can be applied in a privacy-preserving manner while significantly reducing computational and communication overhead.
    \item We analyze the accuracy loss of vector-wise filtering rules in low-dimensional space and propose an adaptive tuning strategy to minimize the perturbation of unfiltered malicious local models on the global model.
    \item  We implement an ABBR construction of the state-of-the-art vector-wise filtering rules, Multi-Krum \cite{he2020secure}, FoolsGold \cite{fung2018mitigating}, FABA \cite{xia2019faba} and FLAME \cite{nguyen2022flame}, and evaluate the construction on 4 public datasets. 
    The experimental results show that, they run significantly fast and have very
little communication overhead. Moreover, it almost has the same Byzantine-robust performance as the baselines.
\end{itemize}

ABBR proposes a promising paradigm for implementing Byzantine-robust aggregation rules in privacy-preserving FL, i.e., executing complex filtering rules in a low-dimensional space and aggregating the remaining local models in the original-dimensional space. We believe that developing Byzantine-robust aggregation rules tailored to this paradigm is a key step toward achieving practical solutions.

\section{Preliminaries}

\subsection{Federated Learning}
\label{subsec:fl}
Federated learning \cite{mcmahan2017communication,konevcny2016federated,yang2019federated} (FL) is a collaborative machine learning approach where multiple clients train a shared model without exchanging private data. A central server coordinates $n$ clients to solve an optimization problem over several iterations. Each iteration involves three main steps:

\myparatight{Step 1} The central server sends the current global model $G$ to random $n$ clients of $m$ clients.

\myparatight{Step 2} Each client updates its local model with the global model's parameters, trains it on its own data $D_i$, and then sends the local model back to the central server. 

\myparatight{Step 3} The central server aggregates these local models by a specific aggregation rule. We only introduce the aggregation rule \emph{federated averaging} (FedAvg) \cite{mcmahan2017communication} as it is commonly applied. In FedAvg, the central server computes the weighted mean of these local models, i.e.
\begin{equation}
    G = \frac{1}{ {\textstyle \sum_{i=1}^{n}} \left | D_i \right | }  {\textstyle \sum_{i=1}^{n} \left | D_i \right |L_i  } 
\end{equation}
where $\left | D_i \right |$ is the size of $D_i$. However, in practice, it is difficult to verify the size of $D_i$. Therefore, many works \cite{blanchard2017machine,munoz2019byzantine} directly compute the mean to replace the weighted mean, and so do we.

\subsection{Byzantine-robust Aggregation Rule}
\label{subsec:bar}
Recently, the machine learning community has developed Byzantine-robust aggregation rules, mainly categorized into vector-wise filtering and dimension-wise filtering rules.

\myparatight{Vector-wise filtering} aims to remove potentially malicious local models. The server identifies potentially malicious local models and filters them based on the pairwise distances (such as Euclidean distance and Cosine distance) between any two local models. Examples include Multi-Krum \cite{he2020secure}, FoolsGold \cite{fung2018mitigating}, FABA \cite{xia2019faba} and FLAME \cite{nguyen2022flame}.

\myparatight{Dimension-wise filtering} aims to remove potentially malicious values for each dimension of the client's local model separately. The server performs statistical analysis (such as median and trimmed mean) on each dimension of all local models to filter out potentially malicious parameters of the local models. Examples include Median~\cite{yin2018byzantine} and Trimmed-mean~\cite{yin2018byzantine}.

Prior works \cite{so2020byzantine,nguyen2022flame,he2020secure,DBLP:conf/esorics/DongCLWZ21} have implemented vector-wise filtering rules instead of dimension-wise filtering rules in privacy-preserving FL. This is because the private operators (such as secure addition, subtraction, and multiplication) required by vector-wise filtering rules produce lower overhead than other private operators (such as secure comparison and multiplexer) required by dimension-wise filtering rules.

\subsection{Secure Two-Party Computation Framework : ABY }
\label{subsec:aby}
ABY \cite{DBLP:conf/ndss/Demmler0Z15} is the default STPC backend for our framework, which combines secure computation schemes \cite{DBLP:conf/crypto/Beaver91a,DBLP:conf/stoc/GoldreichMW87,DBLP:conf/focs/Yao82b} based on Arithmetic sharing, Boolean sharing, and Yao sharing. In ABY, a private value can be secretly shared between two parties through three types of sharing, and there are efficient conversions between the three types. It is worth noting that the operations addition (ADD), subtraction (SUB), and multiplication (MUL) on Arithmetic sharing are more efficient than those on Boolean sharing and Yao sharing. The operations comparison (CMP) and multiplexer (MUX) are only implemented on Boolean and Yao sharing. We introduce above private operations in detail in Appendix~\ref{aby_detail}.

In STPC, since ADD and SUB do not require communication, they are generally considered free operators. In contrast, MUL, CMP, and MUX require the two parties to communicate some intermediate information, so they are generally considered expensive operators.


\section{ABBR Framework}
In this section, we first introduce the threat model. Then, we discuss challenges for improving the efficiency of privacy-preserving and Byzantine-resilient defense in FL. Lastly, we present an overview of our ABBR framework and a simple approach to using ABBR.

\begin{figure*}[!th]
\centering
\includegraphics[height=6.5cm]{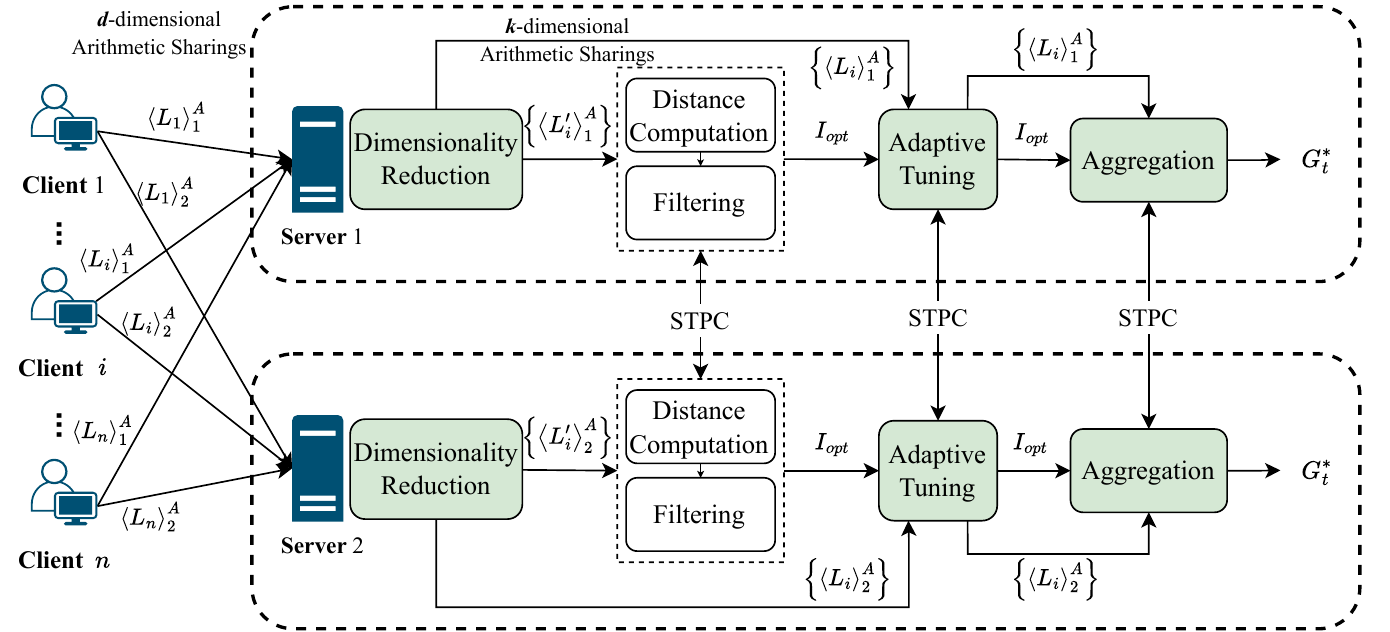}
\caption{Overview of ABBR in $t$-th iteration of FL.}
\label{fig:overview}
\end{figure*}

\subsection{Problem Setting}
\label{subsec:problem_statement}
\myparatight{Based on two-server architecture}
Previous works~\cite{so2020byzantine,nguyen2022flame,he2020secure,DBLP:conf/esorics/DongCLWZ21} have implemented vector-wise filtering rules (we explain the reasons for not adopting dimension-wise filtering rules in Section \ref{subsec:bar}) in privacy-preserving FL by using two different architectures. In the single-server architecture~\cite{so2020byzantine}, clients share secret splits of their local model parameters with each other. All clients need to perform the sophisticated privacy computation of vector-wise filtering rules. In the two-server architecture~\cite{nguyen2022flame,he2020secure,DBLP:conf/esorics/DongCLWZ21}, clients share secret splits of local model parameters between two servers. Only two servers need to perform the sophisticated privacy computation. The single-server architecture places heavy computational and communication burdens on clients, which is challenging for those with limited resources. In contrast, the two-server architecture shifts most of this overhead to the servers, which typically have greater computational power and bandwidth. Therefore, we believe that the two-server architecture is more promising.

\myparatight{Threat Model} In the two-server architecture, FL systems face two main threats, as described in previous works \cite{nguyen2022flame,he2020secure,DBLP:conf/esorics/DongCLWZ21}. We assume that a malicious adversary can corrupt less than 50\% of the clients, including their data, training processes, and local models. The remaining clients are benign and aim to obtain an accurate global model, while the compromised clients attempt to  manipulate the global model to produce incorrect predictions. The two servers are honest-but-curious, meaning they follow the protocol but try to learn as much as possible from the training data without colluding with others.

\myparatight{The gap between theory and practice} 
Previous studies~\cite{he2020secure,DBLP:conf/esorics/DongCLWZ21} employing a two-server architecture have used STPC to implement vector-wise filtering rules in privacy-preserving FL. However, even servers equipped with more advanced computational resources struggle to handle the massive computational and communication overhead introduced by STPC, leading to a considerable gap between theoretical designs and practical applications. For example, the private FLAME protocol \cite{nguyen2022flame} requires a total server-side runtime of 14,840 seconds (approximately 4 hours, as indicated in Table \ref{tab:setup_runtime} and Table \ref{tab:online_runtime}) for just a single training iteration with ResNet-18 and 32 clients. This excessive overhead and prolonged runtime significantly hinder the practical applicability of such methods.

As a result, it is urgent to reduce the overhead of these methods and improve their practicality.

\subsection{Challenge}
\label{subsec:challenge}
We observe that the overhead is mainly spent on private computations of pairwise distances between any two local models when implement vector-wise filtering rules in privacy-preserving FL. Dimensionality reduction offers a potential solution by projecting all local models into a shared low-dimensional space, but there are two challenges to be addressed.

\myparatight{C1 - The private computation of dimensionality reduction may produce more overhead}Dimensionality reduction techniques can be formalized as the multiplication between the data $u \in \mathbb{R}^d$ (in vector form) and the projection matrix $R \in \mathbf{R}^{d \times k}$ ( $k < d$) i.e.
\begin{equation}
    u' = u \times R.
\end{equation}
However, in the two-server architecture, local models are stored as secret shares on the server. If the dimensionality reduction method is data-dependent, the projection matrix must also be secret to maintain model privacy. This results in a large number of expensive operations in STPC for performing private multiplications between each local model and the projection matrix. Consequently, the overhead from these operations could outweigh the reduction in computations gained by using pairwise distances between low-dimensional local models.

\myparatight{C2 - The vector-wise filtering rules may not be accurate in low-dimensional spaces} Dimensionality reduction inherently causes information loss in local models. This results in an error margin when computing distances in the low-dimensional space compared to the original-dimensional space (see~\ref{subsubsec:error_analysis}). Consequently, the accuracy of vector-wise filtering based on low-dimensional distances is affected, potentially allowing malicious local models to evade detection. Furthermore, since the perturbations in vector-wise filtering are unpredictable, mitigating this accuracy loss poses a significant challenge.

\subsection{System Overview}

\begin{algorithm}[!t]
	\caption{Prototype of ABBR}
	\label{alg:overview}
	\small
    \begin{spacing}{1}
	\begin{algorithmic}[1]
	    \State \textbf{Input:} $G_0,T,m,n,r,d,\varepsilon,\eta$ \Comment{$G_0$ is the initial global model, $T$ is the number of FL iterations,$m$ is the number of total clients, $n$ is the number of selected clients in an iteration,  $r$ is the ID of server, $d$ is the original dimension of local model, $\varepsilon$ is the desired error in distance preservation, $\eta$ is the probability of successful projection.}
	    \State \textbf{Output:} $G_T$ \Comment{$G_T$ is the final global model.}
	    \For{each training iteration $t$ in $\left [ 1,T \right ] $}
            \State randomly selects $n$ clients from $m$ clients and generate the index set $I_n = \{1,\cdots, n \}$ 
	    \For{each client $i$ in $I_n$}
	    \State ${\langle L_i \rangle}^A \gets$ \Call{ClientProcess}{$G_{t-1},r$}
	    \EndFor
	    \State $\{{\langle L'_1 \rangle}^A,\cdots,{\langle L'_n \rangle}^A\} \gets$ \Call{DimensionalityReduction}{$\{{\langle L_1 \rangle}^A,\cdots,{\langle L_n \rangle}^A\},d,\varepsilon,\eta$} \Comment{$ \forall i \in (1, \dots  n), {\langle L'_i \rangle}^A$ is the secret sharings of low-dimensional local model.}
	    \State $\{{\langle e_{11} \rangle}^A,\cdots,{\langle e_{nn} \rangle}^A\} \gets$ \Call{DistanceComputation}{$\{{\langle L'_1 \rangle}^A,\cdots,{\langle L'_n \rangle}^A\}$} \Comment{$ \forall i,j \in (1, \dots  n), {\langle e_{ij} \rangle}^A$ is the secret sharings of pairwise distance between local model $L'_{i}$ and $L'_{j}$ in the low-dimensional space.}
        \State $I_{opt} \gets$ \Call{Filtering}{$\{{\langle e_{11} \rangle}^A,\cdots,{\langle e_{nn} \rangle}^A\}$} \Comment{$ I_{opt}$ is the index set of accepted local models.}
	    \State $ \{{\langle L_1 \rangle}^A,\cdots,{\langle L_n \rangle}^A\} \gets$ \Call{AdaptiveTuning}{$I_{opt}, G_{t-1}, \{{\langle L_1 \rangle}^A,\cdots,{\langle L_n \rangle}^A\},\{{\langle L'_1 \rangle}^A,\cdots,{\langle L'_n \rangle}^A\}$} 
        \State $G_t \gets$ \Call{Aggregation}{$I_{opt},{\langle L_1 \rangle}^A,\cdots,{\langle L_n \rangle}^A$}
	    \EndFor
	    \\
	    \Function {ClientProcess}{$G_{t-1},r$}
	     \State $L_{init} = G_{t-1}$
	     \State $L \gets$  \Call{LocalTraining}{$L, D$} \Comment{The client trains the local model with its local dataset $D$.}
        \State ${\langle L \rangle}_1^A,{\langle L \rangle}_2^A \gets$ \Call{ArithmeticSharing}{$L$} \Comment{The client generates two arithmetic sharings of local model $L$ by using $d$ private operators $\mathbf{SHR} $.}
        \State \Return ${\langle L \rangle}_r^A$
	    \EndFunction
		
	\end{algorithmic}
	\end{spacing}
\end{algorithm}

\myparatight{High-level idea} We present the structure of our ABBR framework, which is a practical framework for privacy-preserving and Byzantine-robust federated learning. The key insight is to filter out malicious local models through vector-wise filtering rules in low-dimensional space and aggregate benign local models in high-dimensional space. As shown in Figure \ref{fig:overview}, the $i$-th iteration of ABBR includes two main components: dimensionality reduction and adaptive tuning. To address challenge C1, we use a random projection for dimensionality reduction (see Section \ref{subsec:dimensionality_reduction}), leveraging its data-independent nature. This allows the server to use only two private operators (ADD and SUB) when projecting secret shares of local models to the same low-dimensional space, In the low-dimensional space, the distance computation component we built in STPC only needs $O(n^2 \log n)$ private multiplication operators, while the previous protocol required $O(n^2d)$. Notably, in typical FL applications~\cite{47976, Hard2018FederatedLF}, $d$ is usually three orders of magnitude larger than $n$. For challenge C2, we implement an adaptive tuning component (see Section \ref{subsubsec:tunning_strategy}) to adaptively tune the norm of remaining local model updates after filtering. This can effectively limit the malicious models that are not filtered out, thereby minimizing the accuracy loss of vector-wise filtering rules.

\myparatight{Using ABBR} Algorithm \ref{alg:overview} outlines the workflow of the $t$-th iteration of ABBR, which is simple to use and requires only two steps. First, the developer selects an appropriate vector-wise filtering method based on robustness needs, which involves configuring three modules: distance computation, filtering, and aggregation. ABBR offers private computations for Euclidean and Cosine distances, allowing the developer to choose one for the distance computation module. For filtering, the developer must use STPC private operators (e.g., MUL, CMP, MUX) to develop a private computation function. For aggregation, ABBR provides additive aggregation for the secret shares of accepted local models, and developers can also customize aggregation rules using STPC operators. Second, the developer configures hyperparameters of the dimensionality reduction component (see Section \ref{subsec:dimensionality_reduction}) to manage the accuracy loss of dimensionality reduction on vector-wise filtering (see Section \ref{subsubsec:error_analysis}).

\myparatight{Complexity comparison}
Communication overhead is the primary bottleneck in privacy-preserving FL with computationally powerful servers. Our analysis thus targets communication complexity, predominantly arising from secure inter-server operations (e.g., secure multiplications) during pairwise distance calculations. While this distance computation is the main communication driver for both existing private baselines~\cite{he2020secure,fung2018mitigating,xia2019faba,nguyen2022flame} and our ABBR method, their complexities diverge significantly. Private baselines incur $O(n^2d)$ communication, whereas ABBR reduces this to $O(n^2\log n)$. Considering that model dimension $d$ is typically orders of magnitude larger than client count $n$ per round in large-scale FL scenarios~\cite{shejwalkar2022back}, ABBR’s substantially lower communication complexity translates to significant improvements in efficiency and scalability.


\section{Dimensionality Reduction}
\label{subsec:dimensionality_reduction}

The dimensionality reduction component should ensure that its privacy-preserving computation introduces less overhead than the reduction it provides. While various dimensionality reduction methods, such as principal component analysis, they are not suitable here. This is because the privacy of local models must be maintained during the process. Specifically, generating a projection matrix would need to occur on the secret shares of local models, and multiplying the secret shares of local models with the secret shares of the projection matrix would be required. The overhead from these steps could surpass the savings gained from computing pairwise distances between low-dimensional secret shares of local models.

To solve the challenge, we make the privacy-preserving dimensionality reduction data-independent by using random projection~\cite{DBLP:journals/jcss/Achlioptas03,arriaga2006algorithmic}. In this approach, both servers generate the same projection matrix locally and keep it in plaintext. Since the plaintext matrix doesn't reveal local model privacy and introduces no communication overhead when multiplied with secret shares, the privacy concerns and overhead are avoided.

The dimensionality reduction component involves three steps: parameter negotiation, projection matrix generation, and dimensionality reduction execution. The detailed algorithm is shown in Algorithm~\ref{alg:dimensionality_reduction} in Appendix.

\myparatight{Parameter negotiation}
A random one server determines a random seed $s$ of the Pseudorandom Generator (PRG) and the dimensionality $k$ of reduced models, then sends them to the other server. The value of $k$ is determined by the following inequality:
\begin{equation}
\label{Target_dimension}
k \ge \left \lceil \frac{4+2\eta}{\varepsilon ^ 2 - \varepsilon ^ 3}\log{n } \right \rceil , \\ 
\end{equation}
where $n$ is the number of local models, the $\eta>0$ and $\varepsilon>0$ are the hyperparameters of the dimensionality reduction component. The $\eta$ controls the probability that distance error between any two low-dimensional local models is within a factor of $\varepsilon$ (detailed in Section~\ref{subsubsec:error_analysis}). 

\myparatight{Projection matrix generation}
Two servers use the same random seed $s$ to generate a same $d \times k$ matrix $\mathbf{P}$ through PRG, each entry of which conforms to the following distribution:
\begin{equation}
\label{Distribution}    
\mathbf{P}_{ij} =\begin{cases}+1 & \text{with probability 1/2}\\
	    -1 & \text{with probability 1/2}
        \end{cases},
\end{equation}
where $\mathbf{P}_{ij}$ is the entry at $i$-th row and $j$-th column. It is worth noting that the matrix $\mathbf{P}$ is in plaintext on both servers and generating the secret shares of low-dimensional local models from it does not reveal any privacy.

Above distribution is chosen due to several advantageous properties. First, generating $\mathbf{P}$ is computationally inexpensive, requiring only a source of random bits. Furthermore, multiplication involving $\mathbf{P}$ and secret shares of local models can be performed more efficiently because it requires only local operations (shown in Equation~\ref{equ:dimension_reduction}) between servers. Second, this distribution is well-known for its ability to approximately preserve pairwise distances in the projected space, which is detailed in Section~\ref{subsubsec:error_analysis}. Third, the projection matrix $\mathbf{P}$ is generated independently of the local models. This simplifies the process and is advantageous from a privacy perspective, as the projection itself does not depend on sensitive data.

\myparatight{Dimensionality reduction execution}
Since the server only holds the secret shares of local models, it needs to perform private operators of STPC to complete the privacy-preserving computation of dimensionality reduction. When reducing the dimensionality of the local model $L$, two server need to perform the same operations to generate secret shares of the $i$-th element of the low-dimensional local model $L'$ by
\begin{equation}
\label{equ:dimension_reduction}
\langle L'[i] \rangle = \sum_{\substack{j \in \mathcal{U} \ \text{s.t. } \mathbf{P}{ji}=1 }} \left \langle L[j] \right \rangle - \sum_{\substack{j \in \mathcal{U} \ \text{s.t. } \mathbf{P}_{ji}=-1 }} \left \langle L[j] \right \rangle,
\end{equation}
where  $\mathcal{U} = \{1, 2, \dots, d\}$.  The index $j$ in the first summation iterates through all indices in $\mathcal{U}$ such that the corresponding entry in the projection matrix $\mathbf{P}_{ji}$ is equal to $+1$. Similarly, the index $j$ in the second summation iterates through all indices in $\mathcal{U}$ such that the corresponding entry $\mathbf{P}_{ji}$ is equal to $-1$. We can see that these operations only require two free operators, i.e., private addition and subtraction. 

\begin{figure}[t]
\centering

\includegraphics[height=3cm]{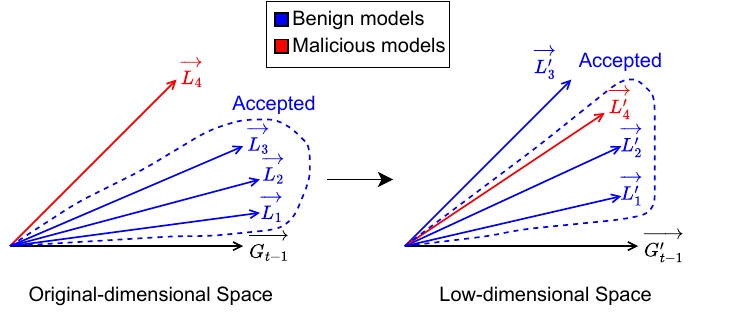}
\caption{Changes in the relative position of local models after being projected from original-dimensional space to low-dimensional space.}
\label{fig:tuning_1}
\end{figure}
\section{Adaptive Tuning}
\label{subsec:adaptive_tuning}

In this section, we first analyze the loss of vector-filtering after dimensionality reduction. Then, we present an adaptive tuning strategy to mitigate the negative influence. 
\subsection{Error Analysis}
\label{subsubsec:error_analysis}

The dimensionality reduction component can be directly applied to all prior works in which the Byzantine-resilient aggregation is vector-wise filtering. As described in Section \ref{subsec:dimensionality_reduction}, the distance between any two local models calculated in low-dimensional space may have errors. We observe that the robustness of vector-wise filtering is perturbed by the error of Euclidean distance or Cosine distance in the low-dimensional space. Therefore, we first analyze the error of the Euclidean and Cosine distance between any two low-dimensional models. Then, we describe our assumptions and describe our theoretical results that the influence of vector-filtering after dimensionality reduction.

\begin{theorem}
\label{theorem_1}
Suppose that the target dimensionality $k$ of low-dimensional models satisfies the Inequation \ref{Target_dimension} and the distribution of the $d \times k$ projection matrix $\mathbf{P}$ satisfies the Equation \ref{Distribution}, the error of Euclidean distance and cosine distance of any two local models in the low-dimensional space is bounded. Formally, we have the error bound of Euclidean distance with the probability at least $1-2e^{-(1+\frac{1}{2}\eta)\log n  }$:
\begin{equation}
\label{Euclidean_error}
(1-\varepsilon){\left \| L_i-L_j \right \|}^2 \leq {\left \| L'_i-L'_j \right \|}^2 \leq (1+\varepsilon){\left \| L_i-L_j \right \|}^2,
\end{equation}
where $L'_i$ and $L'_j$ are the low-dimensional local models after performing the dimensionality reduction component on local model $L_i$ and $L_j$. We have the error bound of cosine distance  with the probability at least $1-4e^{-(1+\frac{1}{2}\eta)\log n  }$:
\begin{multline}
\label{Cosine_error3}
\frac{1}{1-\varepsilon} \left ( 1-  \frac{L_i}{\left \| L_i \right \|} \cdot\frac{L_j}{\left \| L_j \right \|} -2\varepsilon \right )  \leq 1- \frac{L'_i}{\left \| L'_i \right \|}\cdot \frac{L'_j}{\left \| L'_j \right \|} \\
\leq \frac{1}{1+\varepsilon} \left ( 1- \frac{L_i}{\left \| L_i \right \|} \cdot\frac{L_j}{\left \| L_j \right \|} +2\varepsilon \right ) .
\end{multline}

\end{theorem}
\begin{proof}
    See Appendix~\ref{proof:theo_1}.
\end{proof}

\begin{figure}[t]
\centering
\subfloat[Adaptive tuning]{\label{fig:tuning_2_a} \includegraphics[height=3cm]{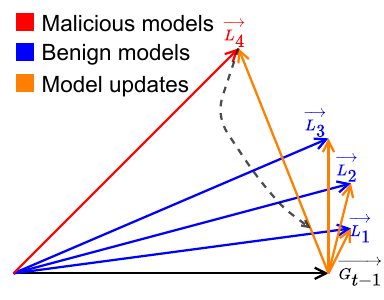}}
\subfloat[Aggregation]{\label{fig:tuning_2_b} \includegraphics[height=3cm]{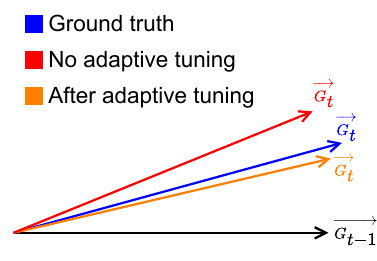}}\\
\caption{High-level idea of adaptive tuning strategy.}
\label{fig:tuning_2}
\end{figure}

\begin{assumption}
\label{assumption_1}
When vector-wise filtering does not consider dimensionality reduction, the difference between the aggregation of the accepted local models and the optimal aggregation of all benign local models is bounded. Formally, when the vector-wise filtering is based on Euclidean distances, we have the following for any $\alpha > 0$ :
\begin{equation}
\label{equ:assumption_1_1}
{\left \| G_t-G_t^* \right \|}^2 \leq \alpha,
\end{equation}
where $G_t$ is the aggregation of the local models accepted by vector-wise filtering in iteration $t$ of FL, $G_t^*$ is the optimal aggregation of all benign local models in iteration $t$ of FL. And when the vector-wise filtering is based on Cosine distances, we have the following for any $\alpha \in [0,2]$:
\begin{equation}
\label{equ:assumption_1_2}
1- \frac{G_t}{\left \| G_t \right \|} \cdot\frac{G_t^*}{\left \|G_t^* \right \|} \leq \alpha.
\end{equation}

\end{assumption}

\begin{theorem}
\label{theorem_2}
Suppose Assumption \ref{assumption_1} holds and the error of pairwise distances is bounded according to Theorem \ref{theorem_1}. When the vector-wise filtering is in our ABBR framework, the difference between the aggregation of the accepted local models and the optimal aggregation of all benign local models is bounded. Formally, when the vector-wise filtering is based on Euclidean distances, we have the following for any $\alpha > 0$ :
\begin{equation}
\label{equ:assumption_1_1}
{\left \| G_t-G_t^* \right \|}^2 \leq \frac{1+\varepsilon}{1-\varepsilon} \alpha,
\end{equation}
where $G_t$ is the aggregation of the local models accepted by vector-wise filtering in iteration $t$ of FL, $G_t^*$ is the optimal aggregation of all benign local models in iteration $t$ of FL. And when the vector-wise filtering is based on Cosine distances, we have the following for any $\alpha \in [0,2]$:
\begin{equation}
\label{equ:assumption_1_2}
1- \frac{G_t}{\left \| G_t \right \|} \cdot\frac{G_t^*}{\left \|G_t^* \right \|} \leq \frac{1-\varepsilon}{1+\varepsilon} (\alpha + 2\varepsilon) + 2\varepsilon.
\end{equation}

\end{theorem}

\begin{proof}
    See Appendix~\ref{proof:theo_2}.
\end{proof}

\begin{table}[t]
        \small

	\centering
	\renewcommand\arraystretch{1.5}
        \caption{Datasets and models used in our evaluations.}
        \label{data_model}
        
	\begin{tabular}{c|c|c|c}
		\hline
		
            Datasets &\#Records &Model &\#params\\
            \hline
            CIFAR-10 &60k &\makecell{ResNet-18 \\ Lightweight }&$\approx$2.8M\\
            \hline
            \makecell{EMNIST \\ (Digits)} &280k &\makecell{CNN \\ (2 conv and 2 fc)} &$\approx$266k\\
            \hline
            Fashion-MNIST &70k &\makecell{ResNet-18 \\ Featherweight} &$\approx$700k\\
            \hline
            Tiny-ImageNet &120k &ResNet-18 &$\approx$11M\\
            \hline
            
		\hline  
	\end{tabular}

\end{table}

\subsection{Tuning Strategy}
\label{subsubsec:tunning_strategy}

As we analyzed in Section~\ref{subsubsec:error_analysis}, the distortion of the distance causes the upper bound on the deviation between the aggregated global model and the optimal global model to become larger. Therefore, we propose the adaptive tuning strategy to reduce the deviation between the aggregated global model and the optimal global model. Due to distance distortion, some malicious local models may be misclassified into the accepted group. Figure \ref{fig:tuning_1} depicts the change in the relative position of local models after they are projected from the original-dimensional space to the low-dimensional space. We can observe that in the original-dimensional space, the relative positions between benign local models are relatively close, so they are accepted by vector-wise filtering. However, the relative positions between the benign local model $L_3$ and the malicious local model $L_4$ are reversed in the low-dimensional space. As a result, the malicious local model $L_4$ is accepted but the benign local model $L_3$ is filtered out. In addition, since ABBR aggregates accepted local models in the original-dimensional space, the aggregated model is inevitably biased in the direction of the malicious model $L_4$.

Since the low-dimensional space is random, it is difficult to predict how the relative positions of local models will shift, making it challenging to mitigate the negative effects of incorrectly accepting malicious models. Drawing on the concept of norm clipping used in previous studies \cite{nguyen2022flame, sun2019can} to reduce the influence of malicious models on the global model, we propose an adaptive tuning strategy, as described in Algorithm~\ref{alg:adaptive_tuning} in Appendix. In this strategy, the clipping bound and factor for local updates are dynamically adjusted in the low-dimensional space, and local updates with norms exceeding the clipping bound are adaptively clipped in the original-dimensional space.

\myparatight{Dynamically Determining the clipping bound and clipping factor} 
As shown in Figure \ref{fig:tuning_1}, in the original-dimensional space, the malicious local model deviates further from the global model compared to the benign local models. This occurs because the adversary's objective of manipulating the local model is entirely different from the training goal of benign models~\cite{cao2020fltrust,nguyen2022flame}. Based on this observation, and assuming that more than half of the clients are benign, we can infer that the median norm of all local model updates is more likely to belong to a benign model. Therefore, we set the clipping bound $S_1$ as the median of the $\ell_2$ norms of the $n$ model updates, $S_1 = \Call{Median}{e_1, e_2, \cdots, e_n}$, where $e_i$ represents the norm of the $i$-th local model update. Additionally, the smaller clipping bound $S_2$ is determined as the minimum norm, $S_2 = \Call{Min}{e_1, e_2, \cdots, e_n}$. Finally, using the bounds $S_1$ and $S_2$, we compute the adaptive clipping factor $\gamma$ for each local model updates in the accepted set as follows:
\begin{equation}
\label{clipping_factor}
\gamma_i =
\begin{cases}
 \frac{S_2}{e_i} & \text{ if } e_i > S_1 \\
 1 &  \text{ if } e_i \le S_1,
\end{cases}
\end{equation}

Our approach significantly differs from previous works~\cite{nguyen2022flame, sun2019can} in how we determine the clipping factor. Instead of simply clipping local model updates with norms exceeding the bound $S_1$, we clip them to the smaller bound $S_2$, as shown in Figure~\ref{fig:tuning_2_a}. The norms of benign local model updates within the accepted set are usually smaller than $S_1$. If a malicious model update is clipped only to $S_1$, its norm may still surpass that of some benign model updates, causing the aggregated model to be heavily influenced by the malicious model. In our approach, the malicious model update is clipped to the smaller bound $S_2$. Since most benign model updates remain unclipped, the resulting aggregated model will tend to be closer to the true model, as shown in Figure~\ref{fig:tuning_2_b}. Figures~\ref{fig:bcr_impact} to~\ref{fig:alpha_impact} demonstrate that our adaptive tuning strategy is highly effective, especially for vector-wise filtering like FoolsGold, greatly enhancing its robustness.

\myparatight{Efficient adaptive clipping}
Calculating the norms of all local model updates in the original-dimensional space would require $O(nd)$ private multiplication operations in STPC, so we opt to compute the norm and determine the clipping factor in the low-dimensional space. While there may be some errors in the clipping factor, leading to malicious local updates not being fully clipped to the bound $S_2$, our experiments (Figures~\ref{fig:bcr_impact} to~\ref{fig:alpha_impact}) show that this strategy does not compromise the resilience of ABBR. Furthermore, even if we calculate the norm in the original-dimensional space, the additional overhead would not diminish the efficiency advantage ABBR holds over previous methods.

\myparatight{Privacy analysis} We discuss whether exposing the clipping factor as plaintext would compromise privacy. First, it is difficult for the server to infer the norm $e_i$ of the $i$-th low-dimensional local model update by the clipping factor $\gamma_i$. This is because the server does not know the norm of each low-dimensional local model update, $S_1$ and $S_2$. Second, the norm of the low-dimensional local model update is a random number within the error range of the norm of the original-dimensional local model update. Even if the server infers the norm of the low-dimensional local model update, it is difficult to infer the norm of the original-dimensional local model update. If there is a stronger need for privacy protection, the clipping factor can also be kept private, the additional overhead introduced does not change the efficiency advantage of ABBR over prior methods.


\begin{table*}[htp]
    \footnotesize
	\centering
	
	\caption{The runtime (in seconds) of private vector-wise filtering in the setup phase. The notation ``-'' denotes no result due to the method exceeding the 30,000s time limit.}
	\label{tab:setup_runtime}
    \renewcommand\arraystretch{1}
	\begin{tabular}{c|c|c|c|c|c|c|c|c|c}
		\hline

            \multirow{2}{*}{Model} &\multirow{2}{*}{\#clients}  &\multicolumn{4}{c|}{Baseline} &\multicolumn{4}{c}{ABBR}\\
  
            &\multirow{2}{*}{~} &Multi-Krum  &FoolsGold &FABA &FLAME &Multi-Krum &FoolsGold  &FABA & FLAME \\
		
		\hline
		\multirow{6}{*}{\makecell{CNN \\ (\#params $\approx 266$k)}} 
		& 4   & 4.05  & 6.79  & 11    & 6.78  & \textbf{0.37} & \textbf{0.45} & \textbf{0.53} & \textbf{0.43} \\
 & 8   & 22    & 27    & 43    & 27    & \textbf{1.1}  & \textbf{1.28} & \textbf{1.94} & \textbf{1.25} \\
 & 16  & 85    & 96    & 171   & 96    & \textbf{4.29} & \textbf{4.7}  & \textbf{7.92} & \textbf{4.66} \\
 & 32  & 390   & 398   & 755   & 398   & \textbf{17}   & \textbf{18}   & \textbf{32.}  & \textbf{18}   \\
 & 64  & 1573  & 1606  & 2702  & 1606  & \textbf{73}   & \textbf{77}   & \textbf{142}  & \textbf{77}   \\
 & 128 & 6445  & 6578  & 12444 & 6576  & \textbf{286}  & \textbf{332}  & \textbf{625}  & \textbf{329}  \\
		\hline
		\multirow{6}{*}{\makecell{ResNet-18 Light \\ (\#params $\approx 2.8$M)}} 
		& 4   & 32    & 69    & 84    & 69    & \textbf{2.85} & \textbf{2.92} & \textbf{3.01} & \textbf{2.91} \\
 & 8   & 158   & 216   & 358   & 216   & \textbf{6.55} & \textbf{6.74} & \textbf{7.4}  & \textbf{6.71} \\
 & 16  & 744   & 879   & 1690  & 879   & \textbf{16}   & \textbf{16}   & \textbf{20}   & \textbf{16}   \\
 & 32  & 2839  & 2999  & 6937  & 2999  & \textbf{42}   & \textbf{43}   & \textbf{57}   & \textbf{43}   \\
 & 64  & 11412 & 12522 & 21345 & 12522 & \textbf{128}  & \textbf{133}  & \textbf{197}  & \textbf{133}  \\
 & 128 & -     & -     & -     & -     & \textbf{411}  & \textbf{458}  & \textbf{751}  & \textbf{455}  \\
		\hline
		\multirow{6}{*}{\makecell{ResNet-18 \\ (\#params $\approx 11$M)}} 
		& 4   & 141   & 271   & 288   & 271   & \textbf{8.69} & \textbf{8.76} & \textbf{8.85} & \textbf{8.75} \\
 & 8   & 709   & 804   & 1631  & 804   & \textbf{23}   & \textbf{23}   & \textbf{24}   & \textbf{23}   \\
 & 16  & 3161  & 2842  & 5173  & 2842  & \textbf{54}   & \textbf{54}   & \textbf{57}   & \textbf{54}   \\
 & 32  & 10814 & 10653 & 22626 & 10653 & \textbf{125}  & \textbf{126}  & \textbf{140}  & \textbf{126}  \\
 & 64  & -     & -     & -     & -     & \textbf{321}  & \textbf{326}  & \textbf{390}  & \textbf{326}  \\
 & 128 & -     & -     & -     & -     & \textbf{830}  & \textbf{876}  & \textbf{1169} & \textbf{873} \\
		\hline
	\end{tabular}
	
\end{table*}

\begin{table*}[htp]
    \footnotesize
	\centering
	
	\caption{The runtime (in seconds) of private vector-wise filtering in the online phase.}
	\label{tab:online_runtime}
    \renewcommand\arraystretch{1}
	\begin{tabular}{c|c|c|c|c|c|c|c|c|c}
		\hline
	
            \multirow{2}{*}{Model} &\multirow{2}{*}{\#clients}  &\multicolumn{4}{c|}{Baseline} &\multicolumn{4}{c}{ABBR}\\
  
            &\multirow{2}{*}{~} &Multi-Krum  &FoolsGold &FABA &FLAME &Multi-Krum &FoolsGold  &FABA & FLAME \\
		
		\hline
		\multirow{6}{*}{\makecell{CNN \\ (\#params $\approx 266$k)}} 
		& 4   & 1.15 & 1.74 & 2.34 & 1.72 & \textbf{0.08} & \textbf{0.11} & \textbf{0.11} & \textbf{0.09} \\
 & 8   & 5.28 & 6.92 & 9.11 & 6.87 & \textbf{0.22} & \textbf{0.27} & \textbf{0.52} & \textbf{0.22} \\
 & 16  & 22   & 24   & 36   & 24   & \textbf{0.79} & \textbf{0.9}  & \textbf{2.25} & \textbf{0.84} \\
 & 32  & 89   & 103  & 164  & 104  & \textbf{3.32} & \textbf{3.8}  & \textbf{9.7}  & \textbf{3.89} \\
 & 64  & 390  & 369  & 626  & 371  & \textbf{13}   & \textbf{16}   & \textbf{45}   & \textbf{17}   \\
 & 128 & 1444 & 1467 & 3013 & 1474 & \textbf{35}   & \textbf{70}   & \textbf{210}  & \textbf{76}   \\
		\hline
		\multirow{6}{*}{\makecell{ResNet-18 Light \\ (\#params $\approx 2.8$M)}} 
		& 4   & 12   & 19   & 25   & 19   & \textbf{0.4}  & \textbf{0.43} & \textbf{0.11} & \textbf{0.41} \\
 & 8   & 55   & 68   & 111  & 68   & \textbf{0.58} & \textbf{0.64} & \textbf{0.52} & \textbf{0.59} \\
 & 16  & 236  & 261  & 441  & 261  & \textbf{1.23} & \textbf{1.38} & \textbf{2.25} & \textbf{1.32} \\
 & 32  & 1000 & 996  & 1645 & 996  & \textbf{4.01} & \textbf{4.5}  & \textbf{9.7}  & \textbf{4.59} \\
 & 64  & 4046 & 3962 & 6411 & 3963 & \textbf{14}   & \textbf{17}   & \textbf{45}   & \textbf{18}   \\
 & 128 & -    & -    & -    & -    & \textbf{37}   & \textbf{72}   & \textbf{210}  & \textbf{78}   \\
		\hline
		\multirow{6}{*}{\makecell{ResNet-18 \\ (\#params $\approx 11$M)}} 
		& 4   & 47   & 81   & 98   & 81   & \textbf{1.29} & \textbf{1.32} & \textbf{0.11} & \textbf{1.30} \\
 & 8   & 222  & 278  & 429  & 278  & \textbf{1.67} & \textbf{1.72} & \textbf{0.52} & \textbf{1.67} \\
 & 16  & 928  & 1042 & 1619 & 1042 & \textbf{2.73} & \textbf{2.86} & \textbf{2.25} & \textbf{2.8}  \\
 & 32  & 4031 & 4187 & 6678 & 4187 & \textbf{6.5}  & \textbf{7}    & \textbf{9.7}  & \textbf{7.09} \\
 & 64  & -    & -    & -    & -    & \textbf{18}   & \textbf{22}   & \textbf{45}   & \textbf{23}   \\
 & 128 & -    & -    & -    & -    & \textbf{46}   & \textbf{81}   & \textbf{210}  & \textbf{88} \\
		\hline
	\end{tabular}
	
\end{table*}

\section{Evaluation}

In this section, we evaluate the performance and robustness of our ABBR framework. We implement the ABBR construction of state-of-the-art vector-wise filtering rules Multi-Krum \cite{he2020secure}, FoolsGold \cite{fung2018mitigating}, FABA \cite{xia2019faba} and FLAME \cite{nguyen2022flame}, denoted as ABBR versions. In addition, we implement the private Multi-Krum, FoolsGold, FABA, and FLAME as baselines according to the private computation scheme in \cite{nguyen2022flame}. To evaluate the performance, we compare the runtime and communication overhead between ABBR versions and baselines. To evaluate the resilience to Byzantine failures, we compare the Byzantine-resilient performance between ABBR versions and baselines under various Byzantine attacks.

\subsection{Setup}
\label{subsec:experimental_setup}

\myparatight{Datasets}
We consider four datasets: MNIST \cite{deng2012mnist}, Fashion-MNIST \cite{xiao2017fashion}, CIFAR-10 \cite{krizhevsky2009learning} and Tiny-ImageNet \cite{le2015tiny}, which are commonly used to evaluate Byzantine attacks and defenses in FL. The details of them are described in Appendix~\ref{detailed_setup}.

\myparatight{Models}
For performance evaluation, we use three models of varying sizes to comprehensively assess the performance benefits of the ABBR framework: a convolutional neural network (CNN), lightweight ResNet-18~\cite{he2016deep}, and ResNet-18~\cite{he2016deep}. The CNN, lightweight ResNet-18, and ResNet-18 models have 266K, 2.8M, and 11M parameters, respectively. For resilience evaluation, we employ CNN for EMNIST, featherweight ResNet-18 for Fashion-MNIST, lightweight ResNet-18 for CIFAR-10, and ResNet-18 for Tiny-ImageNet. A summary of the datasets and models can be found in Table~\ref{data_model}.

\begin{table*}[t]
    \footnotesize
	\centering
	
	\caption{The communication overhead (GB) of private vector-wise filtering in the setup phase.}
	\label{tab:setup_communication}
    \renewcommand\arraystretch{1}
	\begin{tabular}{c|c|c|c|c|c|c|c|c|c}
		\hline
	
            \multirow{2}{*}{Model} &\multirow{2}{*}{\#clients}  &\multicolumn{4}{c|}{Baseline} &\multicolumn{4}{c}{ABBR}\\
  
            &\multirow{2}{*}{~} &Multi-Krum  &FoolsGold &FABA &FLAME &Multi-Krum &FoolsGold  &FABA & FLAME \\
		
		\hline
		\multirow{6}{*}{\makecell{CNN \\ (\#params $\approx 266$k)}} 
		& 4   & 4.56  & 7.61  & 12    & 7.61  & \textbf{0.01} & \textbf{0.03} & \textbf{0.05} & \textbf{0.03} \\
 & 8   & 21    & 27    & 49    & 27    & \textbf{0.11} & \textbf{0.16} & \textbf{0.27} & \textbf{0.15} \\
 & 16  & 91    & 104   & 195   & 104   & \textbf{0.65} & \textbf{0.78} & \textbf{1.41} & \textbf{0.74} \\
 & 32  & 378   & 402   & 779   & 402   & \textbf{3.46} & \textbf{3.88} & \textbf{6.93} & \textbf{3.54} \\
 & 64  & 1535  & 1586  & 3118  & 1583  & \textbf{17}   & \textbf{19}   & \textbf{33}   & \textbf{17}   \\
 & 128 & 6194  & 6307  & 12473 & 6285  & \textbf{85}   & \textbf{99}   & \textbf{154}  & \textbf{77}   \\
		\hline
		\multirow{6}{*}{\makecell{ResNet-18 Light \\ (\#params $\approx 2.8$M)}} 
		& 4   & 48    & 80    & 128   & 80    & \textbf{0.01} & \textbf{0.03} & \textbf{0.05} & \textbf{0.03} \\
 & 8   & 225   & 289   & 512   & 289   & \textbf{0.11} & \textbf{0.16} & \textbf{0.27} & \textbf{0.15} \\
 & 16  & 960   & 1090  & 2051  & 1090  & \textbf{0.65} & \textbf{0.78} & \textbf{1.41} & \textbf{0.74} \\
 & 32  & 3970  & 4233  & 8209  & 4233  & \textbf{3.46} & \textbf{3.88} & \textbf{6.93} & \textbf{3.54} \\
 & 64  & 16068 & 16680 & 32839 & 16677 & \textbf{17}   & \textbf{19}   & \textbf{33}   & \textbf{17}   \\
 & 128 & -     & -     & -     & -     & \textbf{85}   & \textbf{99}   & \textbf{154}  & \textbf{77}   \\
		\hline
		\multirow{6}{*}{\makecell{ResNet-18 \\ (\#params $\approx 11$M)}} 
		& 4   & 194   & 323   & 517   & 323   & \textbf{0.01} & \textbf{0.03} & \textbf{0.05} & \textbf{0.03} \\
 & 8   & 901   & 1163  & 2066  & 1163  & \textbf{0.11} & \textbf{0.16} & \textbf{0.27} & \textbf{0.15} \\
 & 16  & 3876  & 4389  & 8267  & 4389  & \textbf{0.65} & \textbf{0.78} & \textbf{1.41} & \textbf{0.74} \\
 & 32  & 15974 & 17039 & 33071 & 17039 & \textbf{3.46} & \textbf{3.88} & \textbf{6.93} & \textbf{3.54} \\
 & 64  & -     & -     & -     & -     & \textbf{17}   & \textbf{19}   & \textbf{33}   & \textbf{17}   \\
 & 128 & -     & -     & -     & -     & \textbf{85}   & \textbf{99}   & \textbf{154}  & \textbf{77} \\
		\hline
	\end{tabular}
	
\end{table*}

\begin{table*}[t]
    \footnotesize
	\centering
	
	\caption{The communication overhead (GB) of private vector-wise filtering in the online phase.}
	\label{tab:online_communication}
    \renewcommand\arraystretch{1}
	\begin{tabular}{c|c|c|c|c|c|c|c|c|c}
		\hline
		
            \multirow{2}{*}{Model} &\multirow{2}{*}{\#clients}  &\multicolumn{4}{c|}{Baseline} &\multicolumn{4}{c}{ABBR}\\
  
            &\multirow{2}{*}{~} &Multi-Krum  &FoolsGold &FABA &FLAME &Multi-Krum &FoolsGold  &FABA & FLAME \\
		
		\hline
		\multirow{6}{*}{\makecell{CNN \\ (\#params $\approx 266$k)}} 
		& 4   & 0.04 & 0.08 & 0.13 & 0.08 & \textbf{0.003} & \textbf{0.003} & \textbf{0.0006} & \textbf{0.003} \\
 & 8   & 0.22 & 0.28 & 0.51 & 0.28 & \textbf{0.005} & \textbf{0.005} & \textbf{0.003}  & \textbf{0.005} \\
 & 16  & 0.95 & 1.08 & 2.03 & 1.07 & \textbf{0.01}  & \textbf{0.01}  & \textbf{0.02}   & \textbf{0.01}  \\
 & 32  & 3.94 & 4.2  & 8.13 & 4.19 & \textbf{0.04}  & \textbf{0.05}  & \textbf{0.08}   & \textbf{0.04}  \\
 & 64  & 16   & 17   & 32   & 17   & \textbf{0.19}  & \textbf{0.25}  & \textbf{0.36}   & \textbf{0.19}  \\
 & 128 & 65   & 66   & 130  & 66   & \textbf{0.98}  & \textbf{1.25}  & \textbf{1.66}   & \textbf{0.9}   \\
		\hline
		\multirow{6}{*}{\makecell{ResNet-18 Light \\ (\#params $\approx 2.8$M)}} 
		& 4   & 0.54 & 0.87 & 1.34 & 0.87 & \textbf{0.04}  & \textbf{0.04}  & \textbf{0.0006} & \textbf{0.04}  \\
 & 8   & 2.37 & 3.04 & 5.35 & 3.04 & \textbf{0.04}  & \textbf{0.04}  & \textbf{0.003}  & \textbf{0.04}  \\
 & 16  & 10   & 11   & 21   & 11   & \textbf{0.05}  & \textbf{0.05}  & \textbf{0.02}   & \textbf{0.05}  \\
 & 32  & 41   & 44   & 86   & 44   & \textbf{0.08}  & \textbf{0.09}  & \textbf{0.08}   & \textbf{0.08}  \\
 & 64  & 168  & 174  & 342  & 174  & \textbf{0.23}  & \textbf{0.28}  & \textbf{0.36}   & \textbf{0.23}  \\
 & 128 & -    & -    & -    & -    & \textbf{1.02}  & \textbf{1.29}  & \textbf{1.66}   & \textbf{0.93}  \\
		\hline
		\multirow{6}{*}{\makecell{ResNet-18 \\ (\#params $\approx 11$M)}} 
		& 4   & 2.17 & 3.52 & 5.38 & 3.52 & \textbf{0.16}  & \textbf{0.16}  & \textbf{0.0006} & \textbf{0.16}  \\
 & 8   & 9.58 & 12   & 22   & 12   & \textbf{0.16}  & \textbf{0.16}  & \textbf{0.003}  & \textbf{0.16}  \\
 & 16  & 40   & 46   & 86   & 46   & \textbf{0.17}  & \textbf{0.17}  & \textbf{0.02}   & \textbf{0.17}  \\
 & 32  & 167  & 178  & 345  & 178  & \textbf{0.2}   & \textbf{0.21}  & \textbf{0.08}   & \textbf{0.2}   \\
 & 64  & -    & -    & -    & -    & \textbf{0.35}  & \textbf{0.4}   & \textbf{0.36}   & \textbf{0.35}  \\
 & 128 & -    & -    & -    & -    & \textbf{1.14}  & \textbf{1.41}  & \textbf{1.66}   & \textbf{1.05} \\
		\hline
	\end{tabular}
	
\end{table*}

\begin{table*}[t]
    \footnotesize
	\centering
	
	\caption{Robustness against various attacks on different datasets. The Metrics are Backdoor Accuracy (BA) and Main Task Accuracy (MA), and all values are percentages. Label Flipping and Gaussian attacks are untargeted and thus have no BA metric.}
	\label{tab:robustness}
    \renewcommand\arraystretch{1}
    \setlength\tabcolsep{5pt}
	\begin{tabular}{c|c|c c|c c|c c|c c|c c|c c|c c|c c}
		\hline
		\multirow{3}{*}{Attack} &\multirow{3}{*}{Dataset}  &\multicolumn{8}{c|}{Baseline} &\multicolumn{8}{c}{\textbf{ABBR}}\\
  
        &\multirow{3}{*}{~} &\multicolumn{2}{c|}{ Multi-Krum}  &\multicolumn{2}{c|}{ FoolsGold} &\multicolumn{2}{c|}{FABA} &\multicolumn{2}{c|}{FLAME} &\multicolumn{2}{c|}{Multi-Krum}   &\multicolumn{2}{c|}{FoolsGold}  &\multicolumn{2}{c|}{FABA}  & \multicolumn{2}{c}{ FLAME} \\
        \cline{3-18}
        &\multirow{3}{*}{~} &BA &MA &BA &MA &BA &MA &BA &MA &BA &MA &BA &MA &BA &MA &BA &MA\\
		\hline
		\multirow{4}{*}{\makecell{Benign \\ Setting}} 
		&CIFAR-10 &- &78.2 &- &80.6 &- &77.6 &- &77.9 &- &78 &- &80.3 &- &77.2 &- &75.4  \\
		& EMNIST        & - & 98.6 & - & 98.7 & - & 98.4 & - & 98.5 & - & 98.3 & - & 98.7 & - & 98.4 & - & 98.4 \\
 & Fashion-MNIST & - & 90.8 & - & 92.2 & - & 90.9 & - & 91.5 & - & 91   & - & 92.2 & - & 90.9 & - & 90.9 \\
 & Tiny-ImageNet & - & 47.4 & - & 54.6 & - & 47.4 & - & 49.5 & - & 47.2 & - & 54.6 & - & 47.1 & - & 49.2\\
		\hline
		\multirow{4}{*}{\makecell{Constrain \\ and \\ Scale}} 
		& CIFAR-10      & 2.1 & 80.1 & 92  & 74.7 & 2.2 & 80.2 & 2.1 & 79.7 & 2.2 & 80.2 & 2.7 & 79.1 & 2.4 & 80   & 1.9 & 80   \\
 & EMNIST        & 0.1 & 98.8 & 1.4 & 97.2 & 0.2 & 98.8 & 0.2 & 98.8 & 0.1 & 98.7 & 0.1 & 98.7 & 0.1 & 98.7 & 0.2 & 98.8 \\
 & Fashion-MNIST & 2.3 & 92.3 & 1.8 & 92.4 & 2.3 & 92.4 & 2.6 & 92.2 & 2.3 & 92.4 & 2.4 & 92.4 & 2.3 & 92.4 & 2.5 & 92.3 \\
		& Tiny-ImageNet & 0.1 & 57.3 & 45.4 & 40.7 & 0.1 & 57.3 & 0.1 & 57.4 & 0.1 & 57.4 & 0.3 & 57   & 0.1 & 57.4 & 0.1 & 57.2\\
		\hline
		\multirow{4}{*}{DBA} 
		& CIFAR-10      & 2.1 & 79.9 & 93   & 53.9 & 2.1 & 79.9 & 2.2 & 79.4 & 2   & 80   & 5.2 & 79.1 & 2.1 & 80   & 2.1 & 79.7 \\
 & EMNIST        & 0.1 & 98.8 & 99.5 & 91.1 & 0.2 & 98.7 & 0.2 & 98.7 & 0.1 & 98.7 & 0.3 & 98.7 & 0.1 & 98.7 & 0.1 & 98.7 \\
 & Fashion-MNIST & 2.6 & 92.1 & 96.6 & 81.8 & 2.3 & 92.4 & 2.3 & 92.3 & 2.3 & 92.3 & 3.8 & 92.2 & 2.3 & 92.4 & 2.5 & 92.2 \\
		& Tiny-ImageNet & 0.1 & 57   & 76 & 20.6 & 0.1 & 57   & 0.1 & 56.9 & 0.1 & 57   & 1.1 & 56.8 & 0.1 & 57   & 0.1 & 57\\
		\hline
        \multirow{1}{*}{Edge-Case} 
		& CIFAR-10      & 3.0   & 79.4 & 88.2 & 59.0   & 3.0   & 79.4 & 4.0   & 78.2 & 3.0   & 79.3 & 9.1 & 78.3 & 3.0   & 79.3 & 4.0   & 78.8 \\
		\hline
        \multirow{1}{*}{PGD} 
		& CIFAR-10      & 3.0   & 79.4 & 64.2 & 68.9 & 3.0   & 79.4 & 4.0   & 78.2 & 3.0   & 79.3 & 7.6 & 78.4 & 3.0   & 79.3 & 4.0   & 78.8 \\
		\hline
        \multirow{4}{*}{\makecell{Label \\ Flipping}} 
		& CIFAR-10      & - & 81.1 & - & 79.5 & - & 81.2 & - & 80.6 & - & 81.1 & - & 79.7 & - & 81.1 & - & 80.7 \\
 & EMNIST        & - & 98.8 & - & 98.2 & - & 98.9 & - & 98.9 & - & 98.9 & - & 98.2 & - & 98.9 & - & 98.9 \\
 & Fashion-MNIST & - & 92.4 & - & 90.3 & - & 92.4 & - & 92.5 & - & 92.3 & - & 90.1 & - & 92.3 & - & 92.3 \\
 & Tiny-ImageNet & - & 57.4 & - & 49.8 & - & 57.3 & - & 57.5 & - & 57.4 & - & 51.3 & - & 57.5 & - & 57.4\\
        \hline
        \multirow{4}{*}{\makecell{Gaussian \\ Attack}} 
		& CIFAR-10      & - & 81.3 & - & 79.3 & - & 81.3 & - & 81   & - & 81.2 & - & 79.7 & - & 81.2 & - & 81.1 \\
 & EMNIST        & - & 98.9 & - & 98.9 & - & 98.8 & - & 98.9 & - & 98.9 & - & 98.9 & - & 98.8 & - & 98.9 \\
 & Fashion-MNIST & - & 92.4 & - & 92.2 & - & 92.5 & - & 92.5 & - & 92.5 & - & 92.3 & - & 92.4 & - & 92.4 \\
 & Tiny-ImageNet & - & 57.6 & - & 57.9 & - & 57.5 & - & 57.5 & - & 57.6 & - & 57.9 & - & 57.6 & - & 57.4\\
        \hline
	\end{tabular}
	
\end{table*}

\begin{figure*}[t]
\centering
\subfloat[Multi-Krum]{\label{fig:multi_krum_various_byzantine_clients} \includegraphics[width=0.23\textwidth]{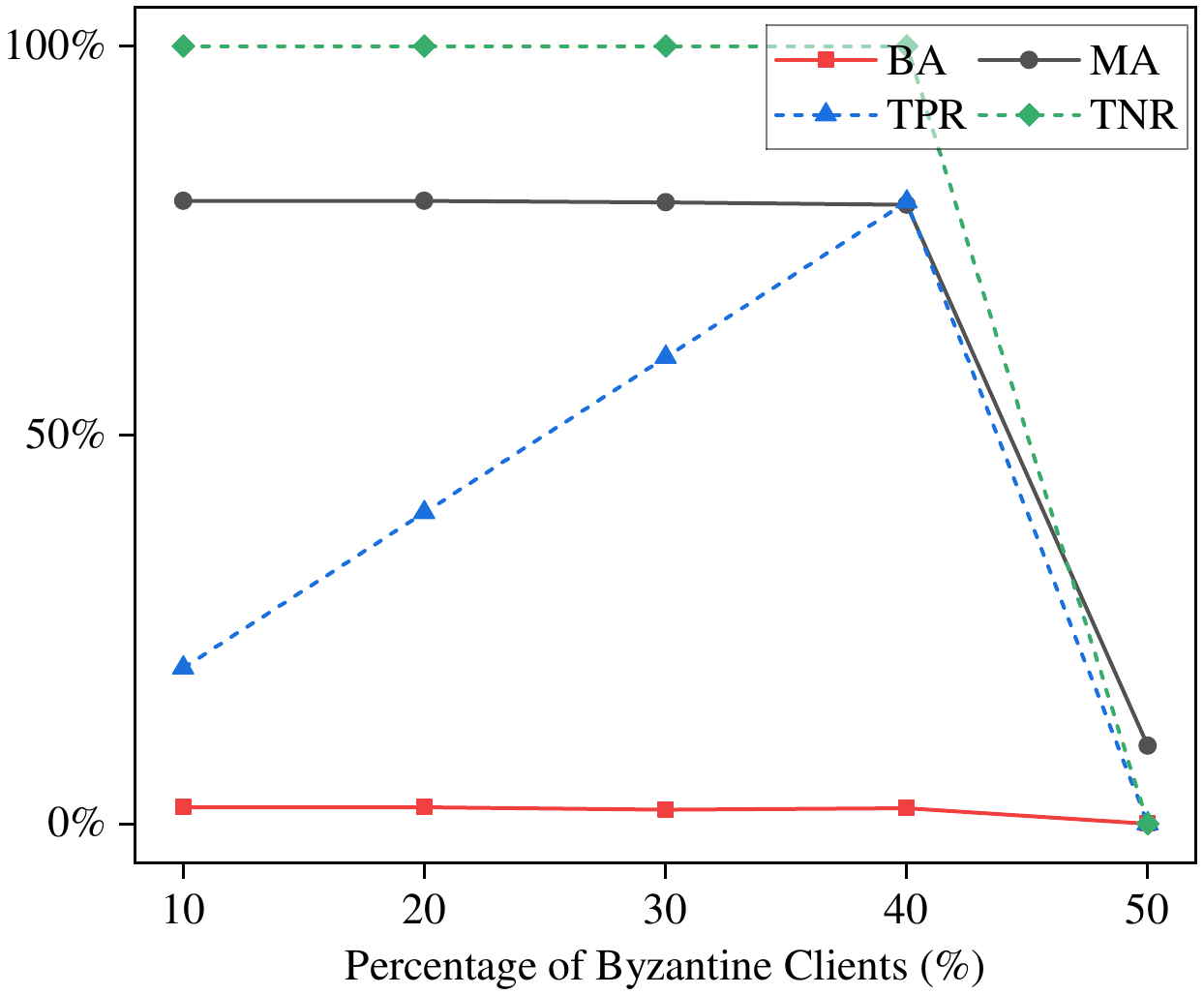}}
\subfloat[FoolsGold]{\label{fig:foolsgold_various_byzantine_clients} \includegraphics[width=0.23\textwidth]{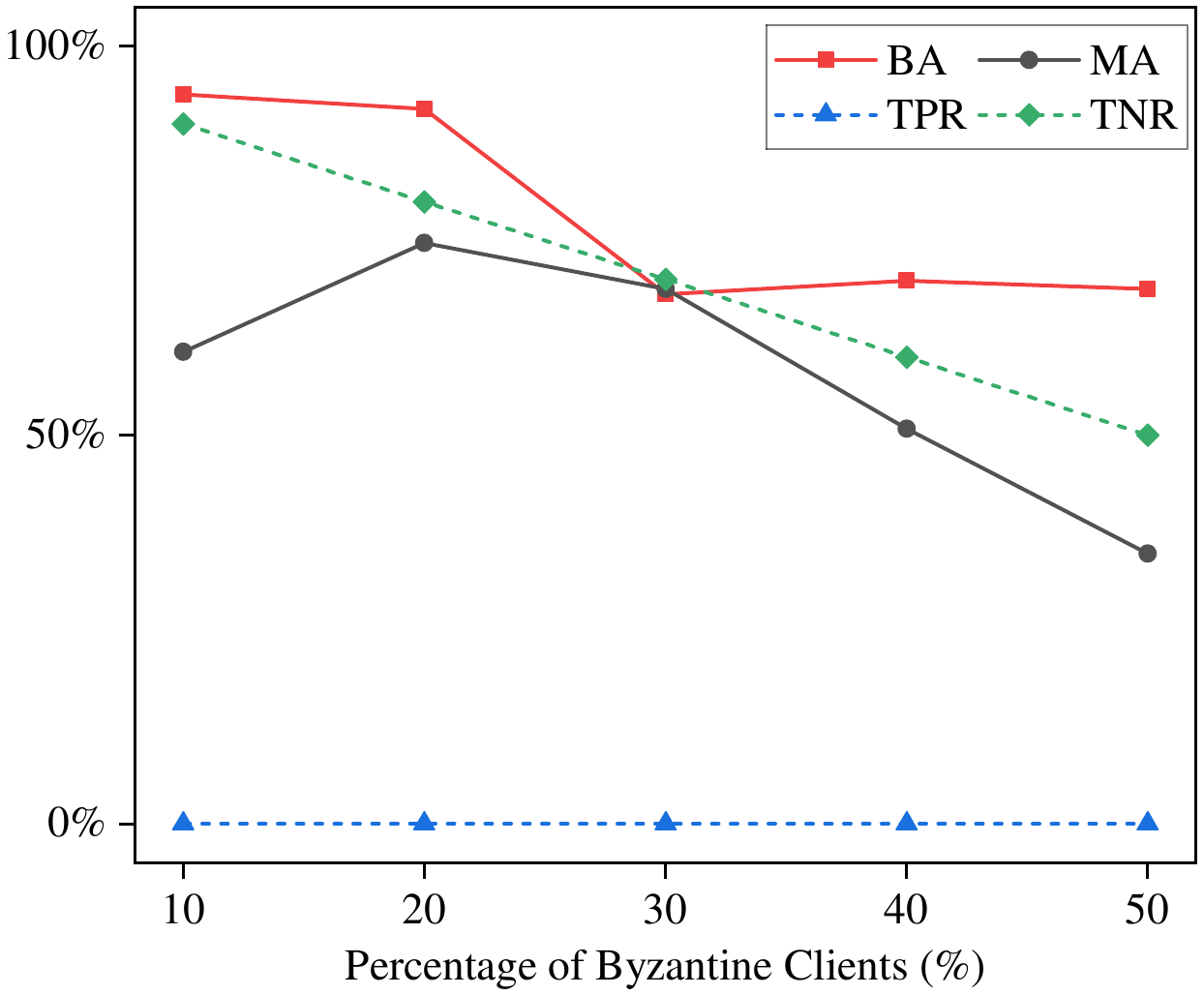}}
\subfloat[FABA]{\label{fig:faba_various_byzantine_clients} \includegraphics[width=0.23\textwidth]{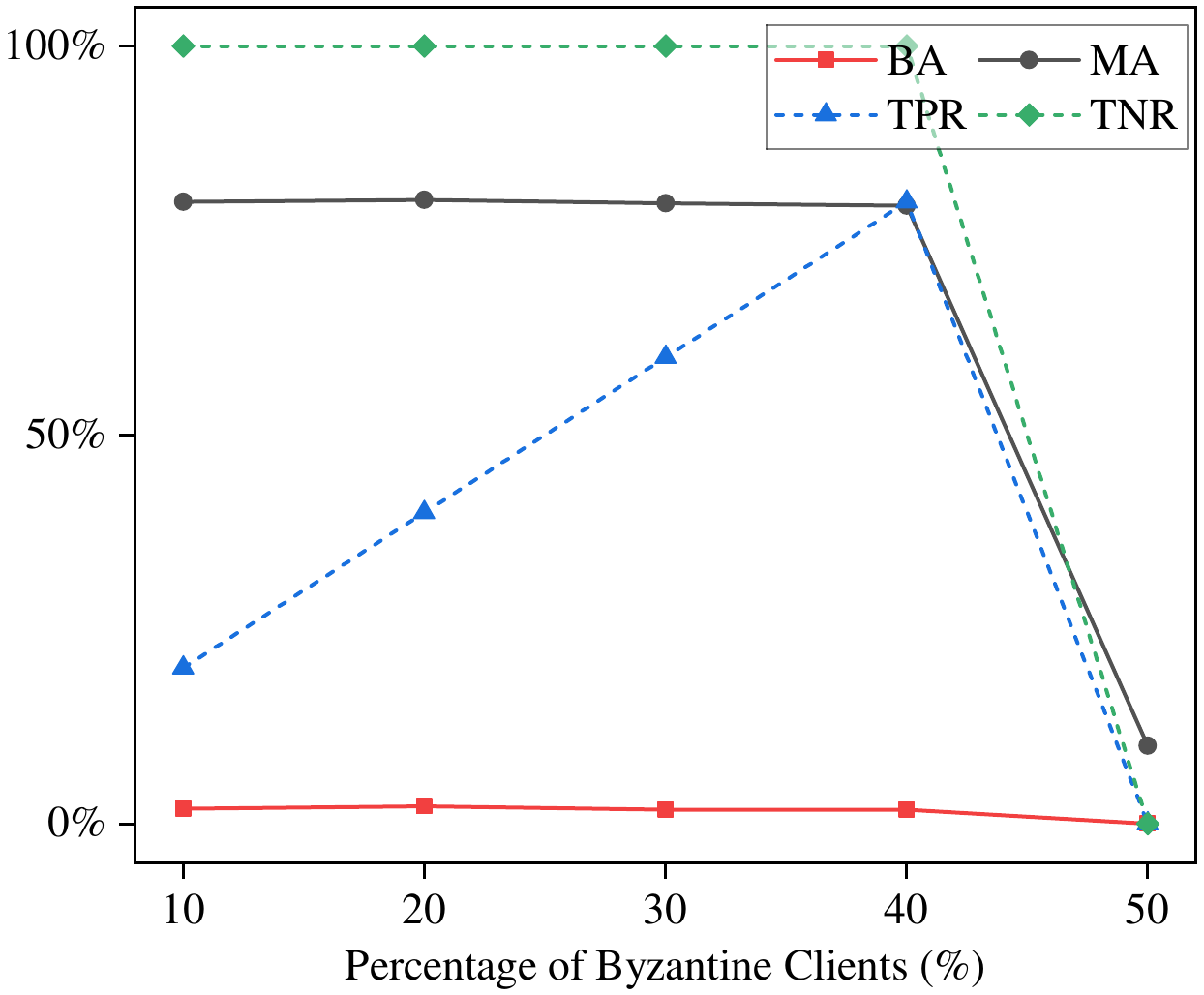}}
\subfloat[FLAME]{\label{fig:flame_various_byzantine_clients} \includegraphics[width=0.23\textwidth]{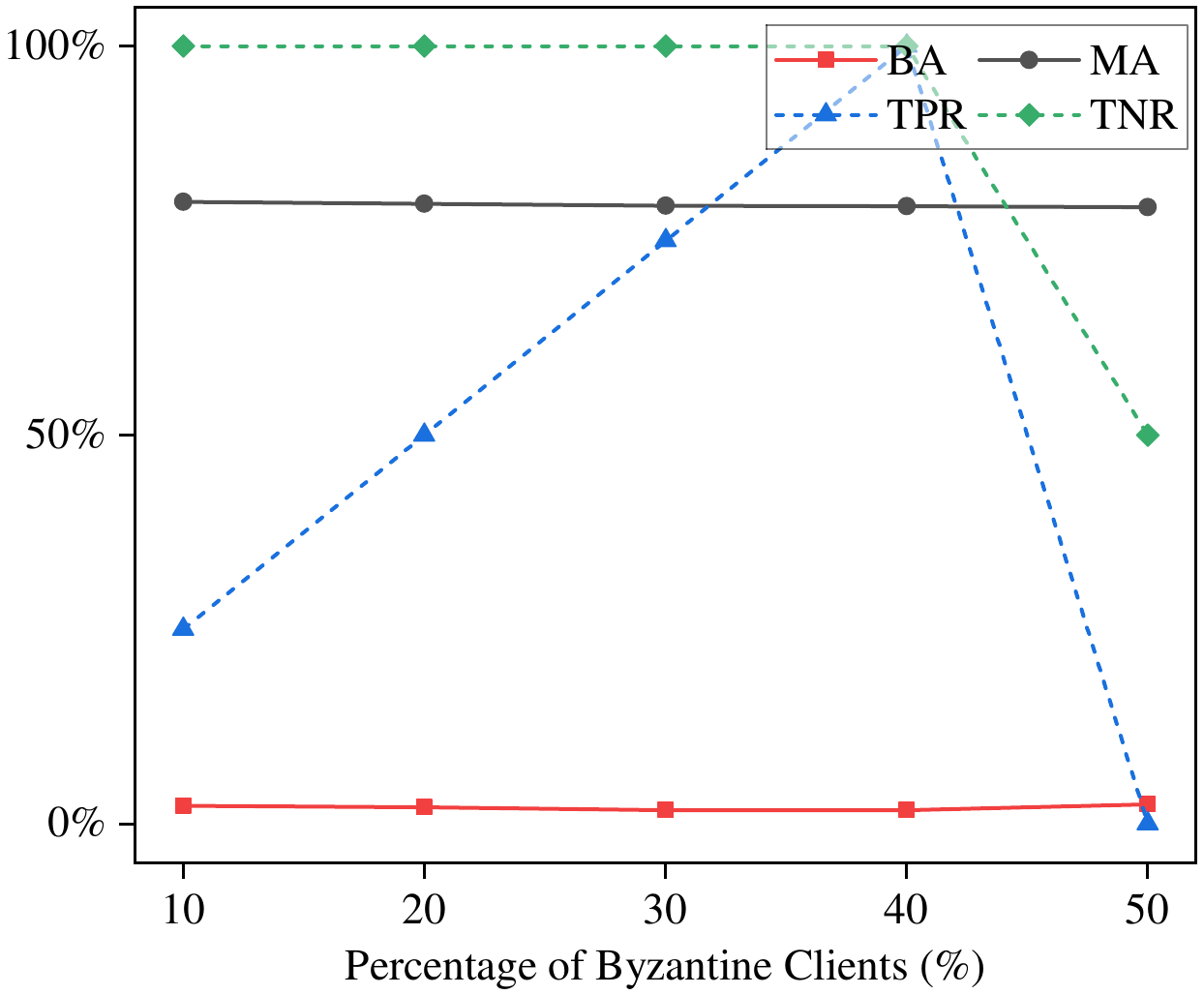}}\\

\subfloat[ABBR Multi-Krum]{\label{fig:abbr_multi_krum_various_byzantine_clients} \includegraphics[width=0.23\textwidth]{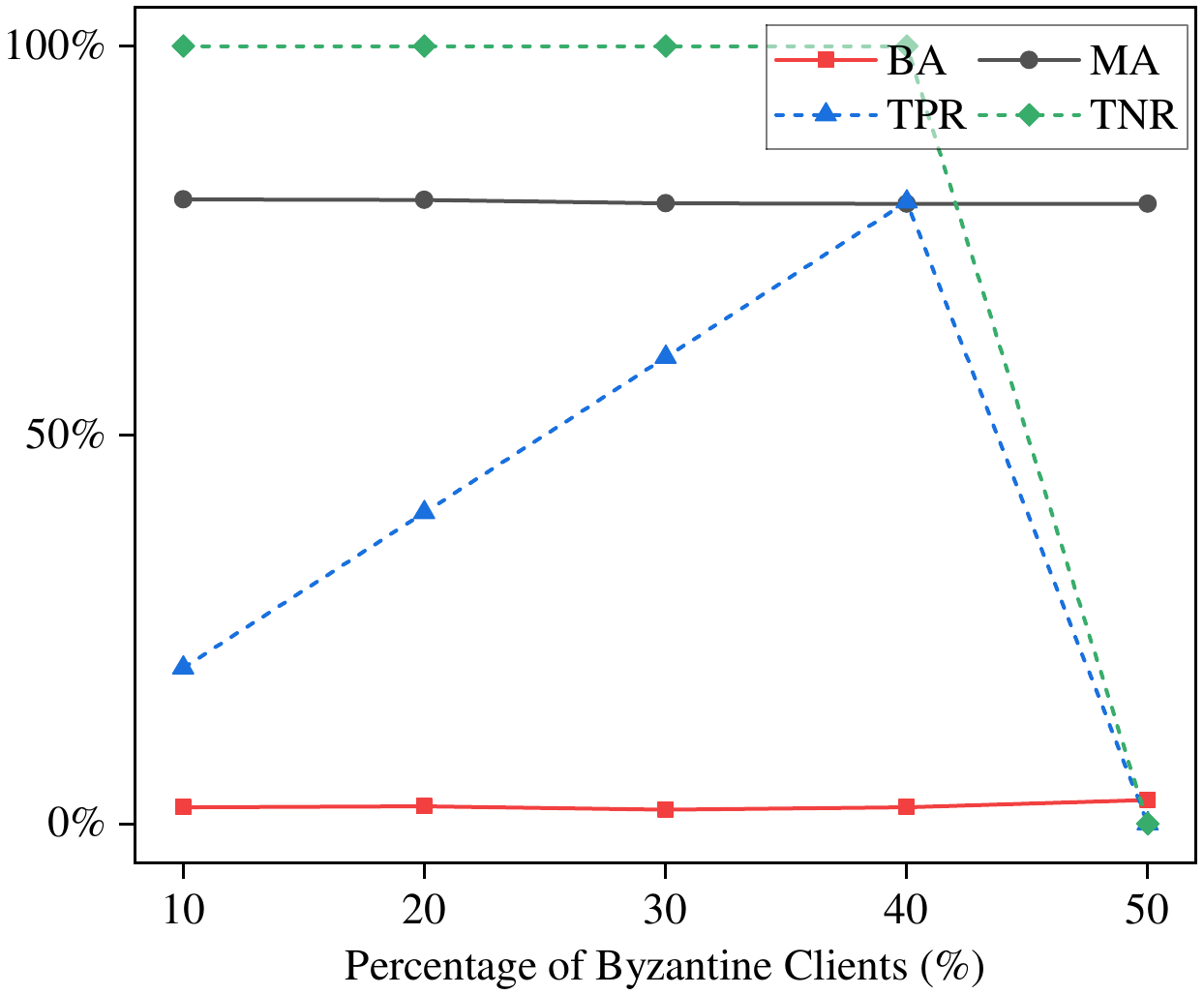}}
\subfloat[ABBR FoolsGold]{\label{fig:abbr_foolsgold_various_byzantine_clients} \includegraphics[width=0.23\textwidth]{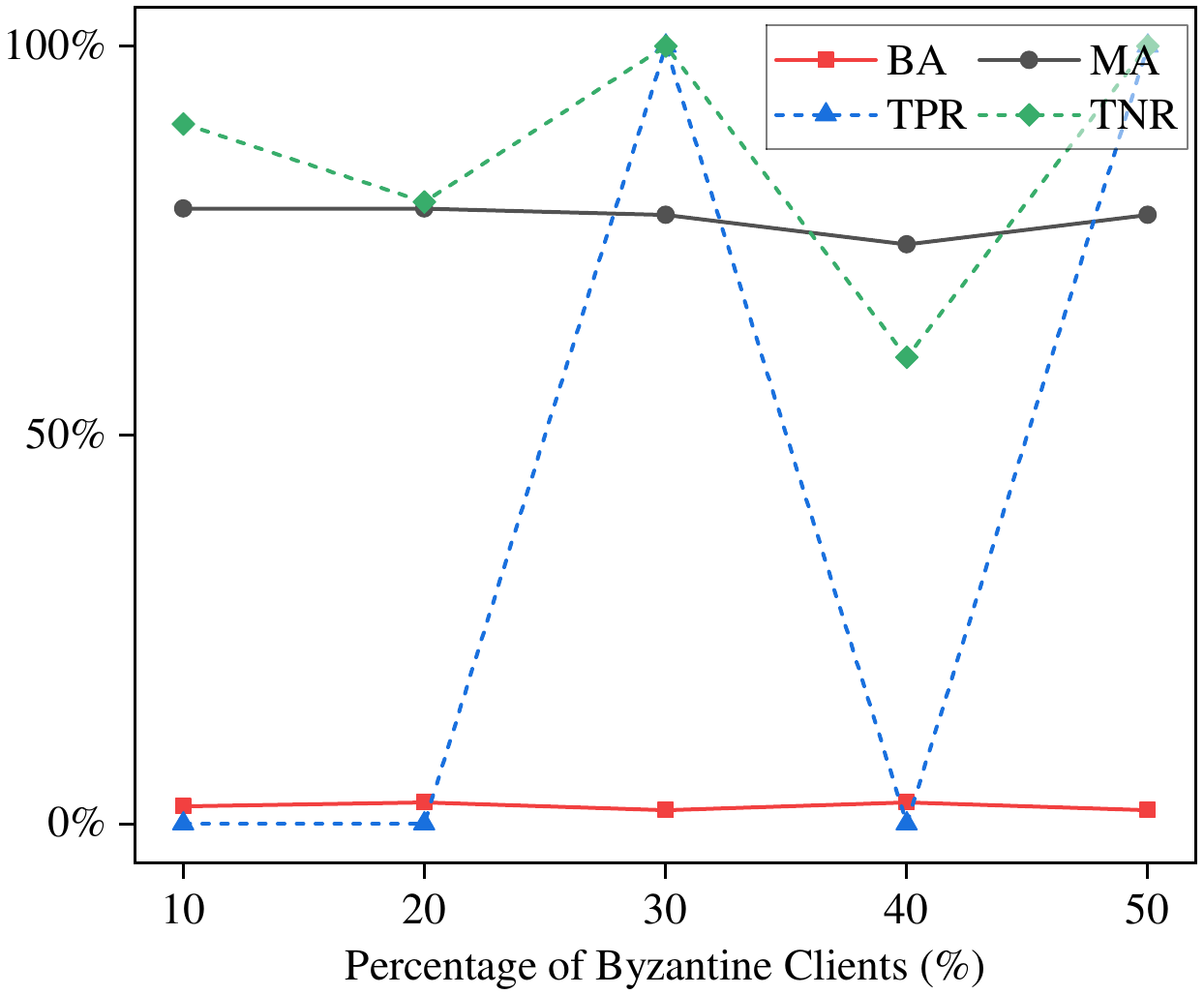}}
\subfloat[ABBR FABA]{\label{fig:abbr_faba_various_byzantine_clients} \includegraphics[width=0.23\textwidth]{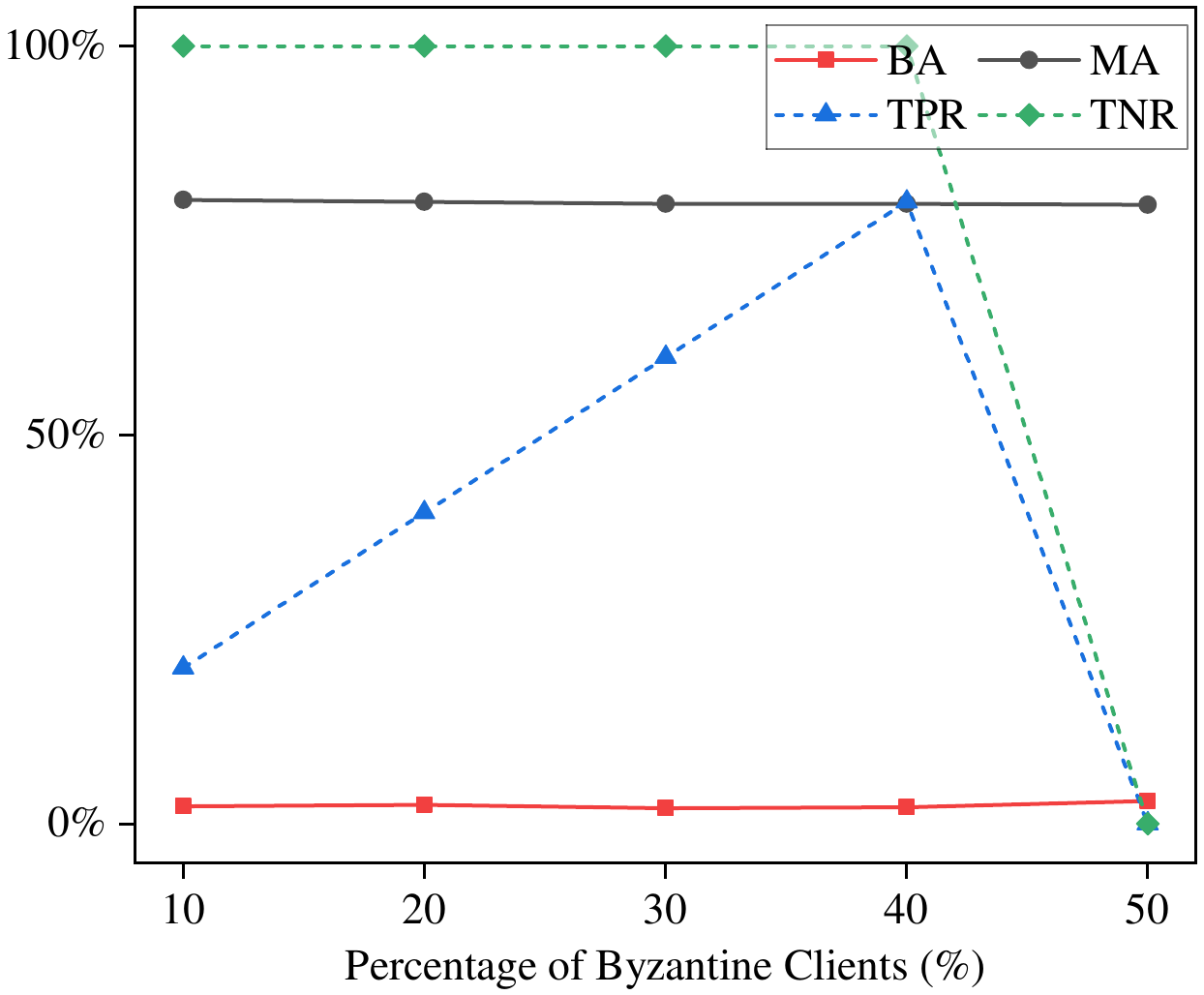}}
\subfloat[ABBR FLAME]{\label{fig:abbr_flame_various_byzantine_clients} \includegraphics[width=0.23\textwidth]{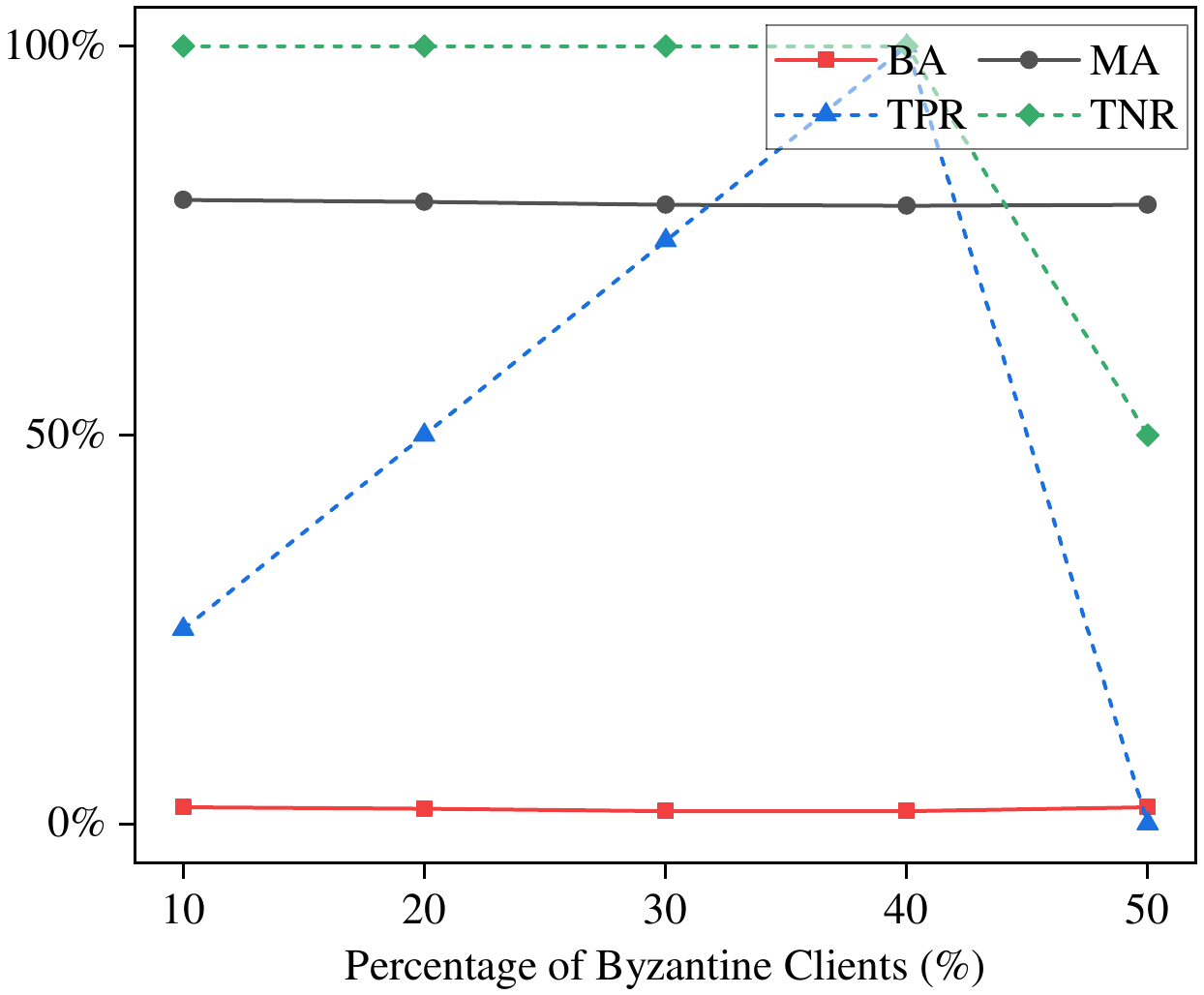}}\\
\caption{Impact of the percentage of Byzantine clients on CIFAR-10 dataset. (a)-(d): original vector-wise filterings, (e)-(f): ABBR versions of vector-wise filterings.}
\label{fig:bcr_impact}
\end{figure*}

\begin{figure*}[t]
\centering
\subfloat[Multi-Krum]{\label{fig:degree_noniid_multi_krum} \includegraphics[width=0.23\textwidth]{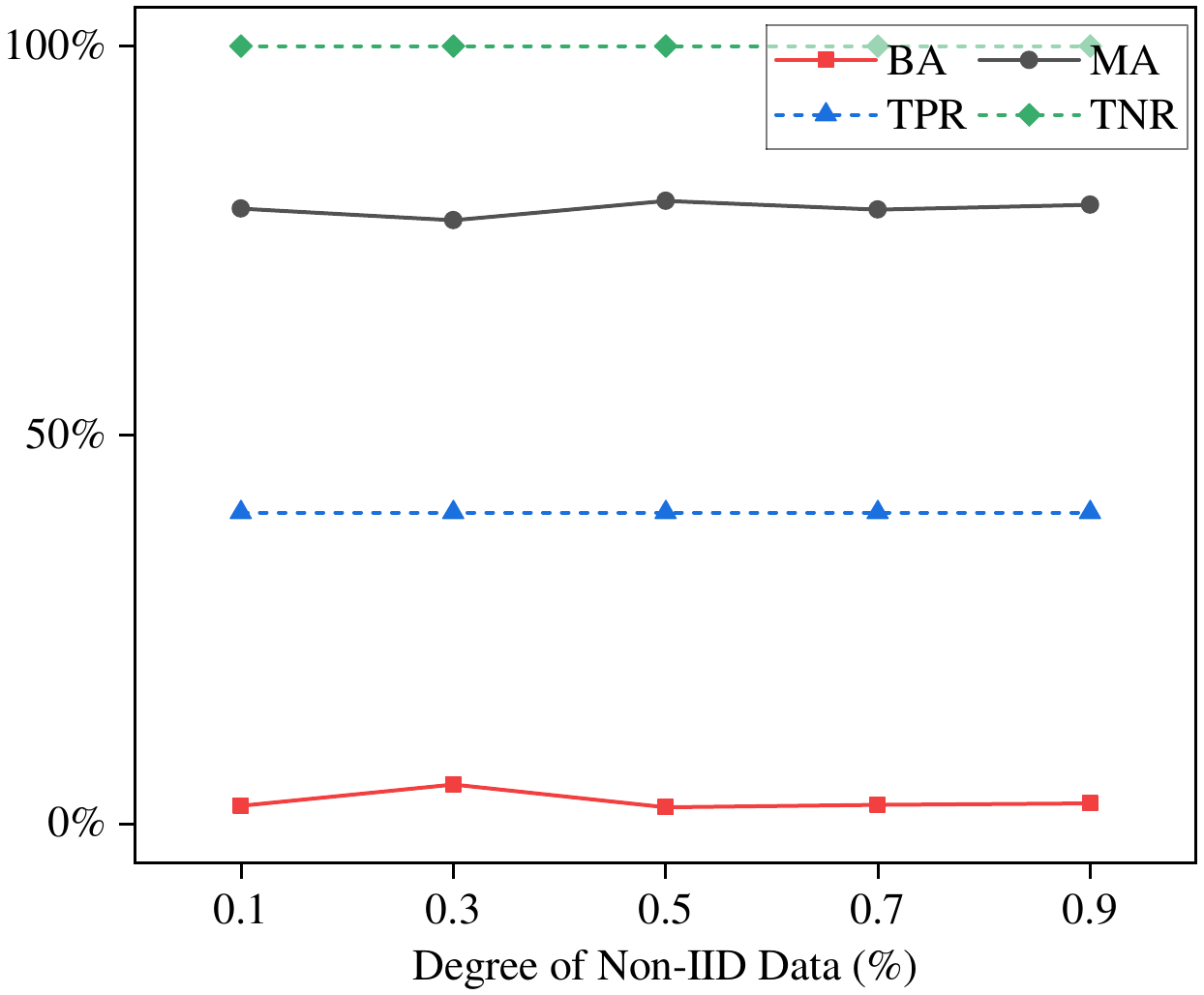}}
\subfloat[FoolsGold]{\label{fig:degree_noniid_foolsgold} \includegraphics[width=0.23\textwidth]{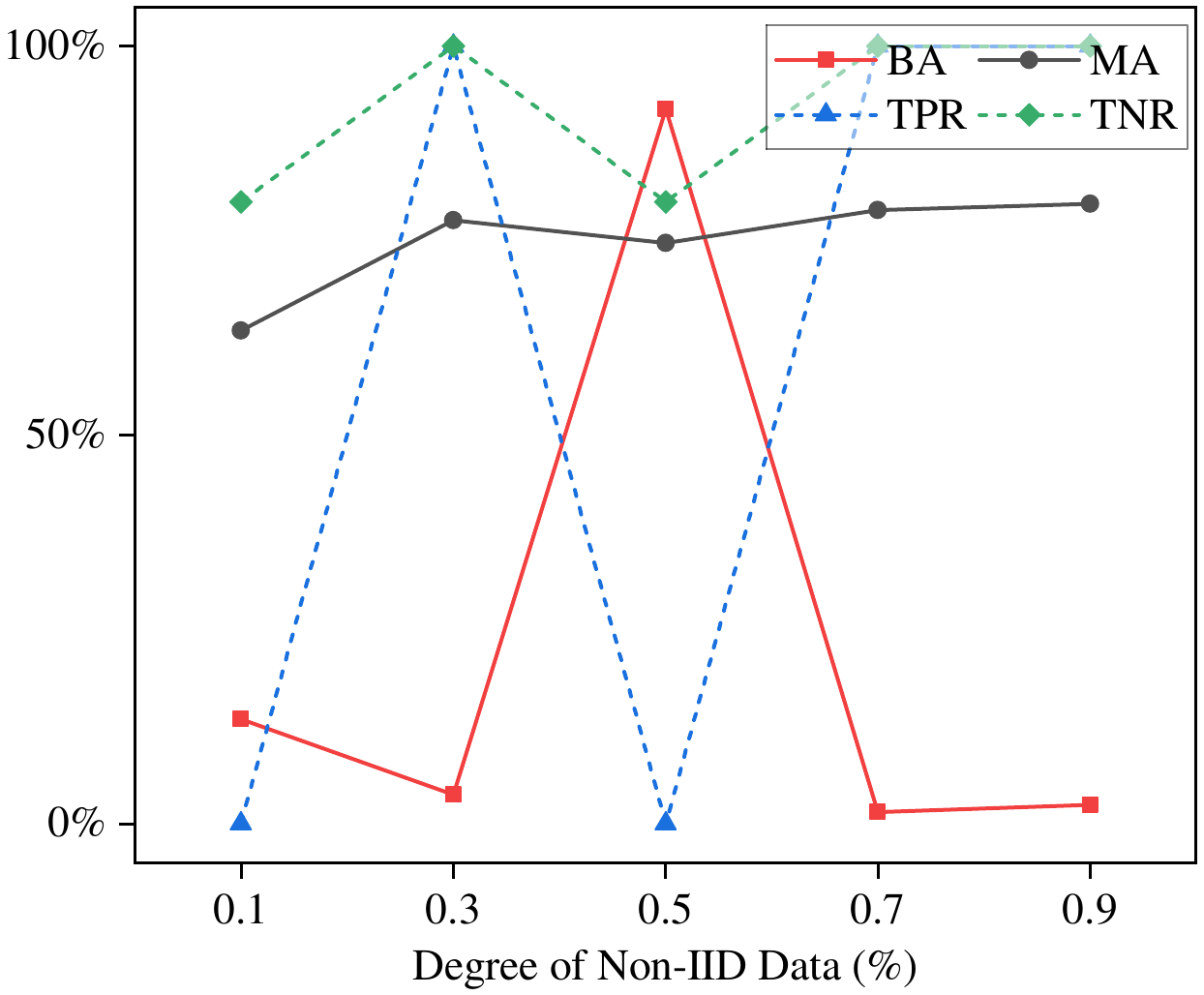}}
\subfloat[FABA]{\label{fig:degree_noniid_faba} \includegraphics[width=0.23\textwidth]{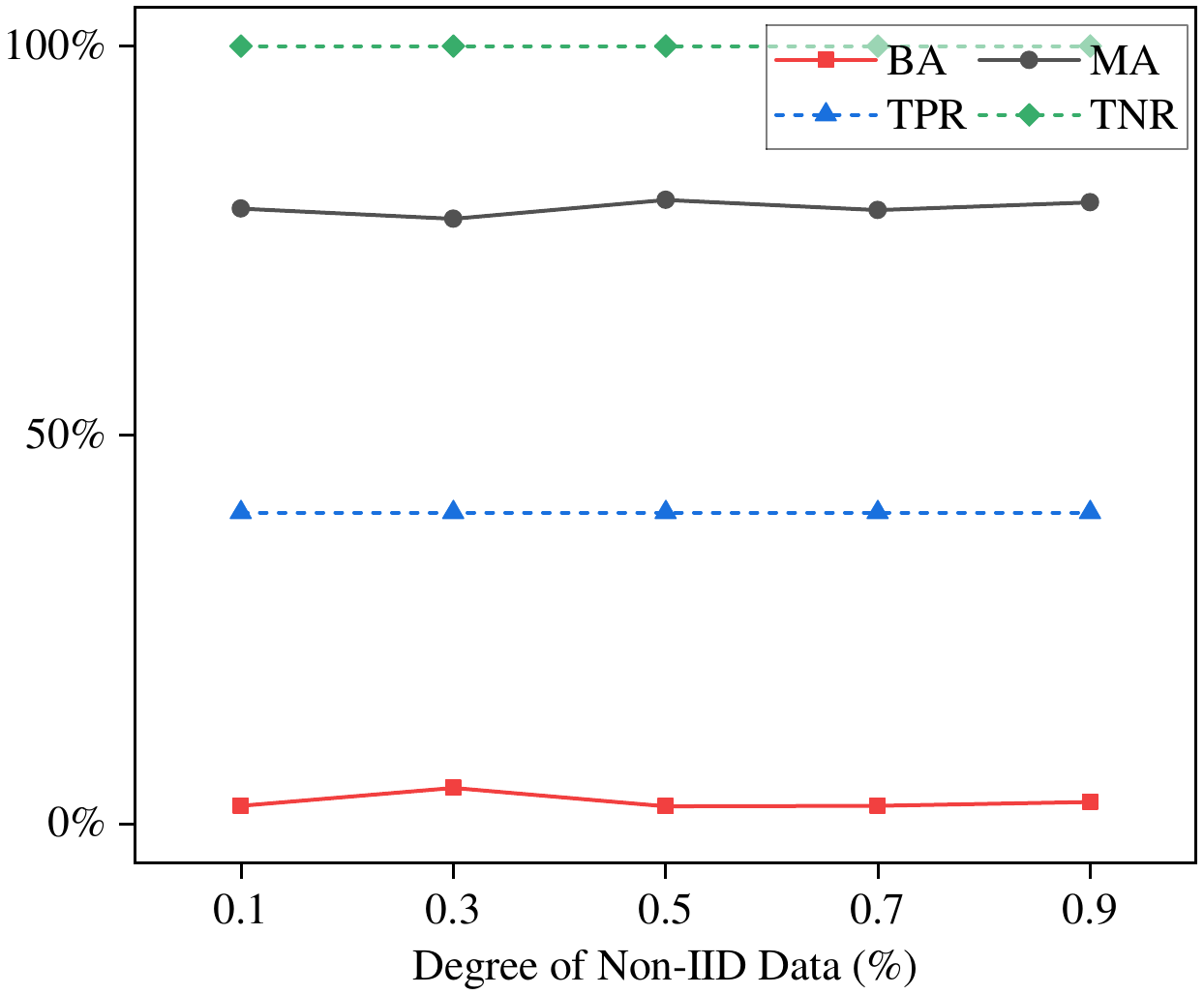}}
\subfloat[FLAME]{\label{fig:degree_noniid_flame} \includegraphics[width=0.23\textwidth]{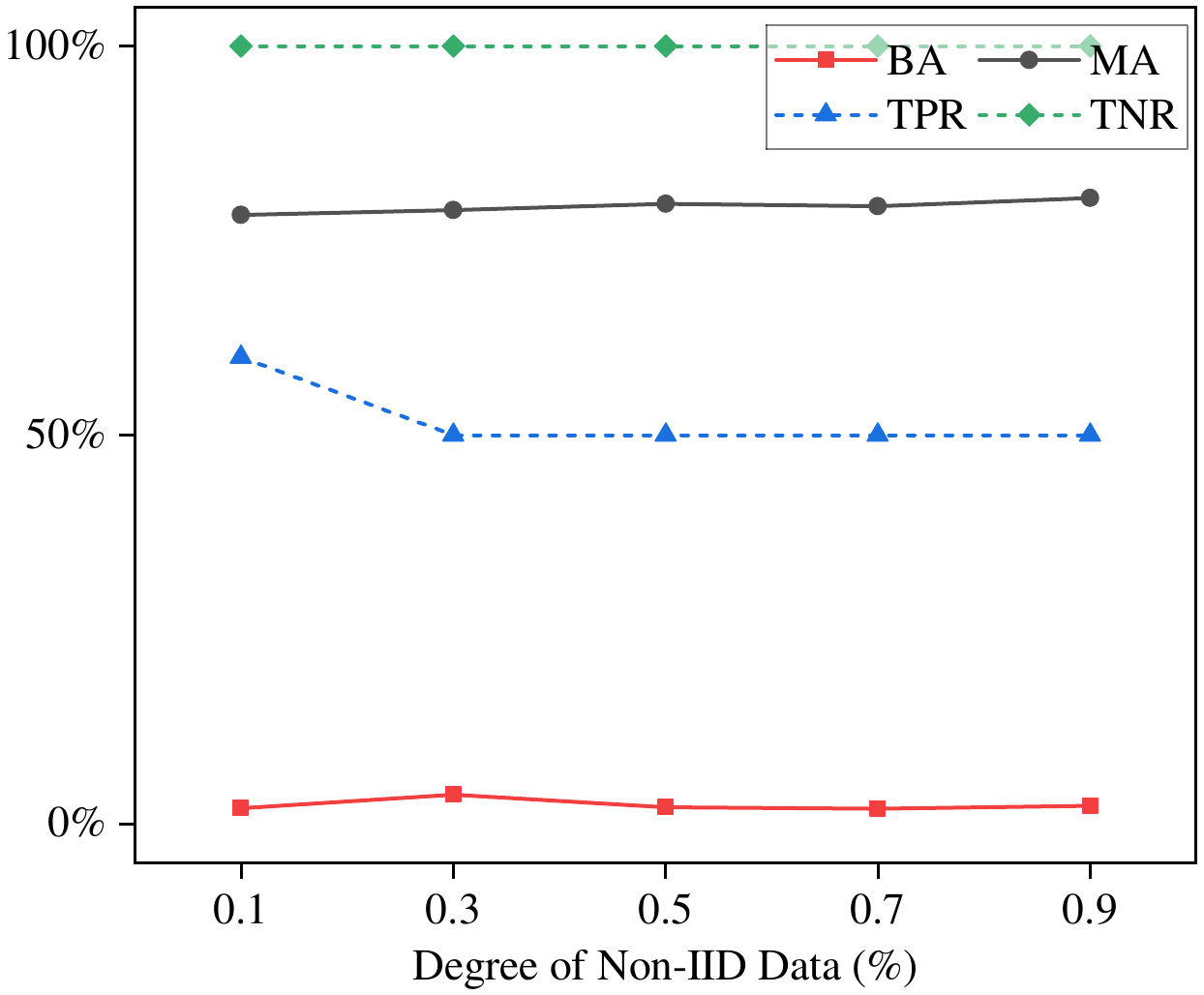}}\\

\subfloat[ABBR Multi-Krum]{\label{fig:degree_noniid_multi_krum_abbr} \includegraphics[width=0.23\textwidth]{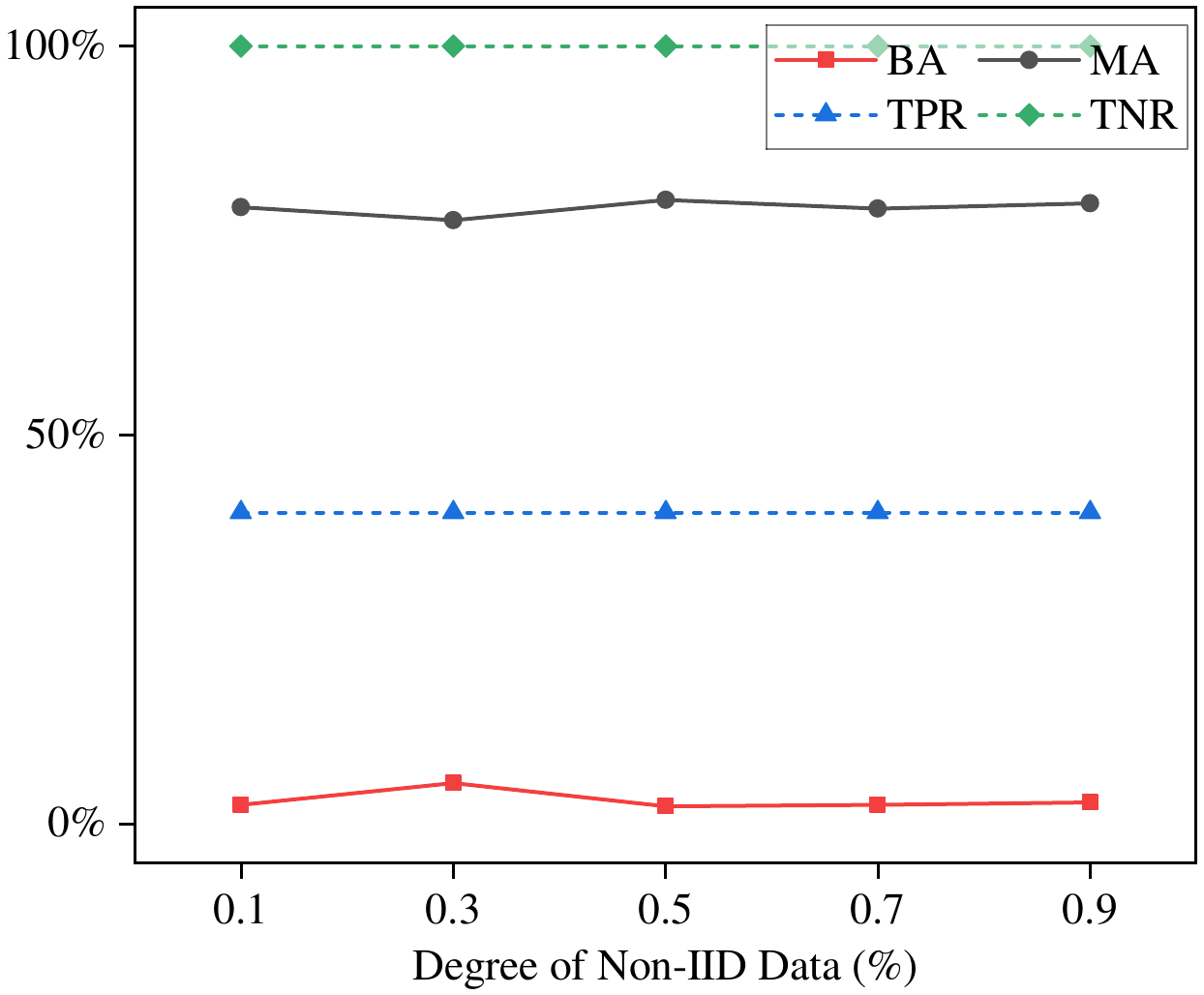}}
\subfloat[ABBR FoolsGold]{\label{fig:degree_noniid_foolsgold_abbr} \includegraphics[width=0.23\textwidth]{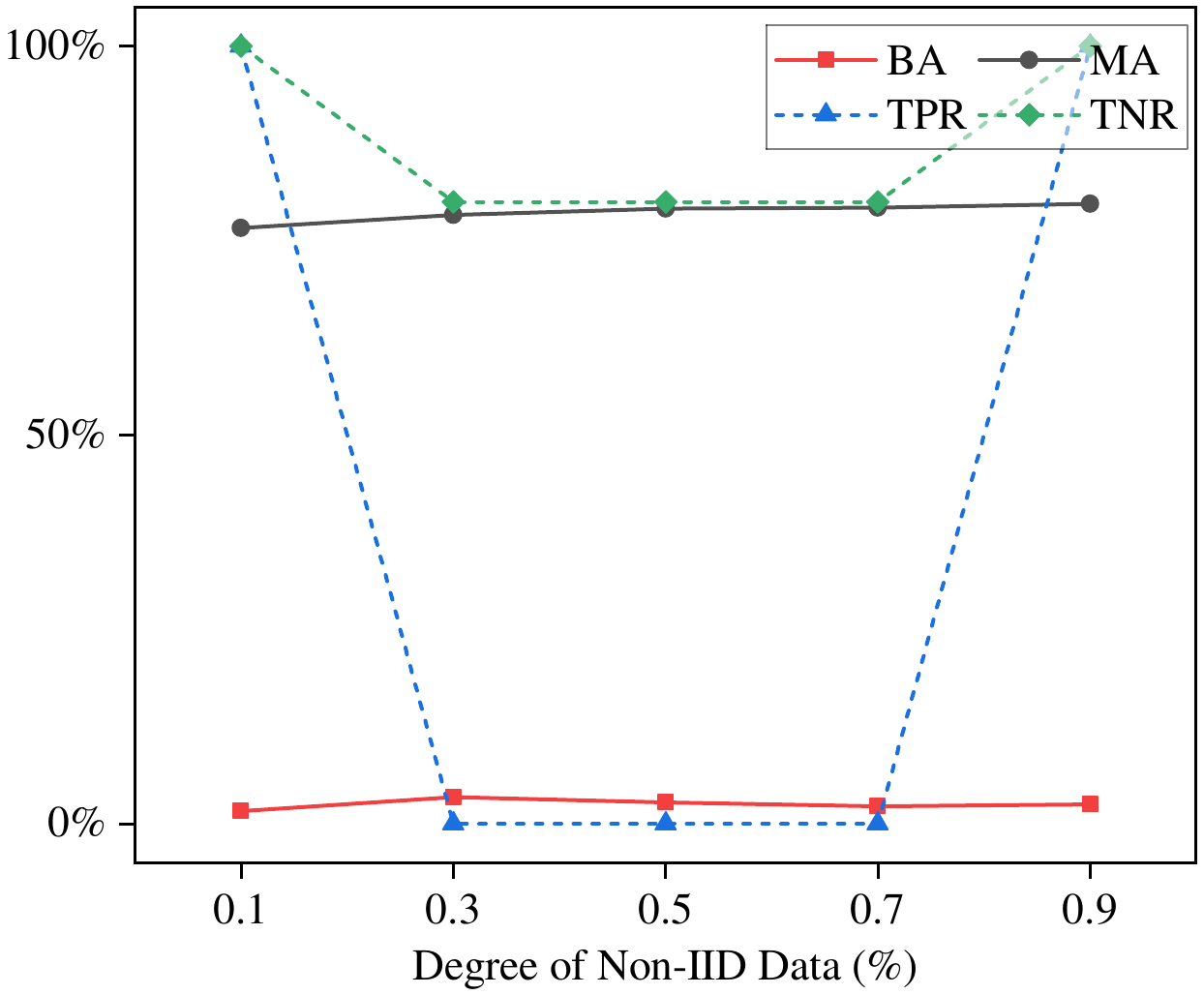}}
\subfloat[ABBR FABA]{\label{fig:degree_noniid_faba_abbr} \includegraphics[width=0.23\textwidth]{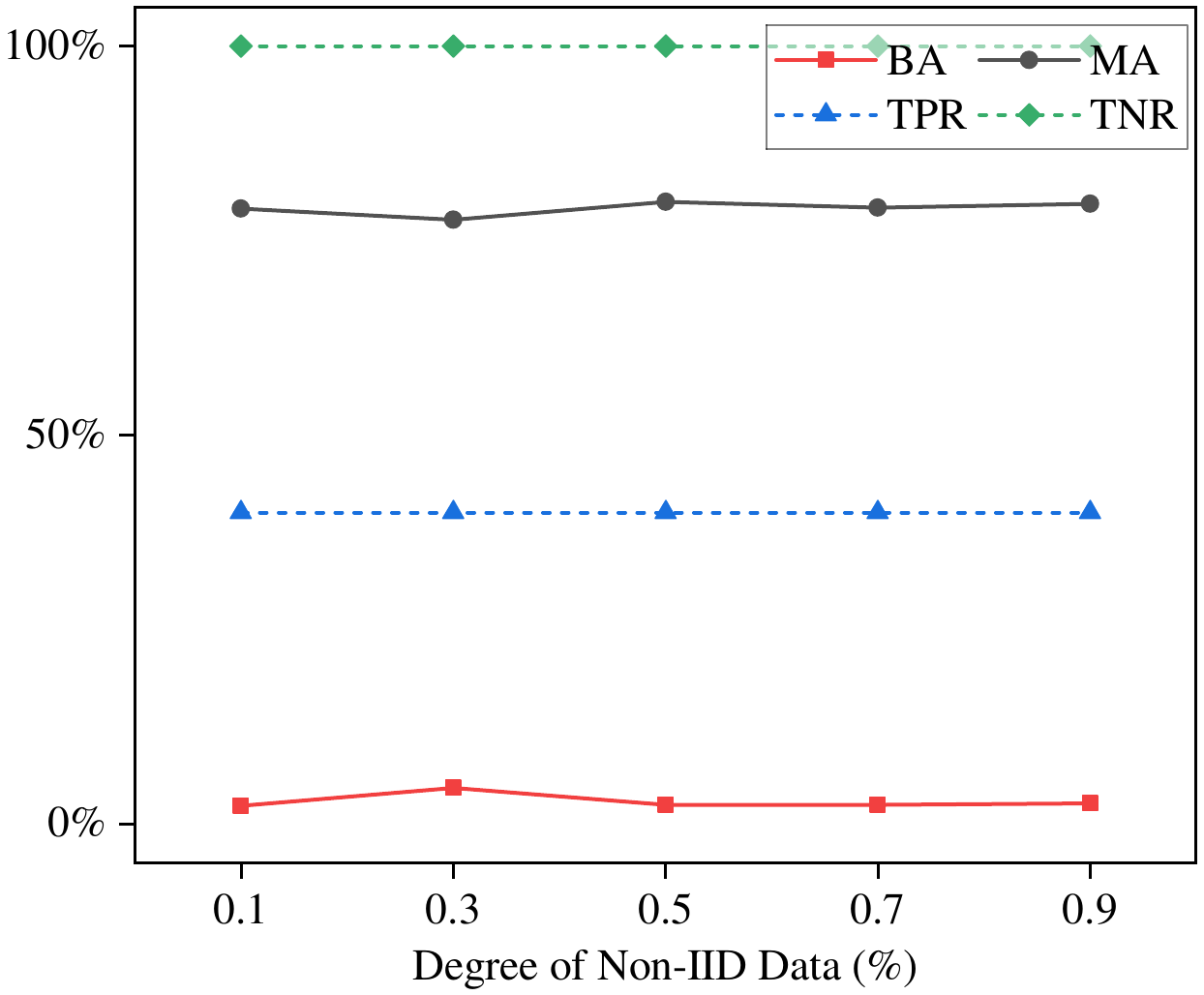}}
\subfloat[ABBR FLAME]{\label{fig:degree_noniid_flame_abbr} \includegraphics[width=0.23\textwidth]{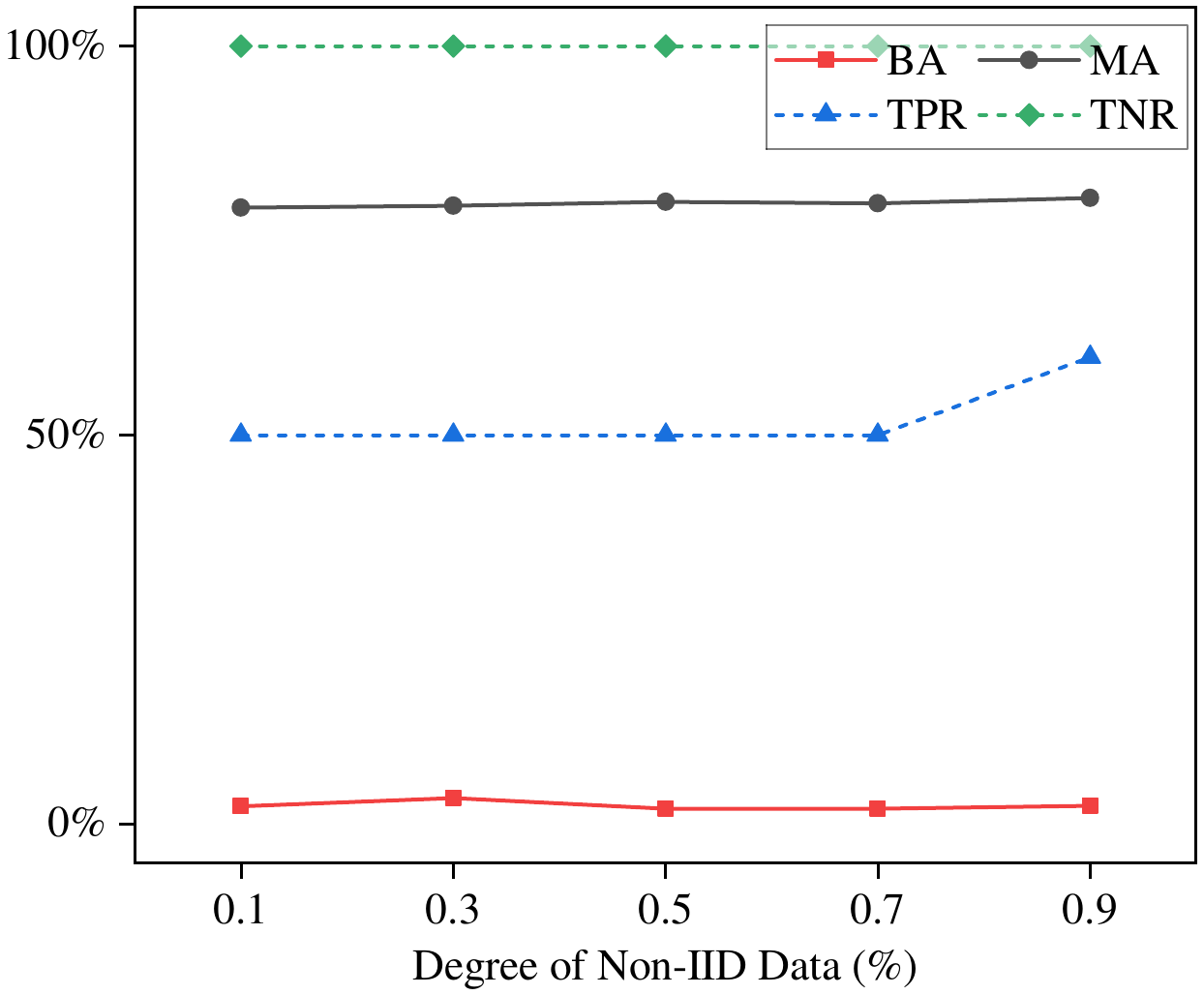}}\\
\caption{Impact of the degree of non-IID data (simulated by using Dirichlet distribution) on CIFAR-10 dataset. (a)-(d): original vector-wise filterings, (e)-(f): ABBR versions of vector-wise filterings.}
\label{fig:alpha_impact}
\end{figure*}

\begin{figure*}[t]
\centering

\subfloat[ABBR Multi-Krum]{\label{fig:abbr_multi_krum_k} \includegraphics[width=0.23\textwidth]{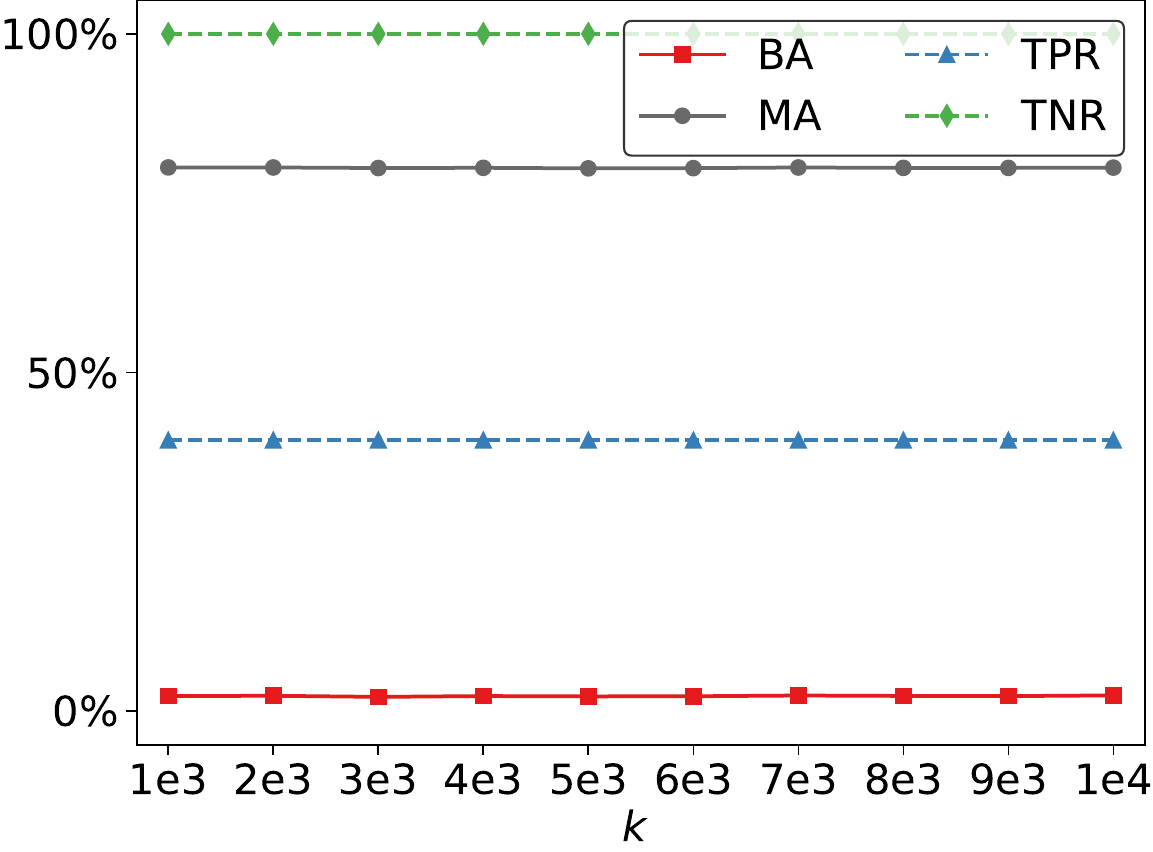}}
\subfloat[ABBR FoolsGold]{\label{fig:abbr_foolsgold_k} \includegraphics[width=0.23\textwidth]{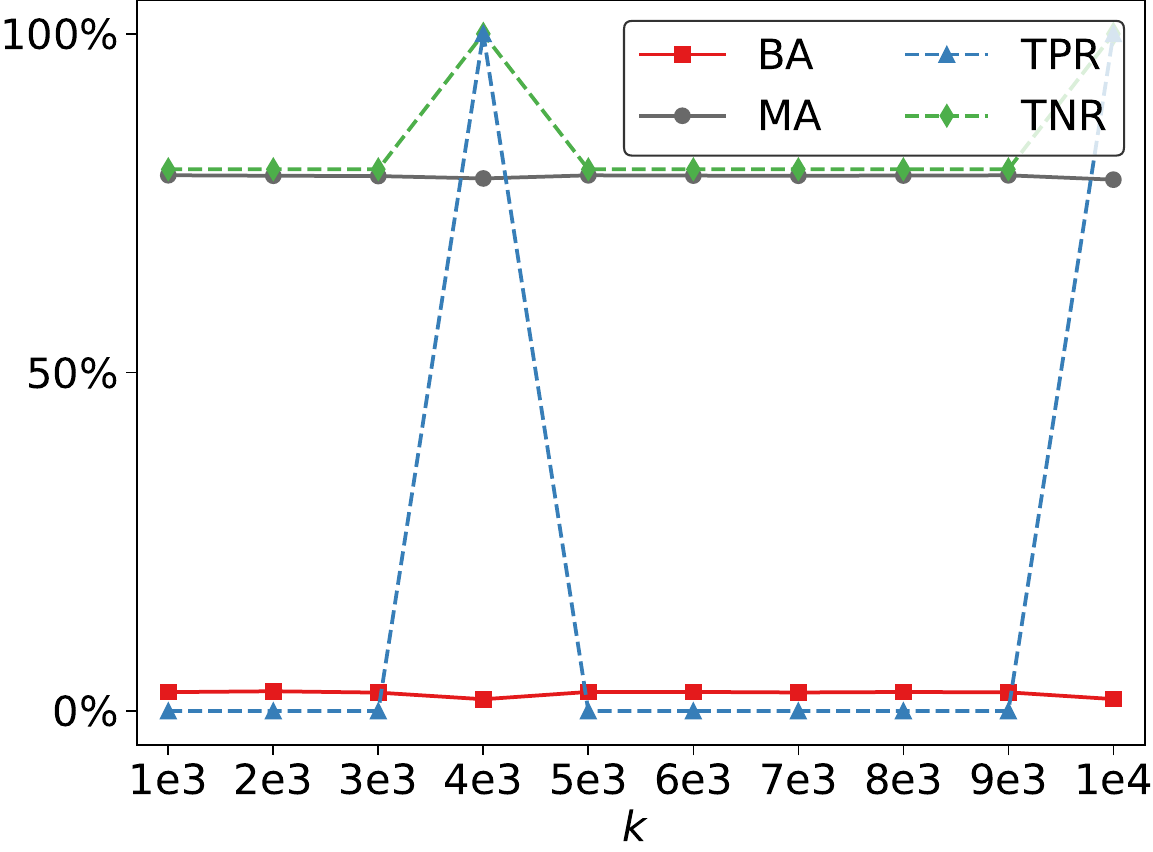}}
\subfloat[ABBR FABA]{\label{fig:abbr_faba_k} \includegraphics[width=0.23\textwidth]{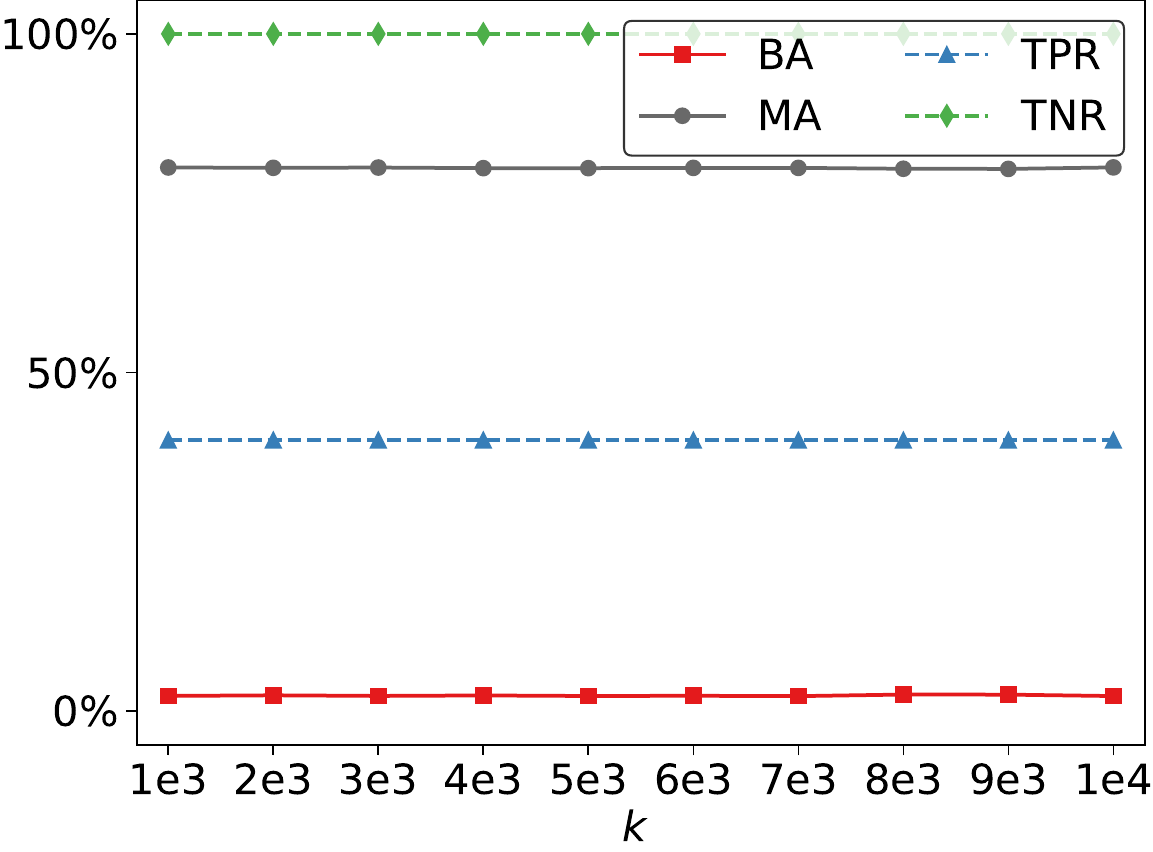}}
\subfloat[ABBR FLAME]{\label{fig:abbr_flame_k} \includegraphics[width=0.23\textwidth]{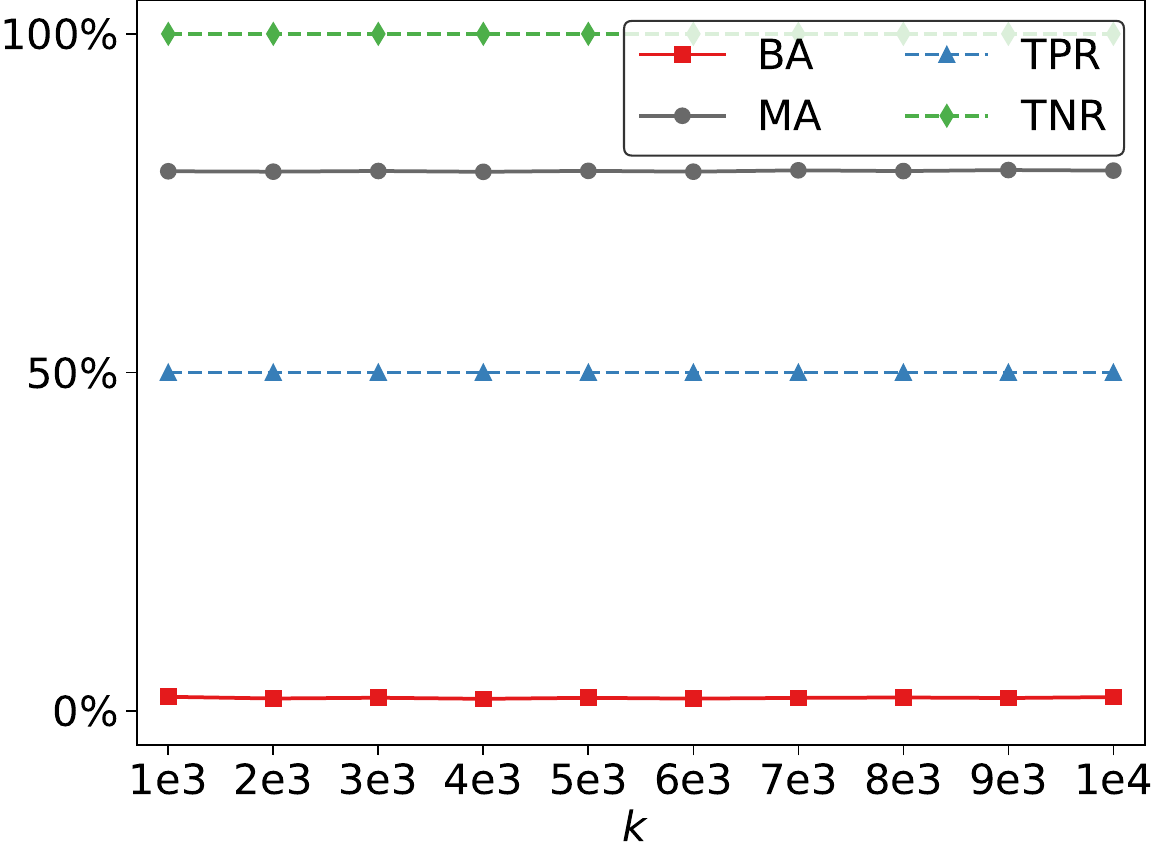}}\\
\caption{Impact of the target dimensionality $k$ for ABBR framework on CIFAR-10 dataset.}
\label{fig:k_impact}
\end{figure*}

\myparatight{Hyperparameters}
By default, we assume that in each iteration of federated learning, $n = 10$ clients are randomly selected from a total of $m = 100$ clients to participate in training. For local training on each client, we set the batch size to 64, the number of local epochs to 2, the optimizer to stochastic gradient descent (SGD), and the loss function to cross-entropy. For the dimensionality reduction component, we set the parameters $\varepsilon$ and $\eta$ to 0.1 and 1, respectively, as this offers a balanced trade-off between robustness and efficiency. Based on the Inequality~\ref{Target_dimension}, we compute the minimum required target dimensionality $k$ for different numbers of participating clients. Specifically, when the number of selected clients in an iteration is 4, 8, 10, 16, 32, 64, or 128, the corresponding values of $k$ are set to 1073, 1465, 1599, 1889, 2332, 2783, and 3240, respectively.

\myparatight{Simulating non-IID data} We simulate non-independent and identically distributed (non-IID) data in two ways. Following prior studies \cite{bagdasaryan2020backdoor,xie2019dba}, we use a Dirichlet distribution to model non-IID training datasets across different clients. Specifically, for each data class, we use a Dirichlet distribution with parameter $p$ to determine the proportion of data assigned to each client. A smaller value of $p$ indicates a higher degree of non-IID distribution. The second method, based on previous works \cite{nguyen2022flame,fang2020local}, simulates non-IID data by assigning data with label $l$ to the $l$-th client group at a proportion $p \in [0,1]$, while the remaining data is randomly distributed to other client groups. In this method, a higher $p$ represents a greater level of non-IID. Unless stated otherwise, we use the first method to simulate non-IID data and set $p = 0.5$ by default for both methods.

\myparatight{STPC Protocol} The STPC protocol in ABBR-FLAME is the same as in the private FLAME, i.e., ABY~\cite{DBLP:conf/ndss/Demmler0Z15}. This protocol is run among two parties. The detailed setting of ABY is described in Appendix~\ref{detailed_setup}. 

\myparatight{Evaluation Metrics} For performance evaluation, we use the metrics of \emph{Runtime} and \emph{Communication Overhead} during both the setup and online phases to assess the efficiency of privacy-preserving computation. 
To assess robustness, we use a set of metrics including \emph{BA} (Backdoor Accuracy), \emph{MA} (Main Task Accuracy), \emph{TPR} (True Positive Rate), and \emph{TNR} (True Negative Rate), following the approach in~\cite{nguyen2022flame}. Further details about these metrics can be found in Appendix~\ref{detailed_setup}.

\subsection{Performance}
\label{subsec:performance}
To demonstrate the efficiency of our framework, we separately compare the runtime and communication overhead of ABBR versions and baselines in the setup phase and online phase.

\subsubsection{Runtime}

\textbf{Our ABBR versions demonstrate a significantly shorter runtime compared to the baselines.} We measure the runtime (in seconds) for both versions of ABBR and the baseline models under varying computational demands. A less intensive scenario involved a CNN model with 4 clients, while a more challenging one used ResNet-18 with 128 clients. As shown in Table~\ref{tab:setup_runtime} and Table~\ref{tab:online_runtime}, our framework achieves a substantial reduction in runtime. For example, with ResNet-18 and 32 clients, the setup phase of our ABBR versions is \underline{86x}, \underline{84x}, \underline{161x}, and \underline{84x} faster than the baselines, and the online phase is \underline{620x}, \underline{598x}, \underline{688x}, and \underline{590x} faster. Notably, with ResNet-18 and 64 clients, the baselines surpassed 30,000 seconds without completing, leading to termination, whereas ABBR versions completed the setup in just 321, 326, 390, and 326 seconds, and the online phase in 18, 22, 45, and 23 seconds.

\subsubsection{Communication overhead}
\textbf{The communication overhead of our ABBR versions is considerably lower than the baselines.} We conduct a separate evaluation of the communication overhead (in GB) for both ABBR versions and the baselines. As shown in Table~\ref{tab:setup_communication} and Table~\ref{tab:online_communication}, our ABBR framework drastically reduces communication overhead. The baselines, in both setup and online phases, incur substantial communication overhead, which not only consumes a large amount of costly bandwidth but also significantly impacts the protocol’s runtime. However, even under the most demanding scenario (ResNet-18 with 128 clients), the communication overhead of ABBR versions remains manageable, at only 85, 99, 154, and 77 GB during the setup phase, and 1.14, 1.41, 1.66, and 1.05 GB during the online phase. This clearly demonstrates that our framework effectively addresses the excessive communication overhead bottleneck.

\subsection{Resilience to Byzantine Failures}
\label{subsec:resilience}

In this section, we evaluate the robustness of 4 baseline vector-wise filtering rules and their corresponding ABBR versions under various attack scenarios. Additionally, we evaluate the resilience of ABBR against two worst-case adaptive attacks.

\myparatight{Attack scenarios} Following the experimental setup of Multi-Krum~\cite{he2020secure}, FoolsGold~\cite{fung2018mitigating}, FABA~\cite{xia2019faba} and FLAME~\cite{nguyen2022flame}, we simulate the attack scenarios through Constrain-and-Scale~\cite{bagdasaryan2020backdoor}, DBA~\cite{xie2019dba}, PGD~\cite{wang2020attack}, Edge-Case~\cite{wang2020attack}, Label Flipping~\cite{fang2020local}, Gaussian~\cite{fang2020local}, which are state-of-the-art attacks in FL. 
It is important to note that Constrain-and-Scale, DBA, PGD, and Edge-Case are targeted attacks, while Label Flipping and Gaussian are untargeted attacks.
We introduce these attacks in detail in Appendix~\ref{detailed_setup}.
Among these attack methods, the Byzantine local model has a good similarity with the benign local model, so the robustness of our ABBR framework can be well evaluated.

In order to quickly evaluate ABBR versions against Byzantine clients in more scenarios, we do not consider privacy-preserving computations here, which does not affect our evaluation results. For CIFAR-10, EMNIST, Fashion-MNIST, and Tiny-ImageNet datasets, we separately pre-train the lightweight ResNet-18, CNN, featherweight ResNet-18, and ResNet-18 for 1000, 200, 1000, and 1000 iterations without adversaries and use them as the initial global models for the above attacks. Unless otherwise mentioned, we set the percentage of Byzantine clients in an iteration is 20\%.

\myparatight{Robustness against state-of-the-art attacks}
We conduct experiments under a benign setting (where the percentage of Byzantine clients is 0 in each iteration) as well as under all the aforementioned attack scenarios. Table~\ref{tab:robustness} illustrates the robustness of our framework against state-of-the-art attacks, showing that the robustness of ABBR versions and the baselines are nearly identical. In the benign setting, most ABBR versions exhibit a \emph{MA} drop of no more than 0.6\% compared to the baselines. While the \emph{MA} of ABBR FLAME decreases by 2.5\% on the CIFAR-10 dataset, this can be mitigated by increasing the dimension $k$ (as described in Section~\ref{subsec:dimensionality_reduction}) of the low-dimensional space. For instance, when the dimension $k$ is increased from 1599 to 3084, the \emph{MA} of ABBR FLAME improves to 76.9\%, which is only 1\% lower than the baseline FLAME.

Under the Constrain-and-Scale, DBA, PGD, and Edge-Case attacks, most ABBR versions show \emph{BA} equal to or lower than the baselines and \emph{MA} equal to or higher. Notably, ABBR significantly enhances the robustness of FoolsGold. For instance, in the CIFAR-10 dataset, ABBR FoolsGold consistently achieves a much lower \emph{BA} than the baseline across all attack types. Upon reviewing the baseline FoolsGold, we found it mistakenly accepted all Byzantine models. While ABBR FoolsGold also accepted them, the adaptive tuning component (see Section~\ref{subsec:adaptive_tuning}) clipped the norms of these models, successfully defending against the attacks. Furthermore, under Label Flipping and Gaussian attacks, the \emph{MA} of ABBR versions drops by less than 0.2\%, demonstrating that ABBR does not compromise the robustness of vector-wise filtering.

\myparatight{Varying the percentage of Byzantine clients}
We investigate the impact of the Byzantine client rate (\emph{BCR} = $\frac{c}{n}$, where $c$ is the number of Byzantine clients) through Constrain-and-Scale attack experiments on the CIFAR-10 dataset. As shown in Figure~\ref{fig:bcr_impact}, most ABBR versions exhibit comparable or even better robustness than the baselines. First, ABBR versions of Multi-Krum, FABA, and FLAME show the same \emph{TPR} and \emph{TNR} as the baselines. This is because the random projection in our ABBR framework approximately preserving pairwise distances in the lower-dimensional space, thereby helping to maintain the accuracy of the filtering mechanism. Second, ABBR FoolsGold outperforms baseline FoolsGold in \emph{TPR} and \emph{TNR} when \emph{BCR} is 30\% or 50\%, likely due to the loss of noise parameters during projection into the low-dimensional space, making Byzantine models easier to detect. Lastly, ABBR versions of Multi-Krum and FABA achieve higher \emph{MA} than baselines at 50\% \emph{BCR}, thanks to our adaptive tuning component, which constrains malicious local models without unduly suppressing the contributions of benign models.

We also examine the scenario where all clients participate in each iteration, shown in Appendix~\ref{appendix:results}, reaching the same conclusions.

\myparatight{Varying the degree of non-IID data}
The degree of non-IID data influences the relative position of local models in the original-dimensional space, making it a key factor in whether dimensionality reduction affects filtering results. To further explore its impact on the ABBR framework, we conduct experiments by varying the degree of non-IID data (simulated via the Dirichlet distribution) on the CIFAR-10 dataset. Figure~\ref{fig:alpha_impact} shows that ABBR versions demonstrate robustness comparable to the baselines across different non-IID levels. First, ABBR versions of Multi-Krum, FABA, and FLAME show similar \emph{TPR}, \emph{TNR}, \emph{BA}, and \emph{MA} as the baselines. Second, ABBR FoolsGold exhibits lower \emph{BA} and higher \emph{MA}, even with lower \emph{TPR} and \emph{TNR} under certain non-IID conditions. This is because our adaptive tuning strategy mitigates the impact of undetected malicious models by ensuring their contributions to the global model are heavily down-weighted, thus preventing them from successfully executing an attack.

Additionally, we perform the experiment using another method to simulate non-IID data (described in Section~\ref{subsec:experimental_setup}) in Appendix~\ref{appendix:results}, reaching the same conclusions.

\myparatight{Varying the target dimensionality $k$ of ABBR} We conduct the experiment by varying the target dimensionality $k$ from 1,000 to 10,000 and  evaluating all ABBR versions against the Constrain-and-Scale attack on CIFAR-10 dataset. The results are presented in Figure~\ref{fig:k_impact}, which demonstrates the robustness of our ABBR framework to variations in the $k$. Specifically, we observe that the TPR and TNR of filtering mechanism do not exhibit a significant decline when $k$ is decreased from 10,000 down to 1,000. This indicates that while a smaller $k$ theoretically increases the potential range of distance distortion, its impact on the practical accuracy of filtering mechanism is minimal.

\myparatight{Adaptive attacks}
We derive two worst-case adaptive attacks to assess the robustness of ABBR. In these scenarios, the adaptive attacker has full knowledge of the aggregation rules and all benign local models. The first attack, called Adaptive Maximum Deviation (Adaptive-MD), is based on the AGR-tailored attack~\cite{shejwalkar2021manipulating}. Here, the attacker aims to maximize the deviation between the aggregated model update post-attack and the ground truth update by using the \emph{inverse sign} vector of all benign updates as the deviation vector, adjusting its scale for maximum impact. The second attack, Adaptive Attack Tuning (Adaptive-AT), targets our tuning strategy by maximizing the deviation while constrained by the median norm of all benign updates. The detailed steps of these attacks are outlined in Algorithms~\ref{alg:Adaptive-MD} and~\ref{alg:Adaptive-AT} in Appendix.

We evaluate the robustness of ABBR under these adaptive attacks by using FLAME and its ABBR version (ABBR FLAME) as defenses. In each iteration, all 100 clients participate, with 20 acting as attackers. The initial scaling factor is set to 10 for Adaptive-MD and 50 for Adaptive-AT. As shown in Table~\ref{table:adaptive_attack}, our ABBR framework maintains the robustness of Byzantine-robust aggregation rules even under these worst-case adaptive attacks. For example, both FLAME and ABBR FLAME achieve consistent main task accuracy compared to the no-defense scenario under the two adaptive attacks.

\begin{table}[]

\footnotesize
\renewcommand\arraystretch{1}
\caption{Robustness of ABBR against adaptive attacks on different datasets.  The two adaptive attacks are untargeted and thus have no BA metric. The Metric is the main task accuray, and all values are percentages.}
\label{table:adaptive_attack}
\begin{tabular}{c|c|c|c|c}
\hline
Datasets                       & Attacks     & \makecell{No \\ Defense} & FLAME & \makecell{ABBR \\ FLAME} \\ \hline
\multirow{2}{*}{CIFAR-10}      & Adaptive-MD & 10.0   & 81.4  & 81.2       \\ \cline{2-5} 
                               & Adaptive-AT & 10.0   & 81.4  & 81.3       \\ \hline
\multirow{2}{*}{EMNIST}        & Adaptive-MD & 10.0   & 98.9  & 98.9       \\ \cline{2-5} 
                               & Adaptive-AT & 10.0   & 98.9  & 98.9       \\ \hline
\multirow{2}{*}{Fashion-MNIST} & Adaptive-MD & 10.0   & 92.3  & 92.4       \\ \cline{2-5} 
                               & Adaptive-AT & 10.0   & 92.4  & 92.4       \\ \hline
\end{tabular}
\end{table}

\section{Related Work}
Over the past decade, many studies have separately explored Byzantine-resilient aggregation rules and client data privacy protection. However, combining both introduces a significant overhead bottleneck, as discussed earlier. Below, we outline three key research areas in more detail.

\myparatight{Byzantine-resilient aggregation rules}
Existing rules can be divided into three main categories. The first category~\cite{nguyen2022flame, guerraoui2018hidden, blanchard2017machine,diakonikolas2019sever,259745,khazbak2020mlguard,li2019abnormal,li2020learning,shen2016auror,munoz2019byzantine} distinguishes Byzantine from benign models using pairwise distances. For example, FLAME \cite{nguyen2022flame} uses HDBSCAN~\cite{mcinnes2017hdbscan} to classify updates by Cosine distance. The second category~\cite{guerraoui2018hidden,yin2018byzantine,chen2017distributed,gao2022secure} identifies malicious elements at the dimensional level, like Median~\cite{chen2017distributed} and Trimmed Mean approaches~\cite{chen2017distributed}. The third category~\cite{cao2020fltrust,andreina2021baffle} mitigates Byzantine failures using auxiliary datasets, such as FLTrust~\cite{cao2020fltrust}, which relies on a benign dataset. However, none of these methods protect client privacy, as they process local updates in plaintext, allowing servers to infer sensitive data.

\myparatight{Privacy protection of client data} Recent researchs~\cite{pyrgelis2017knock,salem2018ml,melis2019exploiting,hitaj2017deep,DBLP:conf/nips/ZhuLH19,DBLP:journals/corr/abs-2001-02610,wang2019beyond,DBLP:conf/uss/Fu0JCWG0L022,DBLP:journals/tdsc/GaoHGLZCL23,DBLP:journals/tdsc/WangHSWXR23} have shown that private data attributes can be inferred from local updates, and even the original data can be reconstructed. To safeguard client data, privacy-preserving FL schemes have been proposed~\cite{bonawitz2017practical,bell2020secure,DBLP:journals/tdsc/XuHXZLHD23}. For example, the work~\cite{bonawitz2017practical} introduced secure aggregation using additive masks to encrypt updates, which cancel out when aggregated, ensuring linear scalability with the number of clients. The work~\cite{bell2020secure} improved this with a more scalable solution, reducing client overhead. However, while these methods protect privacy, they cannot defend against Byzantine clients.

\myparatight{Privacy-Preserving and Byzantine-resilient FL protocol} Few studies~\cite{he2020secure,nguyen2022flame,DBLP:conf/esorics/DongCLWZ21,so2020byzantine,DBLP:journals/tdsc/ZhaoJFWSL22} address both Byzantine resilience and client data privacy simultaneously. Protocols like those in~\cite{he2020secure,nguyen2022flame,DBLP:conf/esorics/DongCLWZ21} use a two-server architecture with secure two-party computation to protect privacy while mitigating Byzantine failures, scaling linearly with model parameters. BREA~\cite{so2020byzantine}, a single-server architecture, uses secure multi-party computation but imposes heavy overhead on clients, leading to low efficiency due to the high computational demands of private operators.


\section{Discussion}

\myparatight{Extensibility to other secure computation protocols} It is important to emphasize that the core methodological contributions of ABBR, the integration of random projection for dimensionality reduction and the adaptive tuning strategy, are not inherently dependent on ABY. These techniques are designed to be protocol-agnostic, relying on fundamental secure primitives such as secure addition, secure multiplication, and secure comparison, which are common to a wide array of secure computation protocols. Consequently, ABBR can be readily instantiated within alternative frameworks. For instance, ABY3~\cite{mohassel2018aby3} and MP-SPDZ~\cite{keller2020mp} could effectively host our ABBR mechanisms, potentially extending its applicability to different numbers of parties or security models.

\myparatight{Limitation and future work } The primary limitation is that the ABBR framework may not be able to adapt to dimension-wise filtering mechanisms. Our ABBR framework currently focuses on the vector-wise filtering mechanisms, which are commonly adopted for the robust FL. The core dimensionality reduction technique employed in ABBR, namely random projection, projects high-dimensional local models into a lower-dimensional space. While this is highly beneficial for efficiency and privacy in the context of vector-wise operations, it inherently alters the per-dimension information. Consequently, existing dimension-wise filtering methods, which rely on analyzing or comparing individual feature dimensions across models (e.g., Medain~\cite{yin2018byzantine} that operates on each dimension independently), may not be directly applicable.

Recognizing the limitation regarding dimension-wise filtering, a promising direction for future work is to investigate efficient privacy-preserving FL frameworks that can effectively adapt dimension-wise filtering mechanisms. This would further broaden the scope of robust and private aggregation methods available to the FL community.

\section{Conclusion}
We propose ABBR, a practical framework for Byzantine-robust and privacy-preserving  FL with low computation and communication overhead. ABBR provides two key components to improve the efficiency of a vector-wise filtering rule in STPC: the \emph{dimensionality reduction method} and the \emph{adaptive turning strategy}.
The ABBR version of a vector-wise filtering rule is easy to construct, and we implement various ABBR constructions of Multi-Krum, FoolsGold, FABA, and FLAME.
The evaluation of ABBR versions on four datasets shows that ABBR framework is more efficient than baseline framework and almost has the same Byzantine resilience.

\section{Acknowledgments}
This work is supported in part by the National Natural Science Foundation of China (No.62272251) and the Natural Science Foundation of Shandong Province (No.ZR2022LZH014).

\bibliographystyle{IEEEtran}
\bibliography{refs}

\appendix


\subsection{The detailed introduction of ABY}
\label{aby_detail}
ABY \cite{DBLP:conf/ndss/Demmler0Z15} is the default STPC backend for our framework, which combines secure computation schemes \cite{DBLP:conf/crypto/Beaver91a,DBLP:conf/stoc/GoldreichMW87,DBLP:conf/focs/Yao82b} based on Arithmetic sharing, Boolean sharing, and Yao sharing. In ABY, a private value can be secretly shared between two parties through three types of sharing, and there are efficient conversions between the three types. We next introduce the Arithmetic Sharing and Boolean Sharing.

\myparatight{Arithmetic Sharing } Arithmetic Sharing is that a private $\ell$-bit integer is shared additively in the ring $\mathbb{Z}_{2^\ell}$ as two random integers, and the private integer is the sum of the two random integers. Next, we introduce the operations we will use on Arithmetic sharings performed in the ring $\mathbb{Z}_{2^\ell}$. We assume there are two parties $P_0$ and $P_1$.

\begin{itemize}
    \item $\mathbf{SHR^A(x)}$: It is a function that a private $\ell$-bit integer $x$ is shared additively as two arithmetic sharings $ (\left \langle x \right \rangle _0^A , \left \langle x \right \rangle _1^A)$, where $x = \left \langle x \right \rangle _0^A + \left \langle x \right \rangle _1^A$. 
    \item $\mathbf{ADD(\left \langle x \right \rangle ^A,\left \langle y \right \rangle ^A)}$: It is a function that the inputs on each party are the arithmetic sharings of the private integer $x$ and $y$ and the output on each party is the sharing of the $x + y$. The subtraction function $\mathbf{SUB(\left \langle x \right \rangle ^A,\left \langle y \right \rangle ^A)}$ is like $\mathbf{ADD(\left \langle x \right \rangle ^A,\left \langle y \right \rangle ^A)}$, so we omit its detail.
    \item $\mathbf{MUL(\left \langle x \right \rangle ^A,\left \langle y \right \rangle ^A)}$: It is a function that the inputs on each party are the arithmetic sharings of the private integer $x$ and $y$ and the output on each party is the sharing of the $xy$.
\end{itemize}

\myparatight{Boolean Sharing} Boolean Sharing is that a private $\ell$-bit value is shared as two random $\ell$-bit values based on bitwise XOR, and the private value is equal to the XOR of the two random values. Next, we introduce the operations we will use on Boolean sharings.
\begin{itemize}
    \item $\mathbf{CMP(\left \langle x \right \rangle ^B, \left \langle y \right \rangle ^B)}$: It is a function that the inputs on each party are the Boolean sharings of private value $x$ and $y$ and the output on each party is the Boolean sharing of 1 if $x > y$ or the sharing of 0 otherwise.  
    \item $\mathbf{MUX(\left \langle x \right \rangle ^B, \left \langle y \right \rangle ^B, \left \langle s \right \rangle ^B)}$: It is a function that the inputs on each party are the Boolean sharings of private value $x$, $y$ and $s$ and the output on each party is $\left \langle x \right \rangle ^B$ if $s = 1$ or $\left \langle y \right \rangle ^B$ otherwise. 

\end{itemize}

\subsection{The proof of Theorem~\ref{theorem_1}}
\label{proof:theo_1}
First, the error bound of Euclidean distance is the Theorem 2 in  \cite{arriaga2006algorithmic}, so omit the proof of it. Then, according the Corollary 2 in \cite{arriaga2006algorithmic}, we can get a conclusion that with probability at least $1-4e^{-(1+\frac{1}{2}\eta)\log n  }$, for any local model $\left \|L_i\right \|$, $\left \|L_j\right \| \le 1$,
\begin{equation}
\label{Cosine_error1}
 L_i \cdot L_j -  \varepsilon \leq L'_i \cdot L'_j  \leq  L_i \cdot L_j +  \varepsilon .
\end{equation}
Although the inequality above indicates the error bound of the inner product of two local models whose norm is not greater than 1, we can obtain the error bound of cosine similarity between two local model with any norm by the normalization, i.e.,
\begin{equation}
\label{Cosine_error2}
 \frac{L_i}{\left \| L_i \right \|} \cdot\frac{L_j}{\left \| L_j \right \|} -  \varepsilon \leq \frac{L'_i}{\left \| L_i \right \|}\cdot \frac{L'_j}{\left \| L_j \right \|}  \leq   \frac{L_i}{\left \| L_i \right \|} \cdot\frac{L_j}{\left \| L_j \right \|} +  \varepsilon .
\end{equation}
However, the cosine similarity between two low-dimensional local models is generally computed by computing the ratio of the inner product to the product of their L2 norms. Therefore, Inequality \ref{Cosine_error2} needs to take into account the error of the L2 norm in order to obtain the final distortion bound of the Cosine similarity, i.e.
\begin{multline}
\label{Cosine_error3}
\frac{1}{1+\varepsilon} \left (  \frac{L_i}{\left \| L_i \right \|} \cdot\frac{L_j}{\left \| L_j \right \|} -  \varepsilon \right )  \leq \frac{L'_i}{\left \| L'_i \right \|}\cdot \frac{L'_j}{\left \| L'_j \right \|} \\
\leq \frac{1}{1-\varepsilon} \left (  \frac{L_i}{\left \| L_i \right \|} \cdot\frac{L_j}{\left \| L_j \right \|} +  \varepsilon\right ) .
\end{multline}
Furthermore, we can obtain the distortion bound of the Cosine distance, i.e., 
\begin{multline}
\label{theorem_1:cosine_distance}
\frac{1}{1-\varepsilon} \left ( 1-  \frac{L_i}{\left \| L_i \right \|} \cdot\frac{L_j}{\left \| L_j \right \|} -2\varepsilon \right )  \leq 1- \frac{L'_i}{\left \| L'_i \right \|}\cdot \frac{L'_j}{\left \| L'_j \right \|} \\
\leq \frac{1}{1+\varepsilon} \left ( 1- \frac{L_i}{\left \| L_i \right \|} \cdot\frac{L_j}{\left \| L_j \right \|} +2\varepsilon \right ) .
\end{multline}

\begin{figure*}[htp]
\centering
\subfloat[Multi-Krum]{\label{fig:all_various_byzantine_clients_multi_krum} \includegraphics[width=0.24\textwidth]{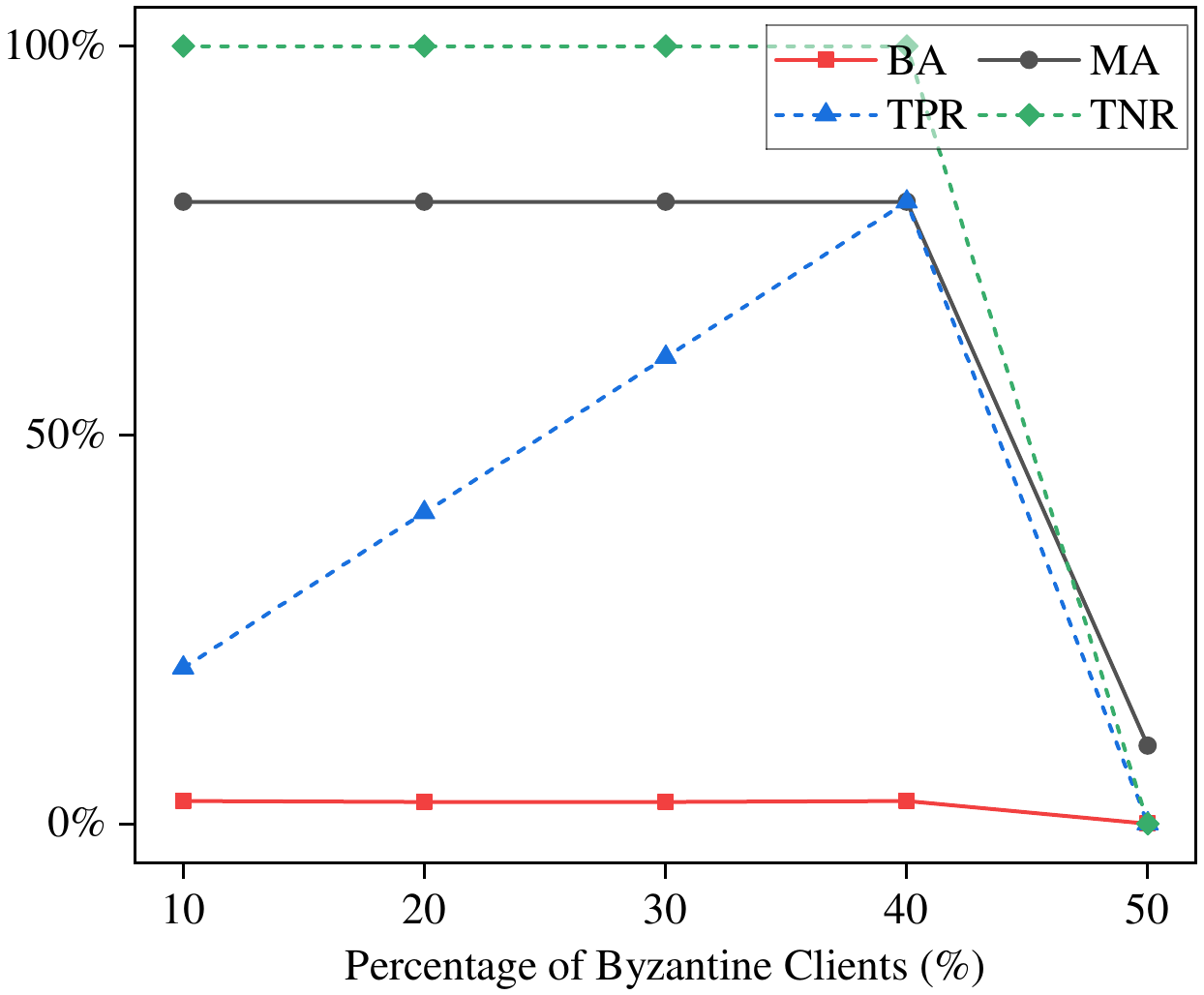}}
\subfloat[FoolsGold]{\label{fig:all_various_byzantine_clients_foolsgold} \includegraphics[width=0.24\textwidth]{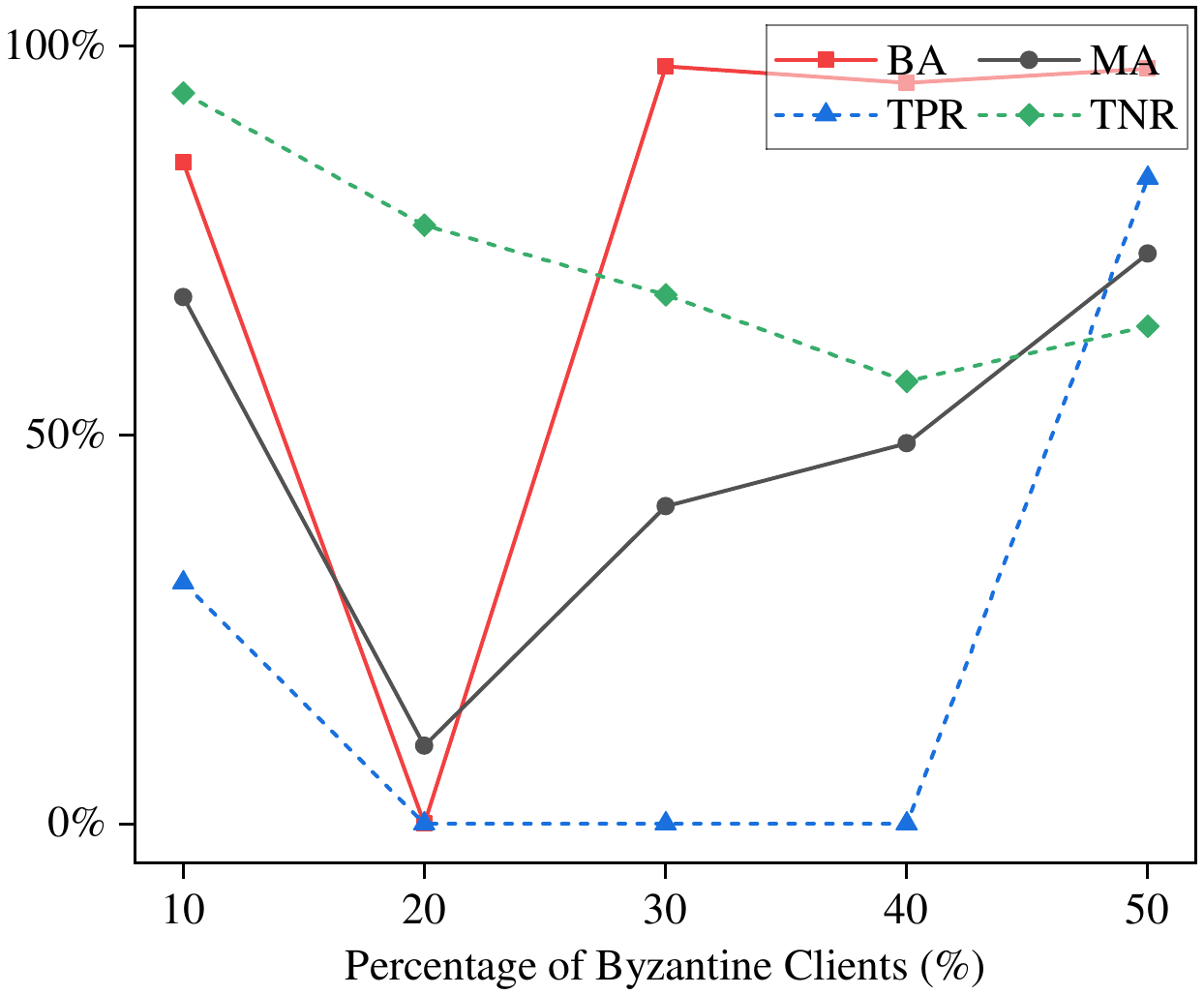}}
\subfloat[FABA]{\label{fig:all_various_byzantine_clients_faba} \includegraphics[width=0.24\textwidth]{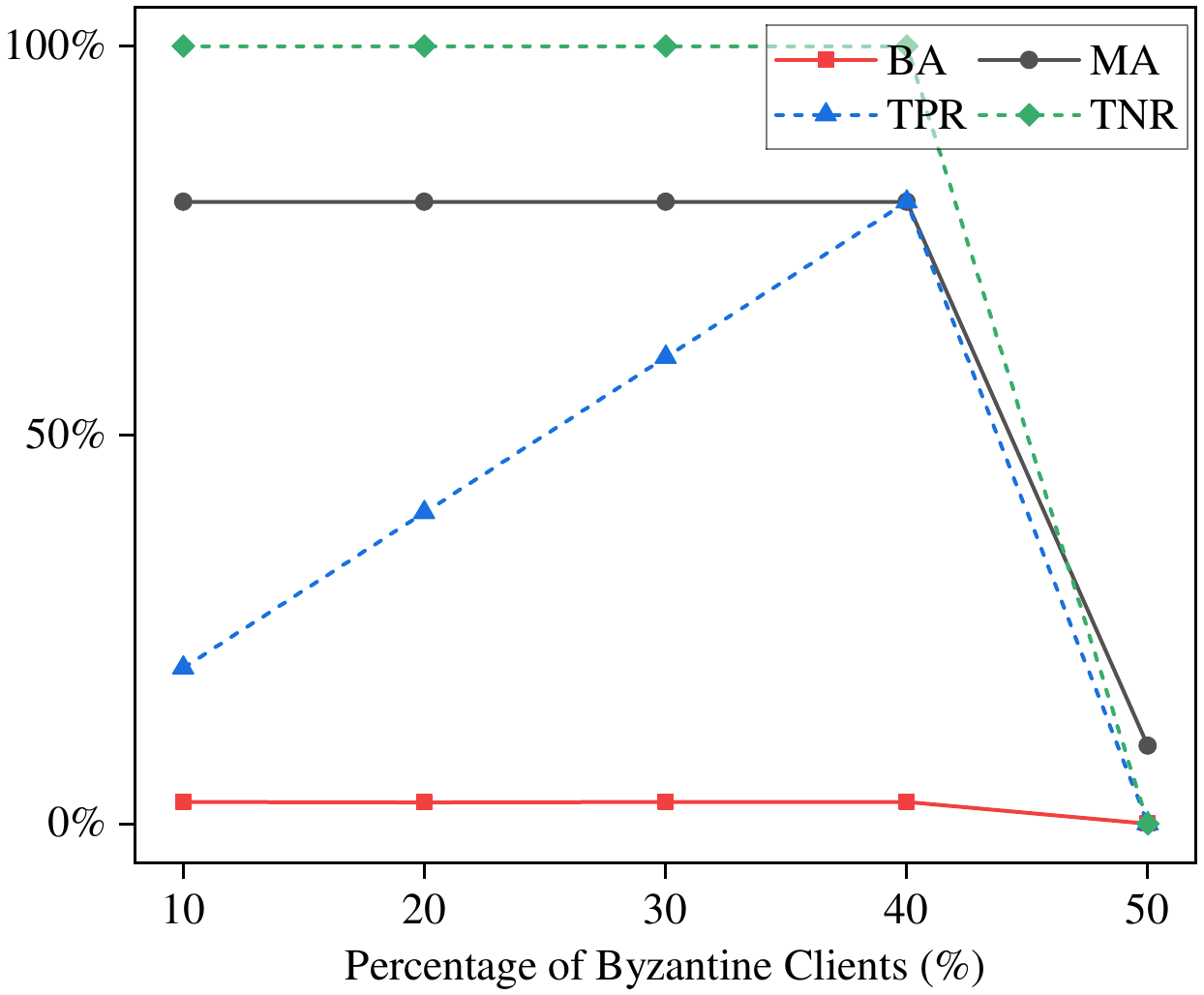}}
\subfloat[FLAME]{\label{fig:all_various_byzantine_clients_flame} \includegraphics[width=0.24\textwidth]{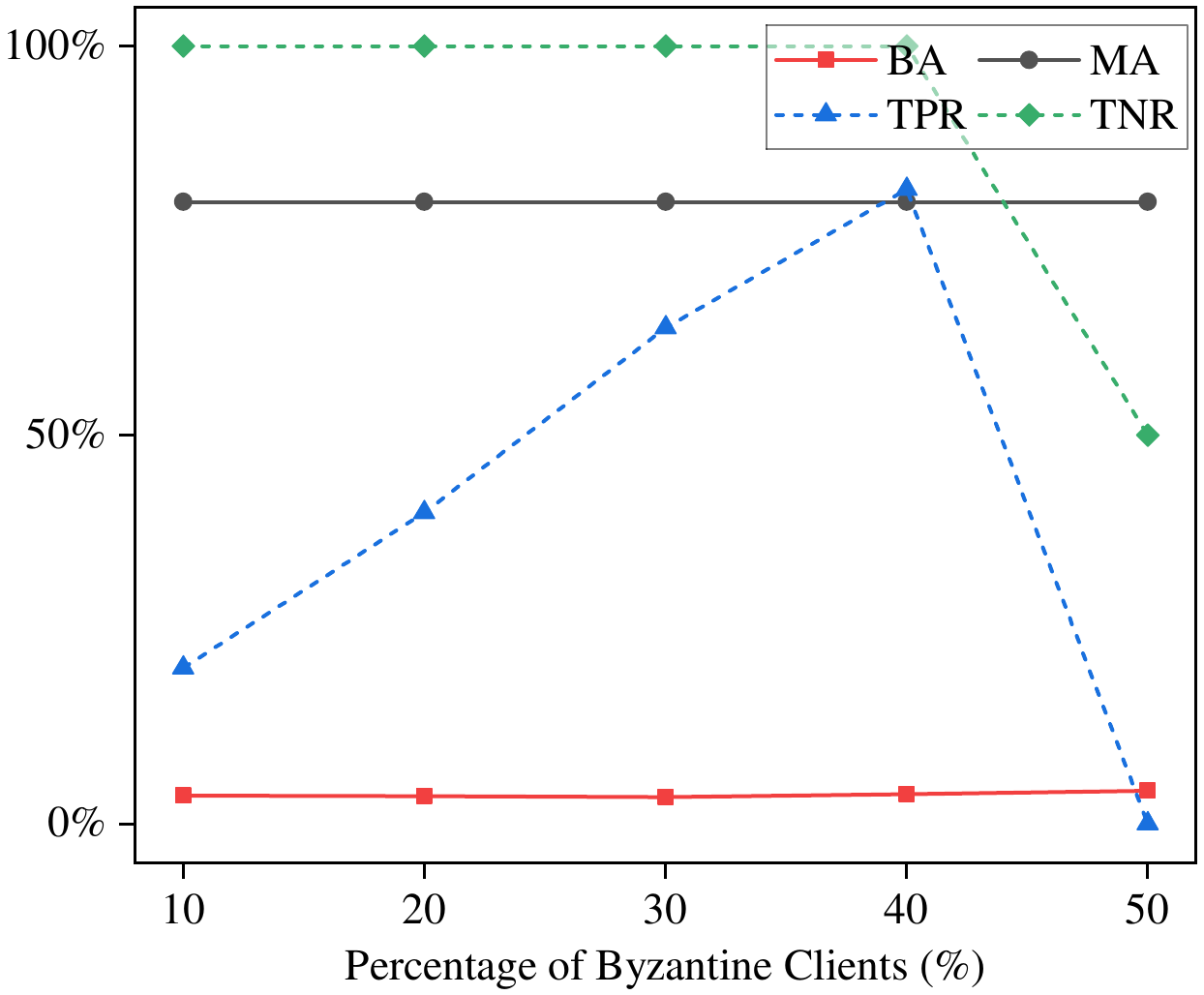}}\\

\subfloat[ABBR Multi-Krum]{\label{fig:all_various_byzantine_clients_multi_krum_abbr} \includegraphics[width=0.24\textwidth]{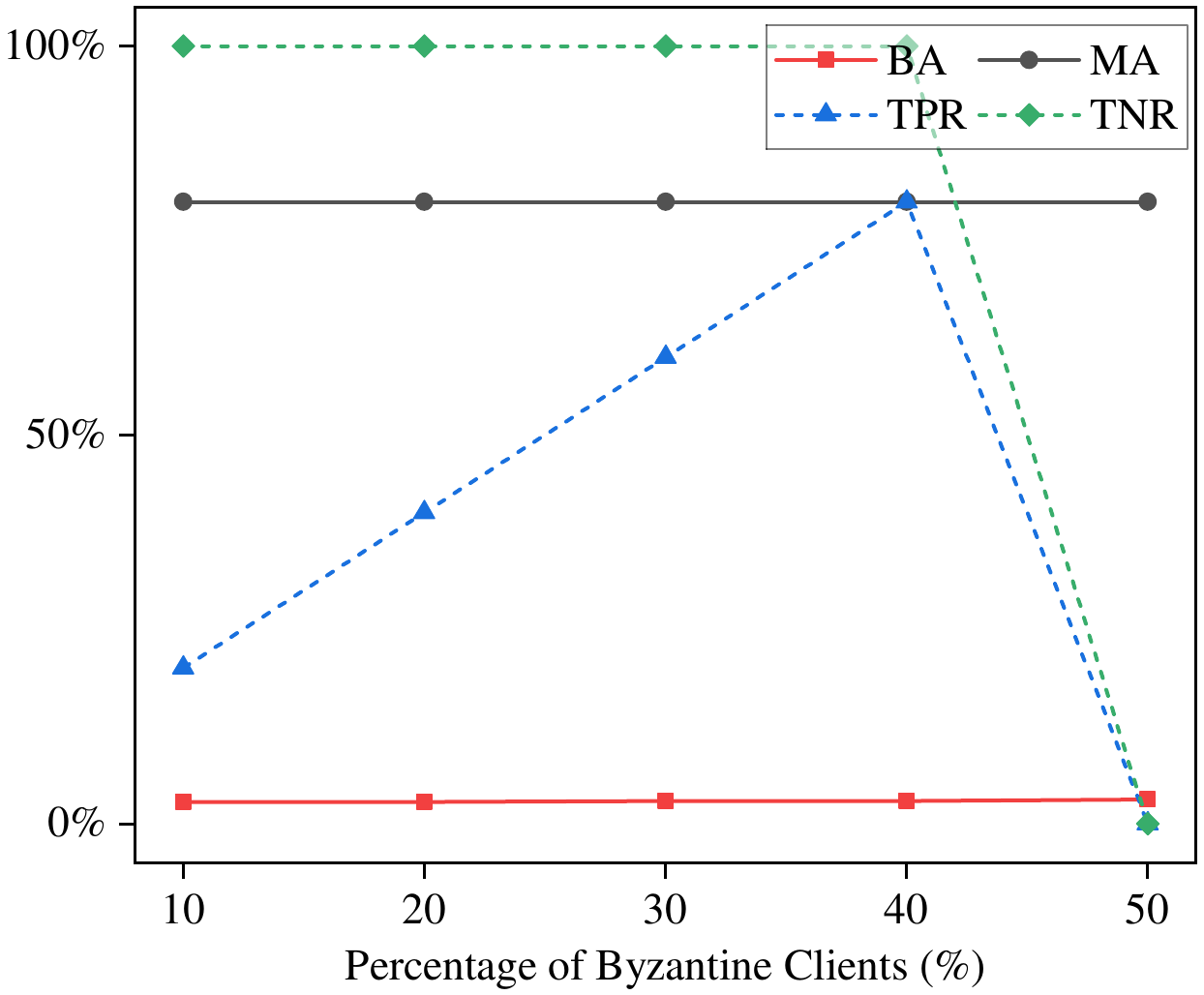}}
\subfloat[ABBR FoolsGold]{\label{fig:all_various_byzantine_clients_foolsgold_abbr} \includegraphics[width=0.24\textwidth]{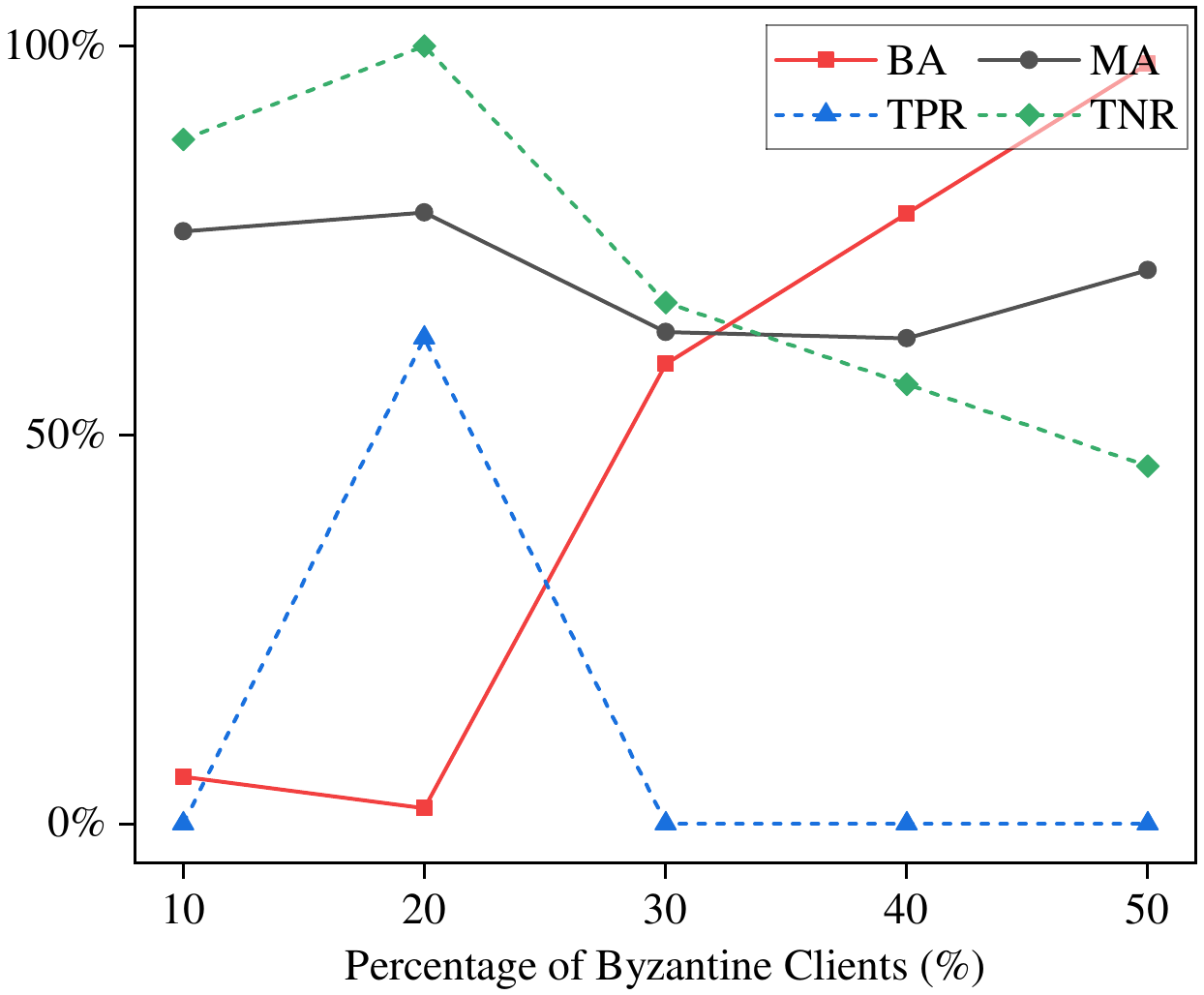}}
\subfloat[ABBR FABA]{\label{fig:all_various_byzantine_clients_faba_abbr} \includegraphics[width=0.24\textwidth]{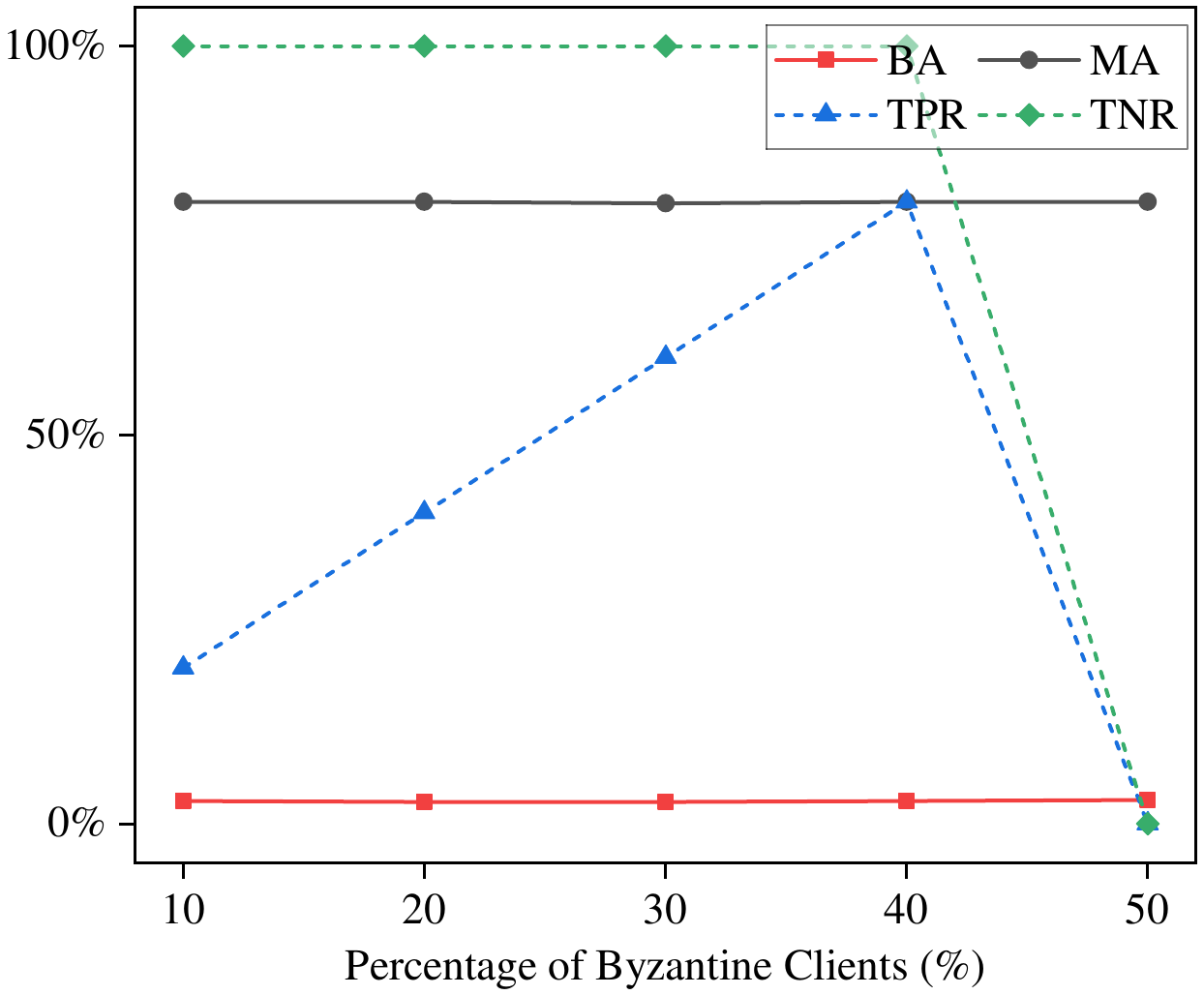}}
\subfloat[ABBR FLAME]{\label{fig:all_various_byzantine_clients_flame_abbr} \includegraphics[width=0.24\textwidth]{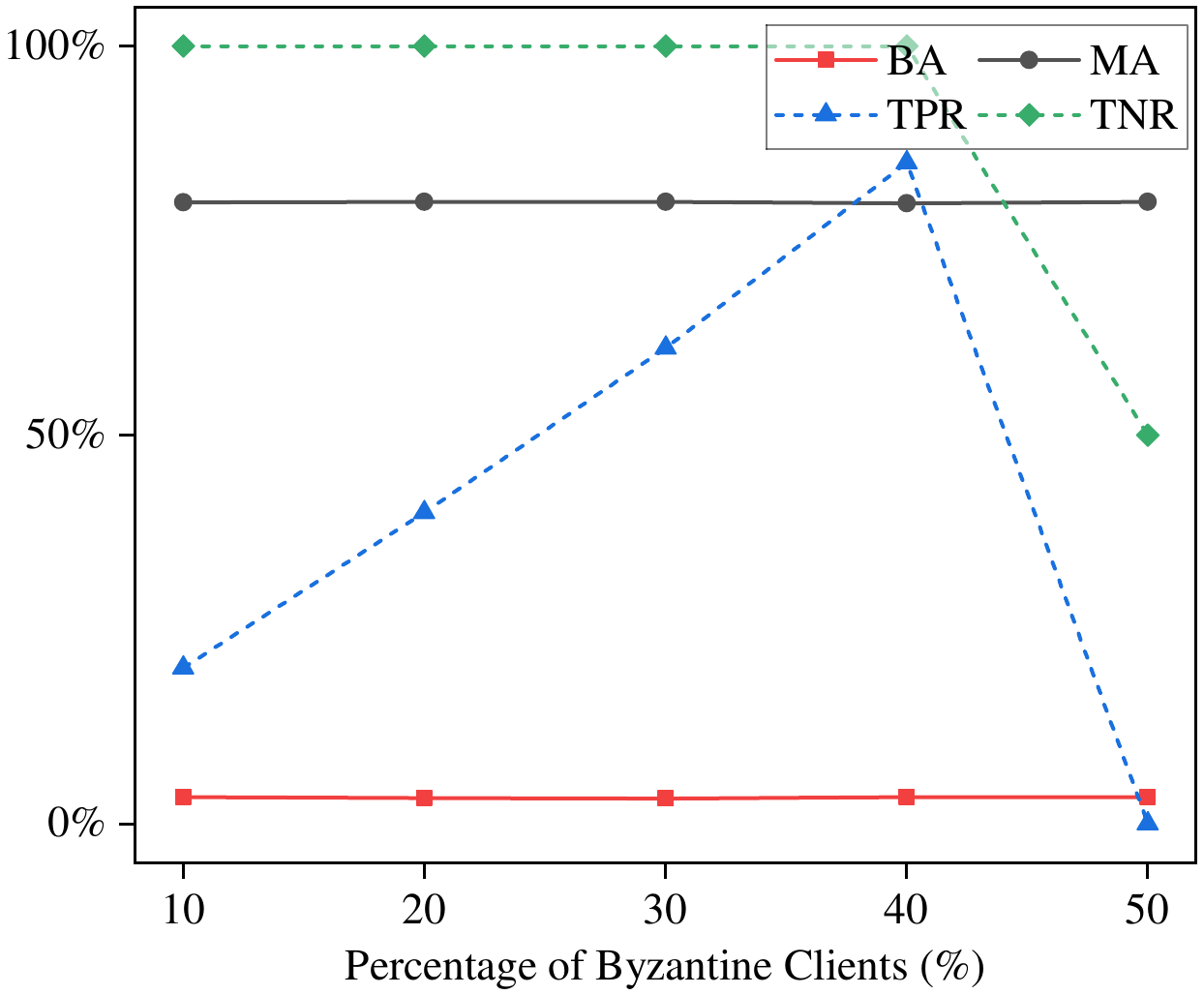}}\\
\caption{Impact of the percentage of Byzantine clients on CIFAR-10 dataset when all clients are participated in an iterarion. (a)-(d): original vector-wise filterings, (e)-(f): ABBR versions of vector-wise filterings.}
\label{fig:all_bcr_impact}
\end{figure*}

\subsection{The proof of Theorem~\ref{theorem_2}}
\label{proof:theo_2}
We first provide the proof based on Euclidean distance. Suppose Assumption~\ref{assumption_1} holds, the goal of vector-wise filtering is formulated as follows:
\begin{equation}
\label{goal_euclidean_distance}
I_{opt} = \mathrm{arg} \min_{I \subset I_{N}}{\left \| \frac{1}{\left | I \right | } \sum_{j\in I} L_j-G_t^* \right \|}^2 \leq \alpha,
\end{equation}
where we assume that the aggregation $G^t =  \frac{1}{\left | I \right | } \sum_{j\in I}$, $I_{N}$ represents the index set of all local models. We assume $I_{opt}$ is the index set of accepted local models by vector-wise filtering with the Euclidean distance in the original-dimensional space, $I'_{opt}$ is the index set of accepted local models by vector-wise filtering with the Euclidean distance in the low-dimensional space.

We first consider the case that $I^{'}_{opt}$ and $I_{opt}$ are equal, then:
\begin{equation}
\label{case_1}
{\left \| \frac{1}{\left | I^{'}_{opt} \right | } \sum_{j\in I^{'}_{opt}} L_j-G_t^* \right \|}^2 = {\left \| \frac{1}{\left | I_{opt} \right | } \sum_{j\in I_{opt}} L_j-G_t^* \right \|}^2.
\end{equation}
Then, we consider the case that $I^{'}_{opt}$ and $I_{opt}$ are not equal. Since the distance between any two local models changes in the low-dimensional space, the optimal group $I^{'}_{opt}$ may change in two directions. The one is: 
\begin{equation}
\label{case_2}
{\left \| \frac{1}{\left | I^{'}_{opt} \right | } \sum_{j\in I^{'}_{opt}} L^{'}_j-G_t^{'*} \right \|}^2 > {\left \| \frac{1}{\left | I_{opt} \right | } \sum_{j\in I_{opt}} L^{'}_j-G_t^{'*} \right \|}^2,
\end{equation}
where $L^{'}_j$, $G^{'}_t$, and $G_t^{'*}$ is the low-dimensional model. However, it contradicts the goal \ref{goal_euclidean_distance}. The other one is:
\begin{equation}
\label{case_3}
{\left \| \frac{1}{\left | I^{'}_{opt} \right | } \sum_{j\in I^{'}_{opt}} L^{'}_j-G_t^{'*} \right \|}^2 \le {\left \| \frac{1}{\left | I_{opt} \right | } \sum_{j\in I_{opt}} L^{'}_j-G_t^{'*} \right \|}^2.
\end{equation}
According to Theorem \ref{theorem_1}, the Inequality \ref{case_3} is occurred only when the following Inequality holds:
\begin{equation}
\label{case_3_1}
{\left \| \frac{1}{\left | I^{'}_{opt} \right | } \sum_{j\in I^{'}_{opt}} L_j-G_t^* \right \|}^2 \le \frac{1+\varepsilon}{1-\varepsilon}{\left \| \frac{1}{\left | I_{opt} \right | } \sum_{j\in I_{opt}} L_j-G_t^* \right \|}^2.
\end{equation}

According to the Inequality \ref{case_3_1}, we can prove the Inequality \ref{equ:assumption_1_1} in Theorem \ref{theorem_2}. The proof of Inequality \ref{equ:assumption_1_2} in Theorem \ref{theorem_2} is the same as above, so we omit it.

\begin{figure*}[htp]
\centering
\subfloat[Multi-Krum]{\label{fig:degree_noniid_2_multi_krum} \includegraphics[width=0.24\textwidth]{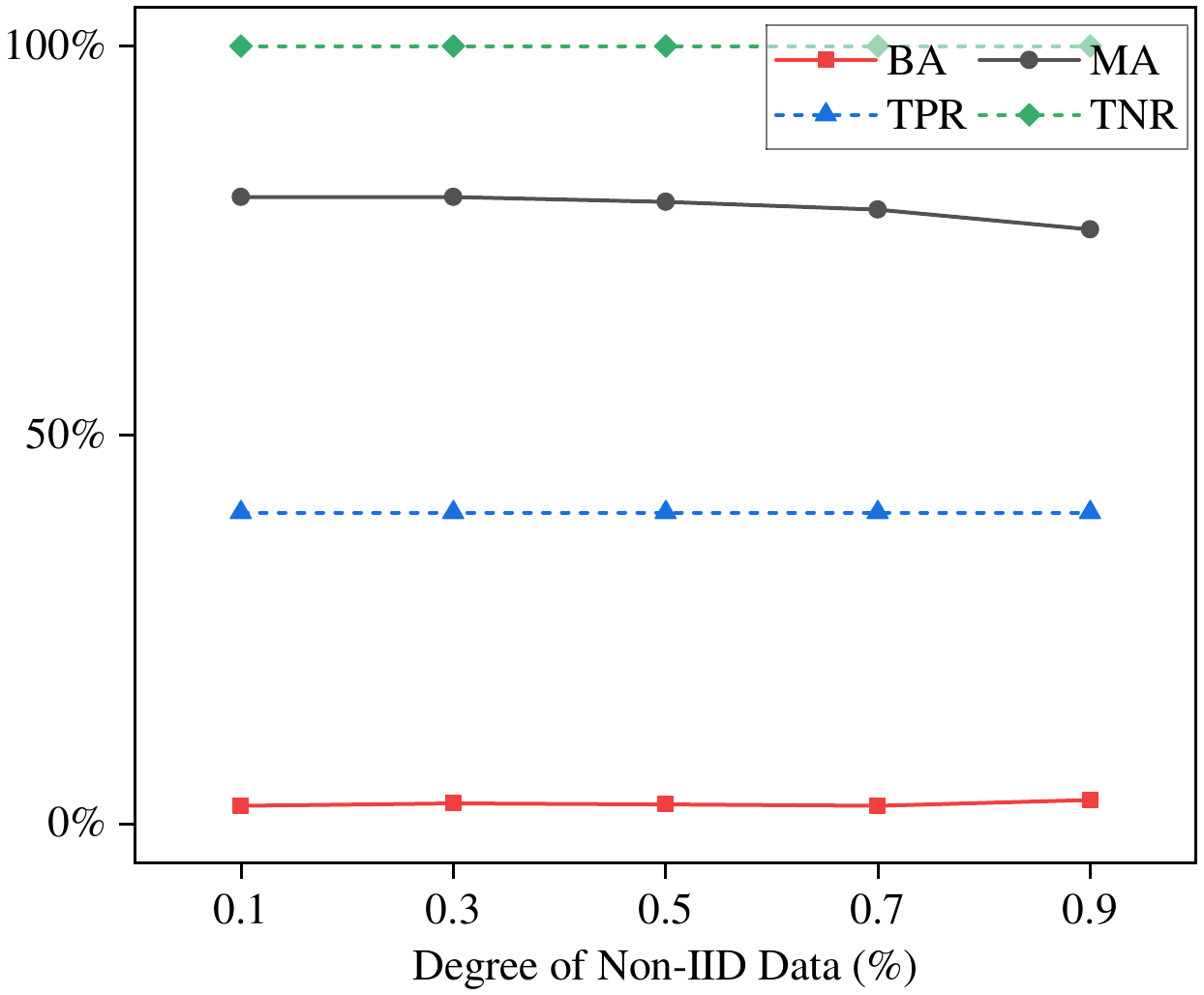}}
\subfloat[FoolsGold]{\label{fig:degree_noniid_2_foolsgold} \includegraphics[width=0.24\textwidth]{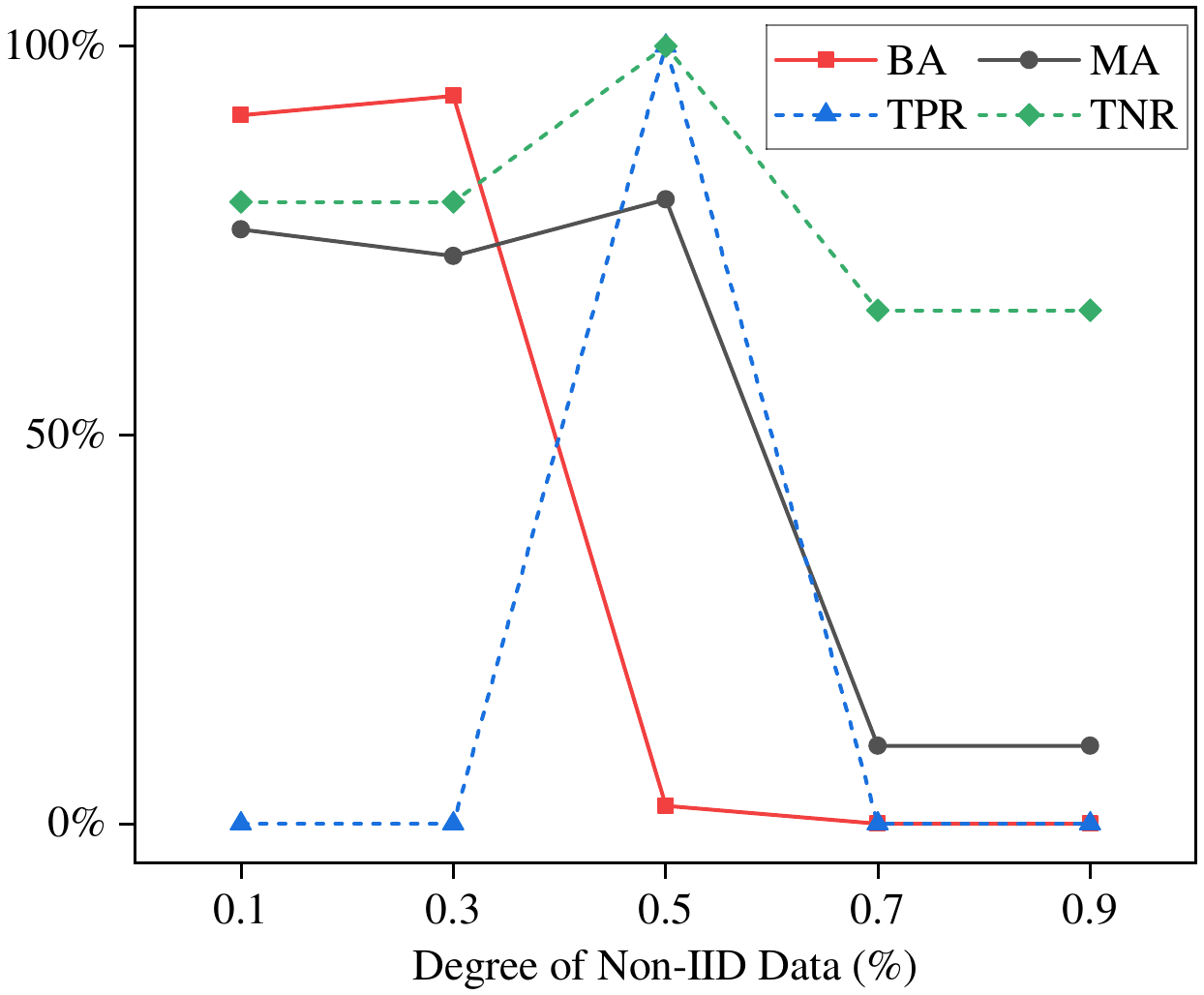}}
\subfloat[FABA]{\label{fig:degree_noniid_2_faba} \includegraphics[width=0.24\textwidth]{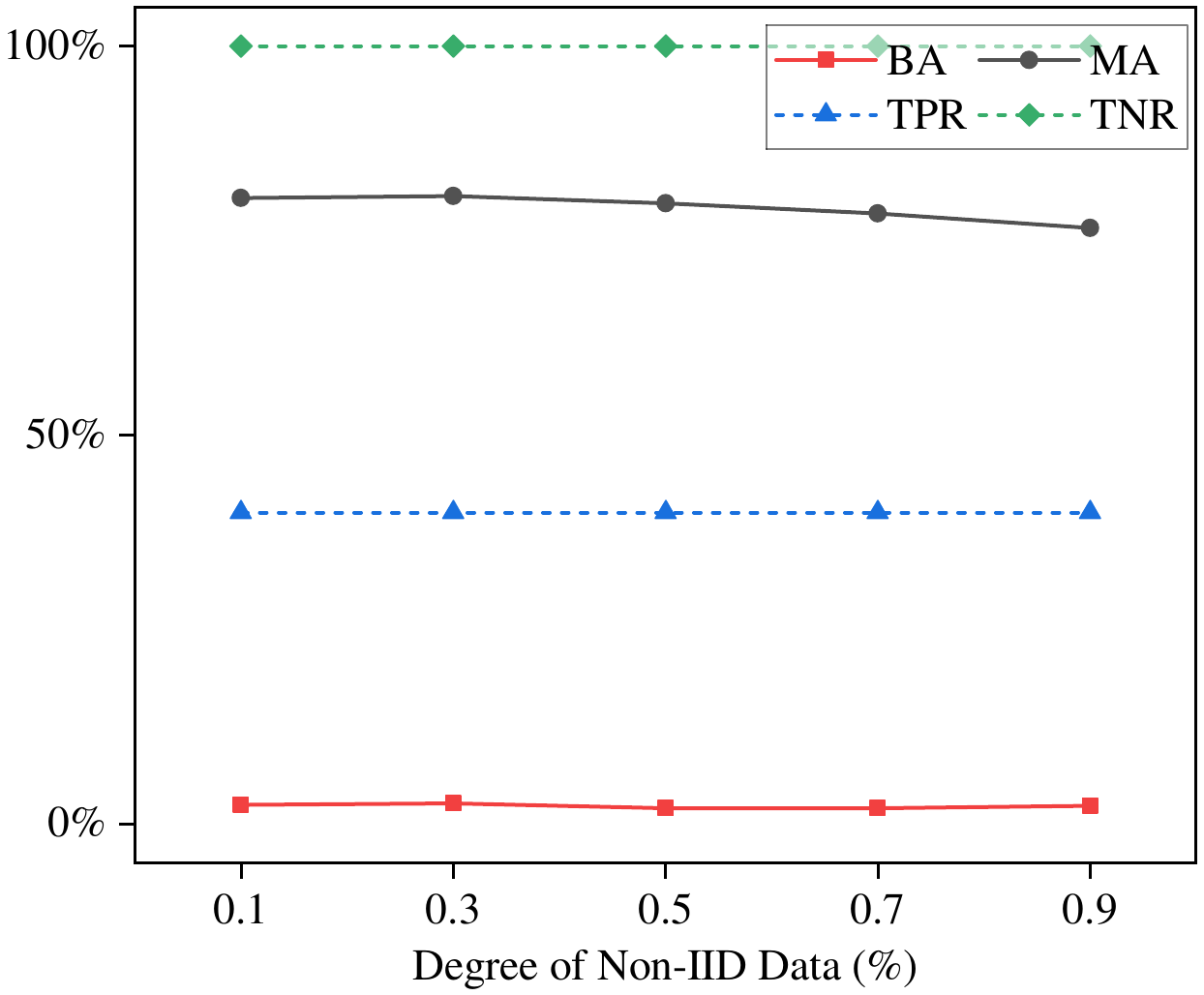}}
\subfloat[FLAME]{\label{fig:degree_noniid_2_flame} \includegraphics[width=0.24\textwidth]{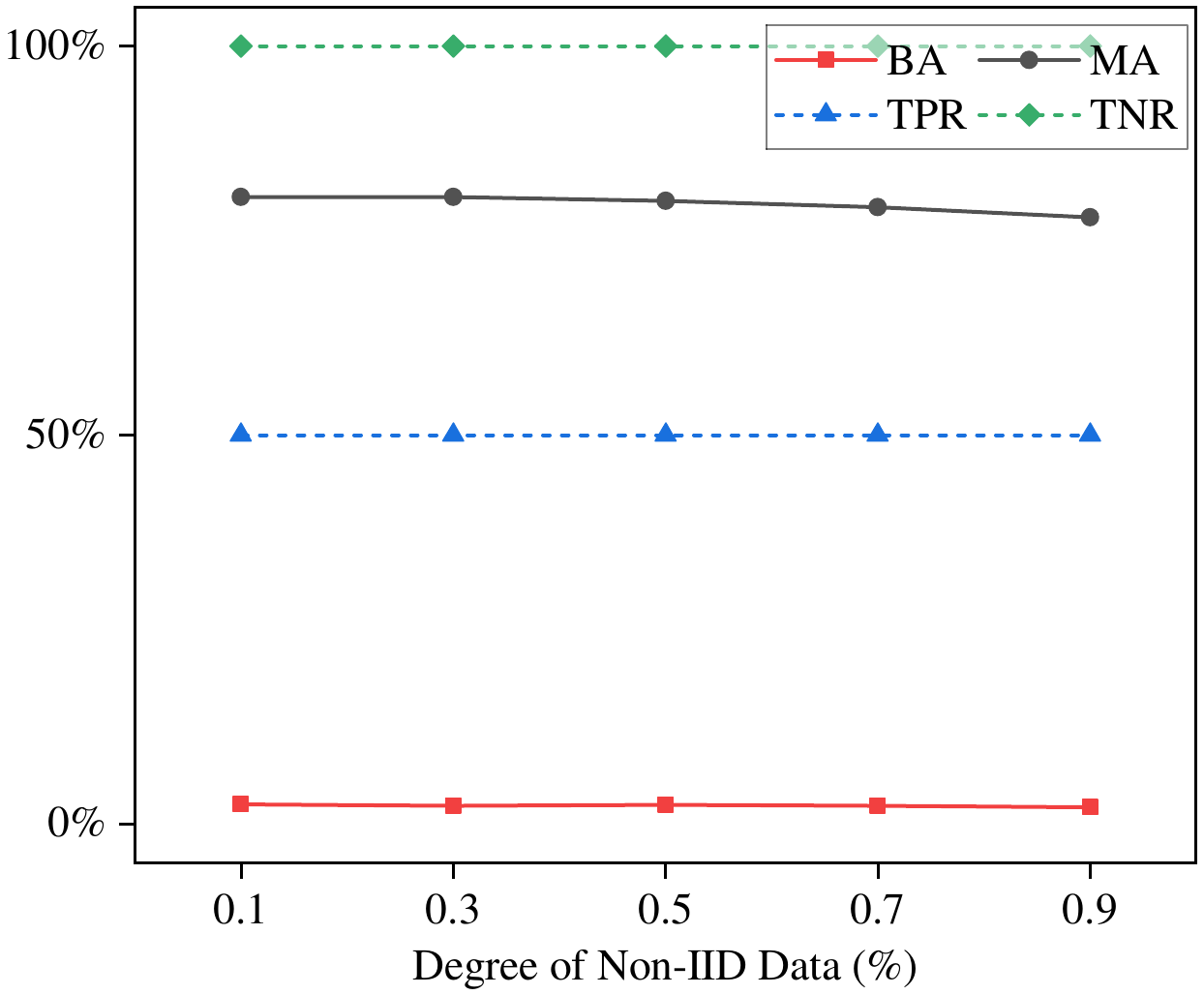}}\\

\subfloat[ABBR Multi-Krum]{\label{fig:degree_noniid_2_multi_krum_abbr} \includegraphics[width=0.24\textwidth]{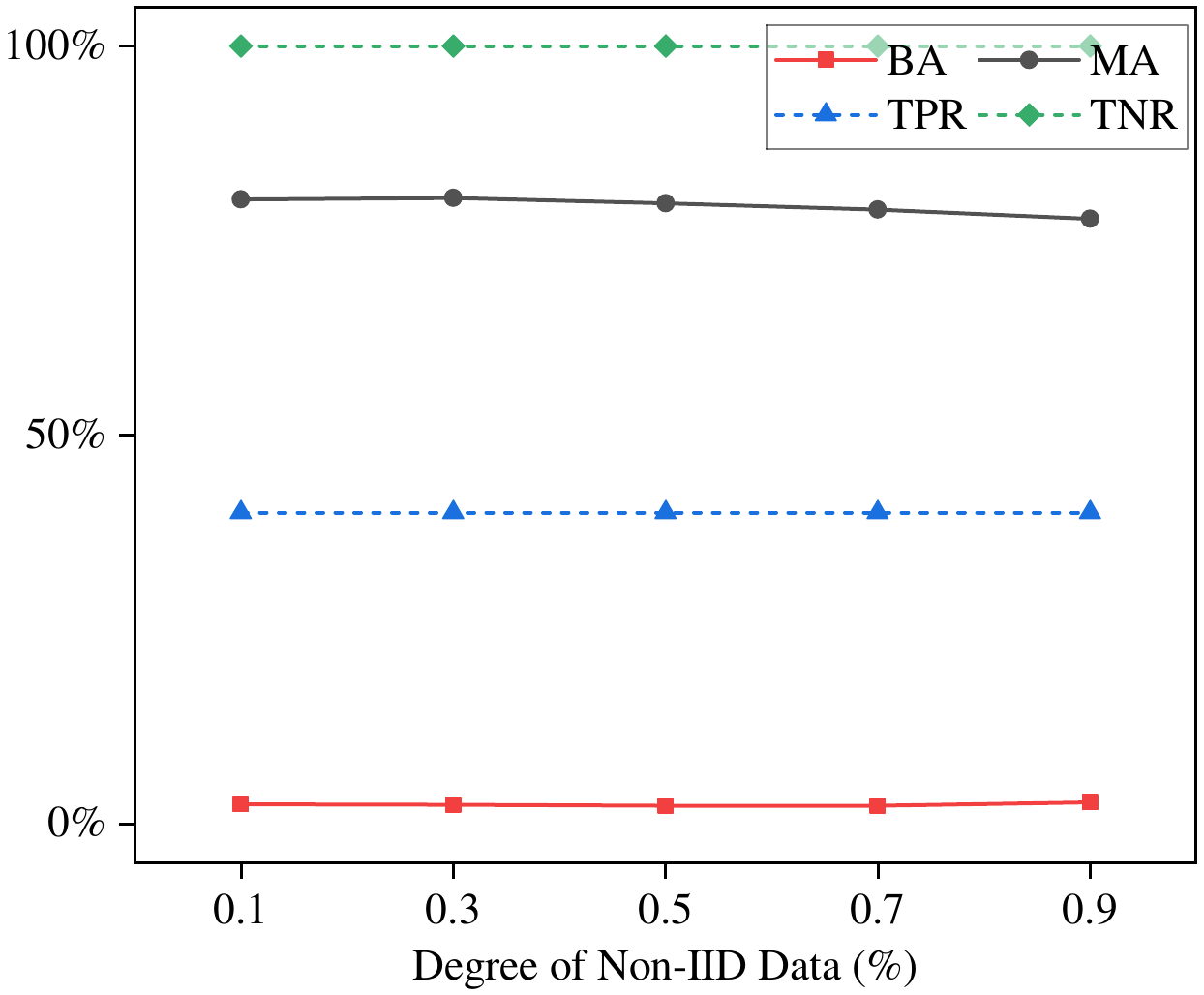}}
\subfloat[ABBR FoolsGold]{\label{fig:degree_noniid_2_foolsgold_abbr} \includegraphics[width=0.24\textwidth]{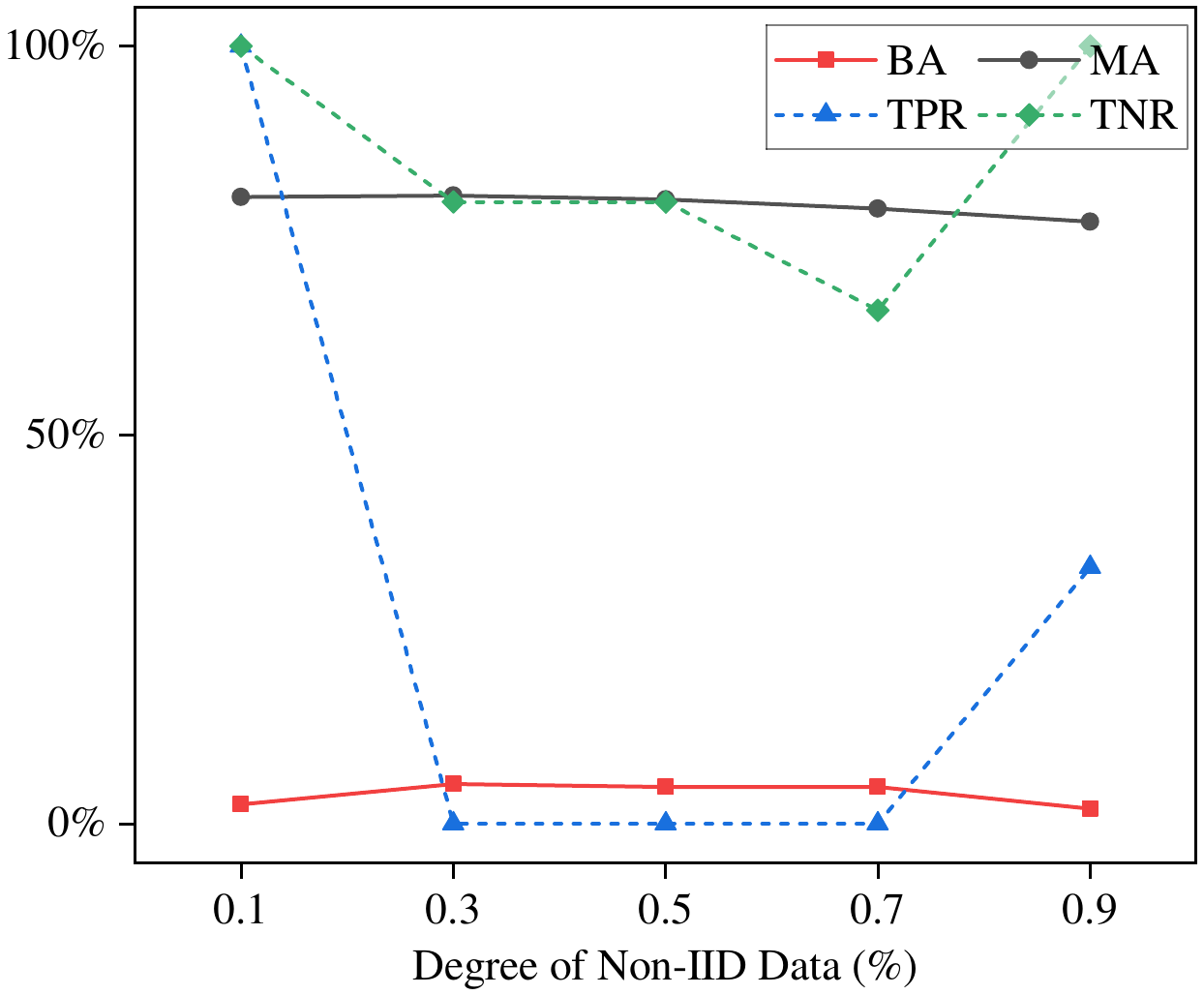}}
\subfloat[ABBR FABA]{\label{fig:degree_noniid_2_faba_abbr} \includegraphics[width=0.24\textwidth]{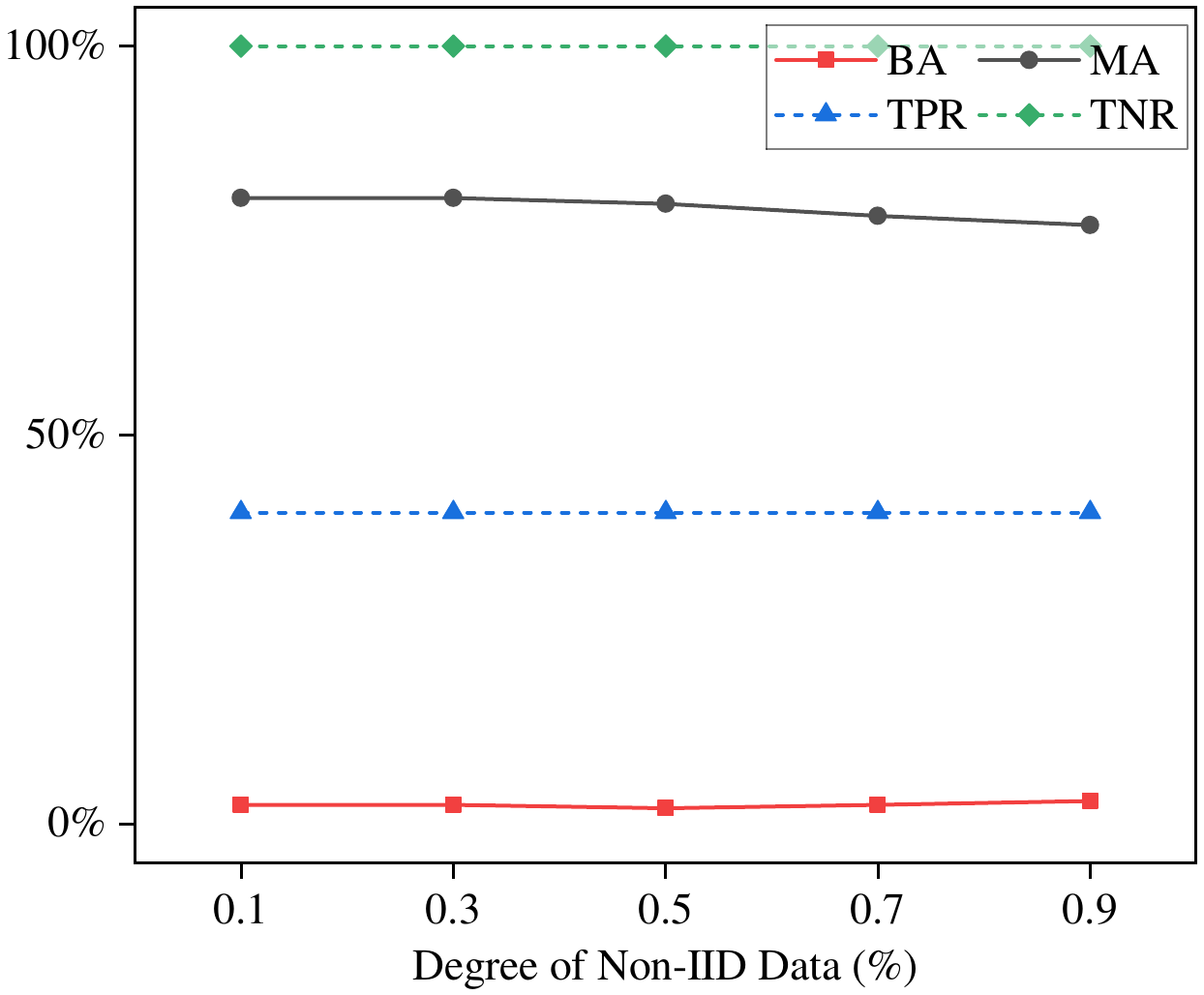}}
\subfloat[ABBR FLAME]{\label{fig:degree_noniid_2_flame_abbr} \includegraphics[width=0.24\textwidth]{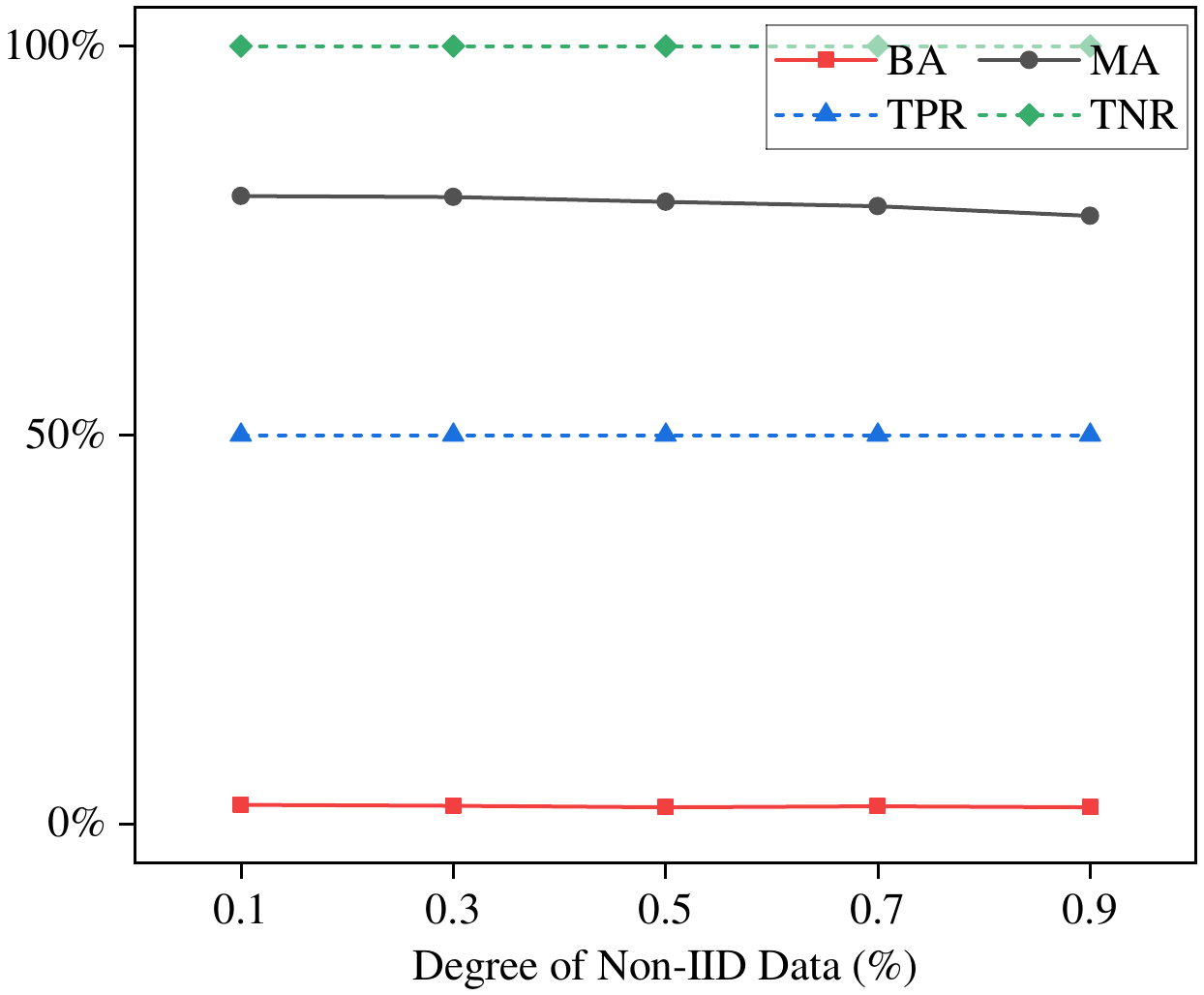}}\\

\caption{Impact of the degree of non-IID data (simulated by assigning data with label $l$ to the $l$-th client group with proportion $p \in [0,1]$ ) on CIFAR-10 dataset. (a)-(d): original vector-wise filterings, (e)-(f): ABBR versions of vector-wise filterings.}
\label{fig:alpha_2_impact}
\end{figure*}

\begin{algorithm}[hpt]
	\caption{Dimensionality Reduction}
	\label{alg:dimensionality_reduction}
	\small
        \begin{spacing}{1}
	\begin{algorithmic}[1]
	    
	    \Function{DimensionalityReduction}{${\langle L_1 \rangle}^A,\cdots,{\langle L_n \rangle}^A,d,\varepsilon,\eta$}
	    \State $s \gets$ \Call{DiffieHellman}{} \Comment{Diffie-Hellman key agreement is used to obtain the random seed.}
	    \State $k = \left \lceil \frac{4+2\eta}{\varepsilon ^ 2 - \varepsilon ^ 3}\log{n} \right \rceil $
	    \State $\mathbf{P} \gets$ \Call{PRG}{$s$} $\bmod 2$ \Comment{PRG is the Pseudorandom Generator that outputs a uniformly random matrix with the random seed $s$. The size of $\mathbf{P} $ is $d \times k$ and its elements are uniformly sampled from $\left \{ 0,1 \right \} $. }
	    \State $\mathbf{P} \left [ \mathbf{P} =0 \right ]  = -1$ \Comment{If the element $e$ of $\mathbf{P} $ is 0, then $e = -1$.}
	 
	    \For{each client $i$ in $\left [ 1,n \right ]$}
            \State ${\langle L'_i \rangle}^A \gets$ \Call{Reducing}{${\langle L_i \rangle}^A, \mathbf{P} ,d,k$}
	    \EndFor
        \State \Return ${\langle L'_1 \rangle}^A,\cdots,{\langle L'_n \rangle}^A$
        \EndFunction
        \\
        \Function{Reducing}{${\langle L \rangle}^A, \mathbf{P} , d, k$}
        \For {$i$ in $[ 0,k )$ }
        \State ${\langle L'[i] \rangle}^A = 0$
        \For {$j$ in $[0,d)$}
        \If {$\mathbf{P} [j,i] = 1$}
        \State ${\langle L'[i] \rangle}^A = \mathbf{ADD}({\langle L'[i] \rangle}^A ,{\langle L[j] \rangle}^A )$
        \Else
        \State ${\langle L'[i] \rangle}^A = \mathbf{SUB}({\langle L'[i] \rangle}^A ,{\langle L[j] \rangle}^A )$
        \EndIf
        \EndFor
        \EndFor
        \State \Return ${\langle L' \rangle}^A$
        \EndFunction

	\end{algorithmic}
	\end{spacing}
        
\end{algorithm}

\begin{algorithm}[hpt]
	\caption{Adaptive Tuning}
	\label{alg:adaptive_tuning}
	\small
        \begin{spacing}{1}
	\begin{algorithmic}[1]
	    
        \Function{AdaptiveTuning \quad}{$I_{opt}, G_{t-1}, \{{\langle L_i \rangle}^A\}_{i \in I_n}, \{{\langle L'_i \rangle}^A\}_{i \in I_n}$}
            \State $G'_{t-1}=G_{t-1} \times \mathbf{P}$ \Comment{$\mathbf{P} $ is the same $d \times k$ projection matrix in Algorithm \ref{alg:dimensionality_reduction}.}
            \For{$i$ in $(1, \dots  n)$}
            \State ${\langle e_i \rangle}^A \gets$ \Call{EuclideanDistances}{${\langle L'_i \rangle}^A, G'_{t-1}$}  \Comment{$e_i$ is the Euclidean distance between local model $L_i$ and global model $G_{t-1}$ in the low-dimensional space.}
            \EndFor
            \State ${\langle S_1 \rangle}^A \gets$ \Call{PrivateMedian}{$\{{\langle e_i \rangle}^A\}_{i \in (1,\dots,n)}$} \Comment{$S_1$ is the median Euclidean distance between all local models and global model.}
            \State ${\langle S_2 \rangle}^A \gets$ \Call{PrivateMin}{$\{{\langle e_i \rangle}^A\}_{i \in (1,\dots,n)}$} \Comment{$S_2$ is the minimum Euclidean distance between all local models and global model.}
            \For{$i$ in $I_{opt}$}
            \State $\gamma_i =
                \begin{cases}
                 \frac{ S_2 }{e_i} & \text{ if } {\langle e_i \rangle}^A > {\langle S_1 \rangle}^A \\
                 1 &  \text{ if } {\langle e_i \rangle}^A \le {\langle S_1 \rangle}^A
                \end{cases}$ \Comment{We use the private operators MUX and CMP in ABY to construct private functions, and then reconstruct $\frac{ S_2 }{e_i}$ from the secret sharing ${\langle \frac{S_2}{e_i}\rangle}^A$. }
            \State ${\langle L_i \rangle}^A = G_{t-1} + \gamma_i \cdot ({\langle L_i \rangle}^A - G_{t-1}) $ \Comment{Since Only the local model $L_i$ is in the form of secret sharing, no expensive operators are introduced in STPC.}
            \EndFor
	  
        \State \Return ${\langle L_1 \rangle}^A,\cdots,{\langle L_n \rangle}^A$
        \EndFunction
        
	\end{algorithmic}
	\end{spacing}
        
\end{algorithm}

\begin{algorithm}[hpt]
	\caption{Adaptive-MD}
	\label{alg:Adaptive-MD}
	\small
        \begin{spacing}{1}
	\begin{algorithmic}[1]
	    \Statex \textbf{Input:} Benign gradients $\bigtriangledown_{i\in\{n\}}$, the number of Byzantine clients $c$, scale factor $\gamma, \gamma_{succ}$, step size $s$, scale threshold $\tau$, loss threshold $L$, iterations $I$.
        \Statex \textbf{Output:} Byzantine gradients $\bigtriangledown_{i\in\{c\}}$
        \State $\gamma \gets 10$
        \State $s \gets 5$
        \State $\bigtriangledown_{ben} \gets $ \Call{Mean}{$\bigtriangledown_{i\in\{n\}}$}
        \State $\bigtriangledown_{sgn}^p \gets $\Call{Sign}{$\bigtriangledown_{i\in\{n\}}$}
        \For{$i$ in $[1, I]$}
        \State $\bigtriangledown_{byz} \gets (\bigtriangledown_{ben} - \gamma \cdot \bigtriangledown_{sgn}^p)$
        \State $\bigtriangledown_{i\in\{c\}} \gets \bigtriangledown_{byz}$
        \State $\bigtriangledown_{agg} \gets$ \Call{RobustAggregation}{ $\bigtriangledown_{i\in\{c\}},\bigtriangledown_{i\in\{n\}}$}
        \State $l = \left \| \bigtriangledown_{agg}- \bigtriangledown_{byz}\right \| $
        \If { $l > L$}
        \State $\gamma_{succ} \gets \gamma$
        \State $\gamma \gets (\gamma + s/2)$
        \Else
        \State $\gamma \gets (\gamma - s/2)$
        \EndIf
        \State $s \gets s/2$
        
        \EndFor
        \State $\bigtriangledown_{i\in\{c\}} \gets (\bigtriangledown_{ben} - \gamma_{succ} \cdot \bigtriangledown_{sgn}^p)$
        \Statex \textbf{Return:} $\bigtriangledown_{i\in\{c\}}$
	\end{algorithmic}
	\end{spacing}
        
\end{algorithm}

\begin{algorithm}[hpt]
	\caption{Adaptive-AT}
	\label{alg:Adaptive-AT}
	\small
        \begin{spacing}{1}
	\begin{algorithmic}[1]
	    \Statex \textbf{Input:} Benign gradients $\bigtriangledown_{i\in\{n\}}$, the number of Byzantine clients $c$, scale factor $\gamma$, iterations $I$.
        \Statex \textbf{Output:} Byzantine gradients $\bigtriangledown_{i\in\{c\}}$
        \State $\gamma \gets 50$
        \State $\bigtriangledown_{ben} \gets $ \Call{Mean}{$\bigtriangledown_{i\in\{n\}}$}
        \State $\bigtriangledown_{sgn}^p \gets $\Call{Sign}{$\bigtriangledown_{i\in\{n\}}$}
        \For{$i$ in $[1, I]$}
        \State $\bigtriangledown_{byz} \gets (\bigtriangledown_{ben} - \gamma \cdot \bigtriangledown_{sgn}^p)$
        \State $\bigtriangledown_{i\in\{c\}} \gets \bigtriangledown_{byz}$
        \State $\tau \gets$ \Call{MedianNorm}{ $\bigtriangledown_{i\in\{c\}},\bigtriangledown_{i\in\{n\}}$}
        
        \State $l \gets \left \| \bigtriangledown_{byz}\right \| $
        \If { $l > \tau$}
        \State $\gamma \gets \gamma/2$
        \Else
        \State break
        \EndIf
        
        \EndFor
        \State $\bigtriangledown_{i\in\{c\}} \gets (\bigtriangledown_{byz} - \gamma_{succ} \cdot \bigtriangledown_{sgn}^p)$
        \Statex \textbf{Return:} $\bigtriangledown_{i\in\{c\}}$
	\end{algorithmic}
	\end{spacing}
        
\end{algorithm}

\subsection{The detailed setup of Experiments}
\label{detailed_setup}
\myparatight{Datasets} MNIST and Fashion-MNIST all consist of 70,000 28$\times$28-dimensional grayscale images in 10 categories. MNIST is a large collection of handwritten digits with 60,000 images for training and 10,000 for testing. CIFAR-10 contains 60,000 color images of 32$\times$32 dimensions in 10 classes, with 5,000 training and 1,000 test images per class. The Tiny-ImageNet contains 120,000 images of 200 classes and each class has 500 training images, 50 validation images, and 50 test images.

\myparatight{STPC Protocol} Since ABY currently does not support floating point numbers well, we convert the model parameters to integers with precision $p=20$ by default and set ABY's sharing length to 64 bits. We simulate two separate servers by creating two processes on a server equipped with two Intel Xeon Platinum 8375C CPUs with 2.90GHz and 256GB RAM, where the LAN between the two servers is about 812 Mbit/s, and the RTT is about 0.1 ms. 

\myparatight{Evaluation Metrics} In performance evaluation, we consider the metrics of \emph{Runtime} and \emph{Communication Overhead} in the setup phase and online phase, which indicate the efficiency of privacy-preserving computation. The execution of the ABY protocol can be divided into setup phase and online phase \cite{DBLP:conf/ndss/Demmler0Z15}. In the setup phase, the two servers generate associated random numbers that are independent of the input data, which can be pre-executed. In the online phase, the two servers use the associated random numbers to calculate private functions and reconstruct the final result. We use the metric statistics function in the ABY protocol to obtain the runtime and communication overhead of our private protocol in the setup phase and online phase.

In robustness evaluation, we consider a set of metrics as follows: \emph{BA - Backdoor Accuracy} indicates the accuracy of the global model in the backdoor task. \emph{MA - Main Task Accuracy} indicates the accuracy of the global model in the main task. The targeted backdoor adversary aims to maximize \emph{BA} while maintaing the \emph{MA}, and the untargeted  adversary aims to minimize \emph{MA} only. the \emph{BA} and \emph{MA} of the global model can provide a better measure of robustness performance. \emph{TPR - True Positive Rate} indicates the ratio of the number of Byzantine local models correctly filtered to the total number of models being filtered. \emph{TNR - True Negative Rate} indicates the ratio of the number of benign local models correctly accepted to the total number of models being accepted.  \emph{TPR} and \emph{TNR} can better measure the impact of dimensionality reduction on filtering.

\myparatight{Constrain-and-Scale attack} This attack makes the backdoor local model excellent concealment by adding a penalty item to the local training loss of the Byzantine clients, thus creating a better attack effect.

\myparatight{DBA attack} This attack makes the backdoor local model more hidden by distributing the backdoor trigger into the local datasets of different Byzantine clients and embedding the entire trigger into the global model by aggregating all backdoor local models.

\myparatight{Edge-Case and PGD attacks} This attack trains a more covert backdoor local model by using data from the tail of the distribution, while the PGD attack will further improve concealment by limiting model updates to a ball at a fixed frequency during its training process.

\myparatight{Label flipping attack} This attack is a data poisoning attack method that modifies the label of training data to reduce the model's accuracy.

\myparatight{Gaussian attack} This attack randomly sample the parameter of local model from the Gaussian distribution which is estimated by the before-attack models of all Byzantine clients.
     
\subsection{Experiment Results}
\label{appendix:results}
In this section, we conduct two additional experiments to explore the robustness of our ABBR framework. 

\myparatight{Varying the percentage of Byzantine clients}
We investigate the impact of the Byzantine client rate (\emph{BCR} = $\frac{c}{n}$, where $c$ is the number of Byzantine clients) when all clients participate in each iteration. In this experiment, we use the Constrain-and-Scale attack and CIFAR-10 dataset. As shown in Figure~\ref{fig:all_bcr_impact}, most ABBR versions exhibit comparable or even better robustness than the baselines.

\myparatight{Varying the degree of non-IID data}
We conduct experiments by varying the degree of non-IID data (simulated by assigning data with label $l$ to the $l$-th client group at a proportion $p \in [0,1]$) on the CIFAR-10 dataset. Figure~\ref{fig:alpha_2_impact} shows that ABBR versions demonstrate robustness comparable to the baselines across different non-IID levels. 

\subsection{Algorithms}
In this section, we present the algorithm of dimensionality reduction in Algorithm~\ref{alg:dimensionality_reduction} and the algorithm of adaptive tuning in Algorithm~\ref{alg:adaptive_tuning}. In addition, we present the algorithm of our two adaptive attacks in Algorithm~\ref{alg:Adaptive-MD} and~\ref{alg:Adaptive-AT}.

\begin{IEEEbiography}[{\includegraphics[width=1in,height=1.25in,clip,keepaspectratio]{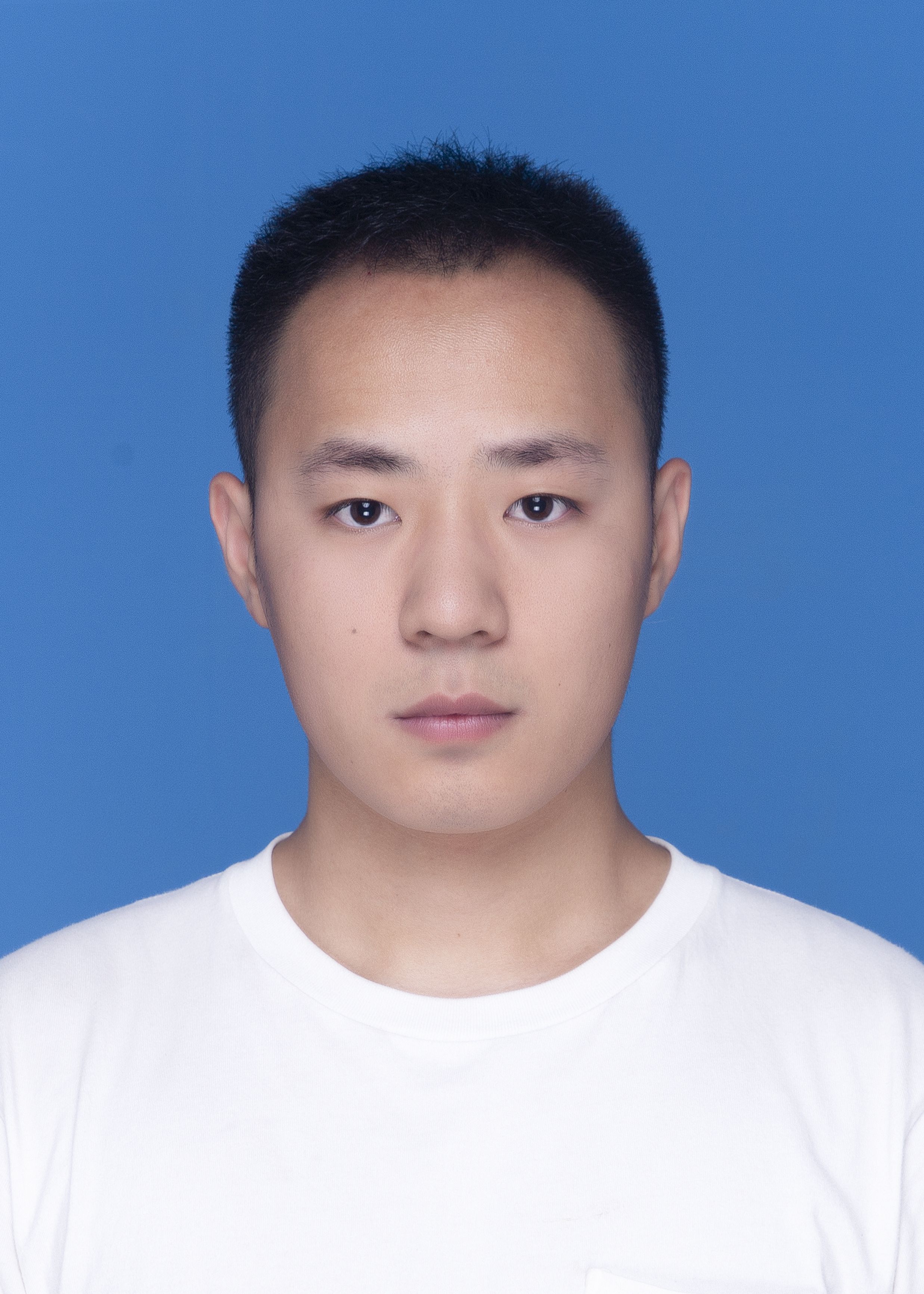}}]{Baolei~Zhang}
received his B.S. degree in software engineering from Henan University, China, in
2020. He is currently pursuing his Ph.D. degree in computer science and technology at Nankai University, China. His research interests include artificial intelligence security and privacy-preserving computation.
\end{IEEEbiography}

\begin{IEEEbiography}[{\includegraphics[width=1in,height=1.25in,clip,keepaspectratio]{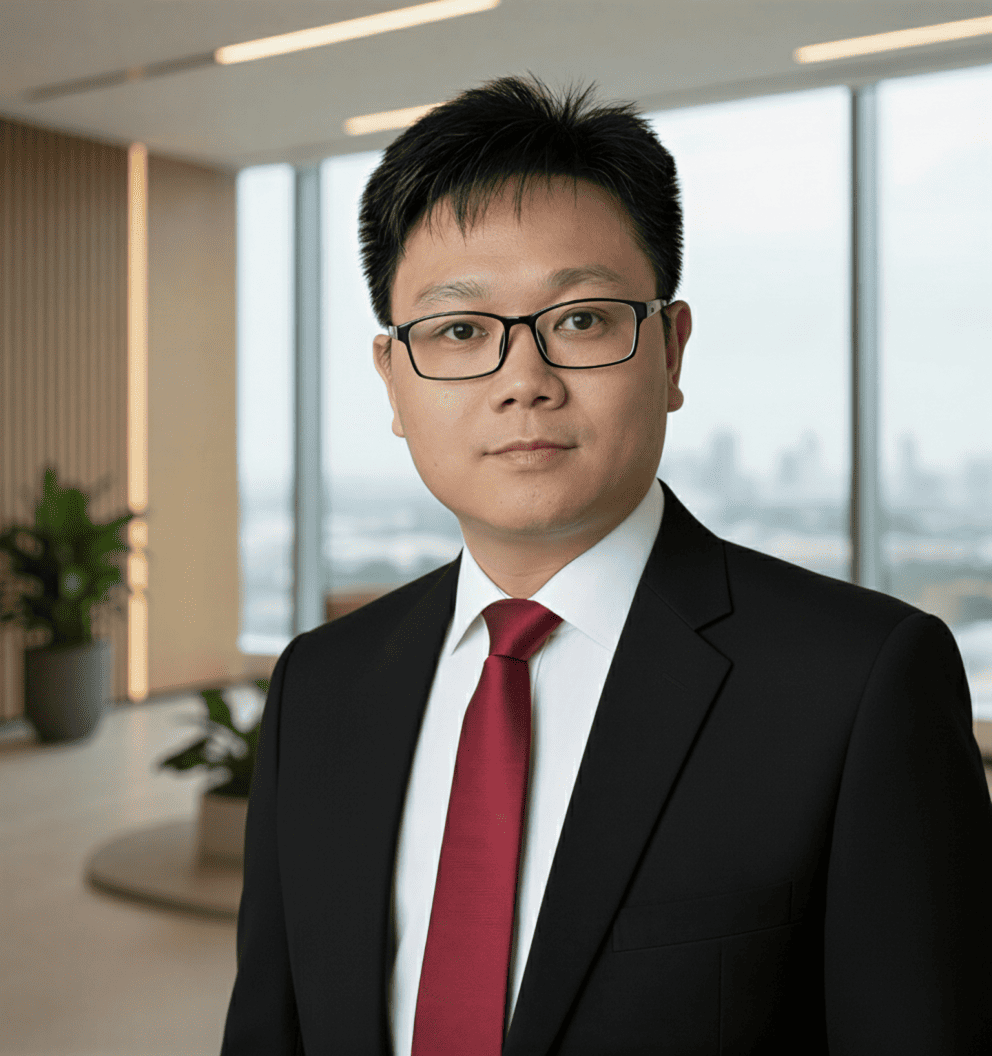}}]{Minghong~Fang}
is a tenure-track Assistant Professor in the Department of Computer Science and Engineering at the University of Louisville. From 2022 to 2024, he served as a Postdoctoral Associate in the Department of Electrical and Computer Engineering at Duke University. Dr. Fang earned his Ph.D. in Electrical and Computer Engineering from The Ohio State University in August 2022. His research broadly focuses on AI safety and security, with an emphasis on understanding the vulnerabilities of modern machine learning systems and developing principled methods to enhance their robustness and trustworthiness across diverse real-world applications.
\end{IEEEbiography}

\begin{IEEEbiography}[{\includegraphics[width=1in,height=1.25in,clip,keepaspectratio]{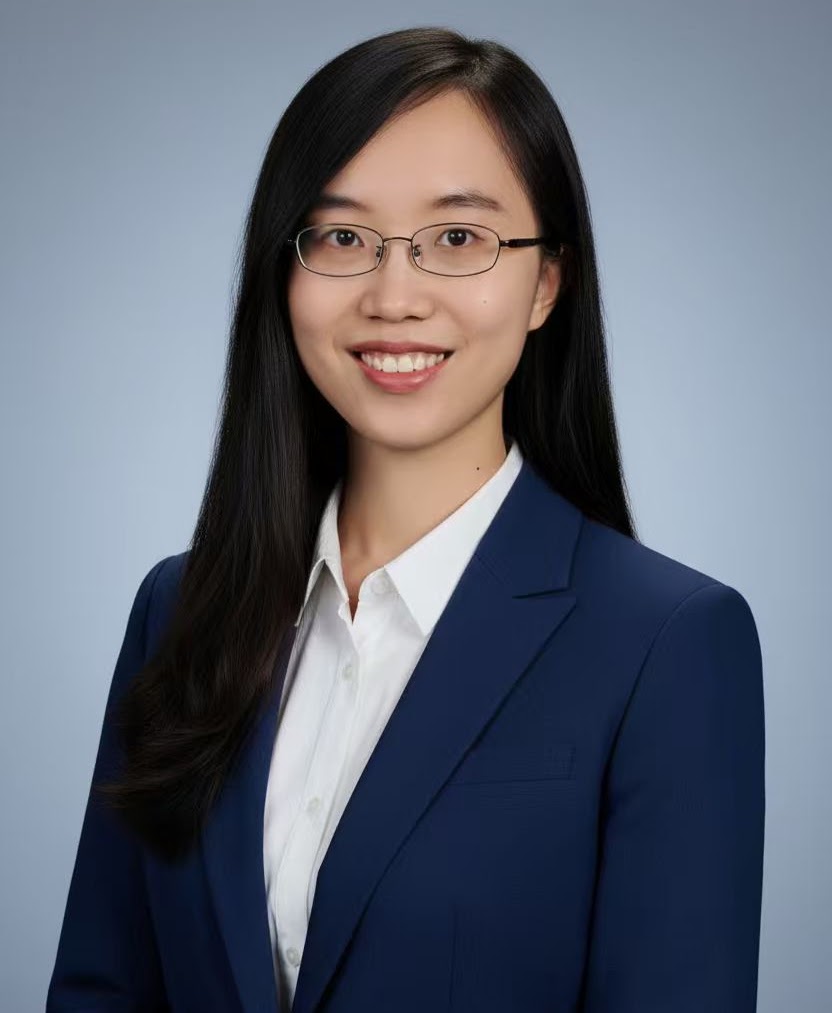}}]{Zhuqing~Liu}
is a tenure-track Assistant Professor in the Department of Computer Science and Engineering at the University of North Texas. She received her Ph.D. in Electrical and Computer Engineering from The Ohio State University in 2024. Her research lies at the intersection of machine learning, artificial intelligence, and trustworthy systems. Her work focuses on adversarial robustness, security forensics in emerging AI pipelines, federated and multi-agent learning, and the reliability of large foundation models.
\end{IEEEbiography}

\begin{IEEEbiography}[{\includegraphics[width=1in,height=1.25in,clip,keepaspectratio]{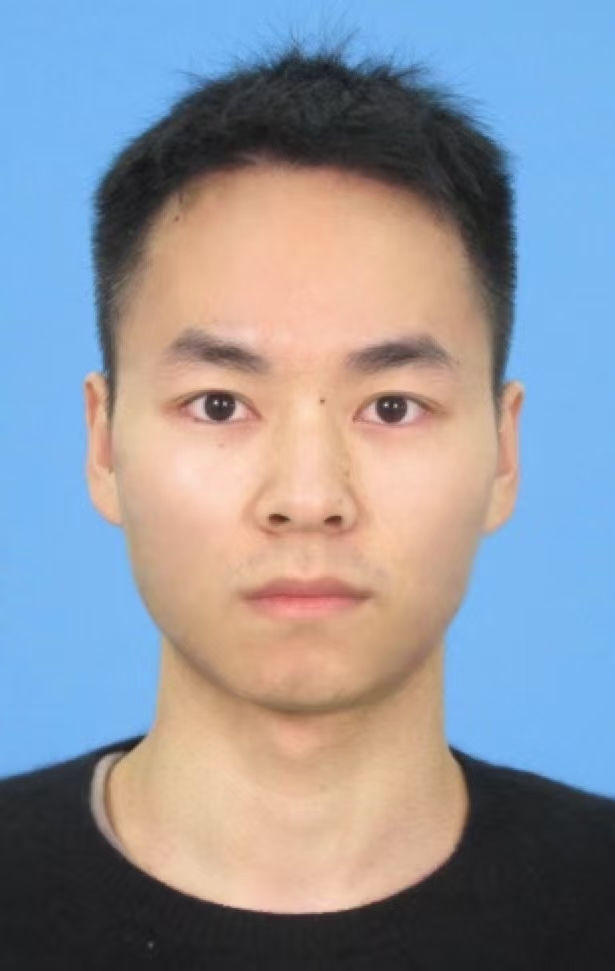}}]{Biao~Yi}
received his B.S. degree from the University of South China in 2019 and his M.S. degree from Shanghai University in 2022. He is currently pursuing his Ph.D. at Nankai University, where his research focuses on trustworthy LLM and artificial intelligence security.
\end{IEEEbiography}

\begin{IEEEbiography}[{\includegraphics[width=1in,height=1.25in,clip,keepaspectratio]{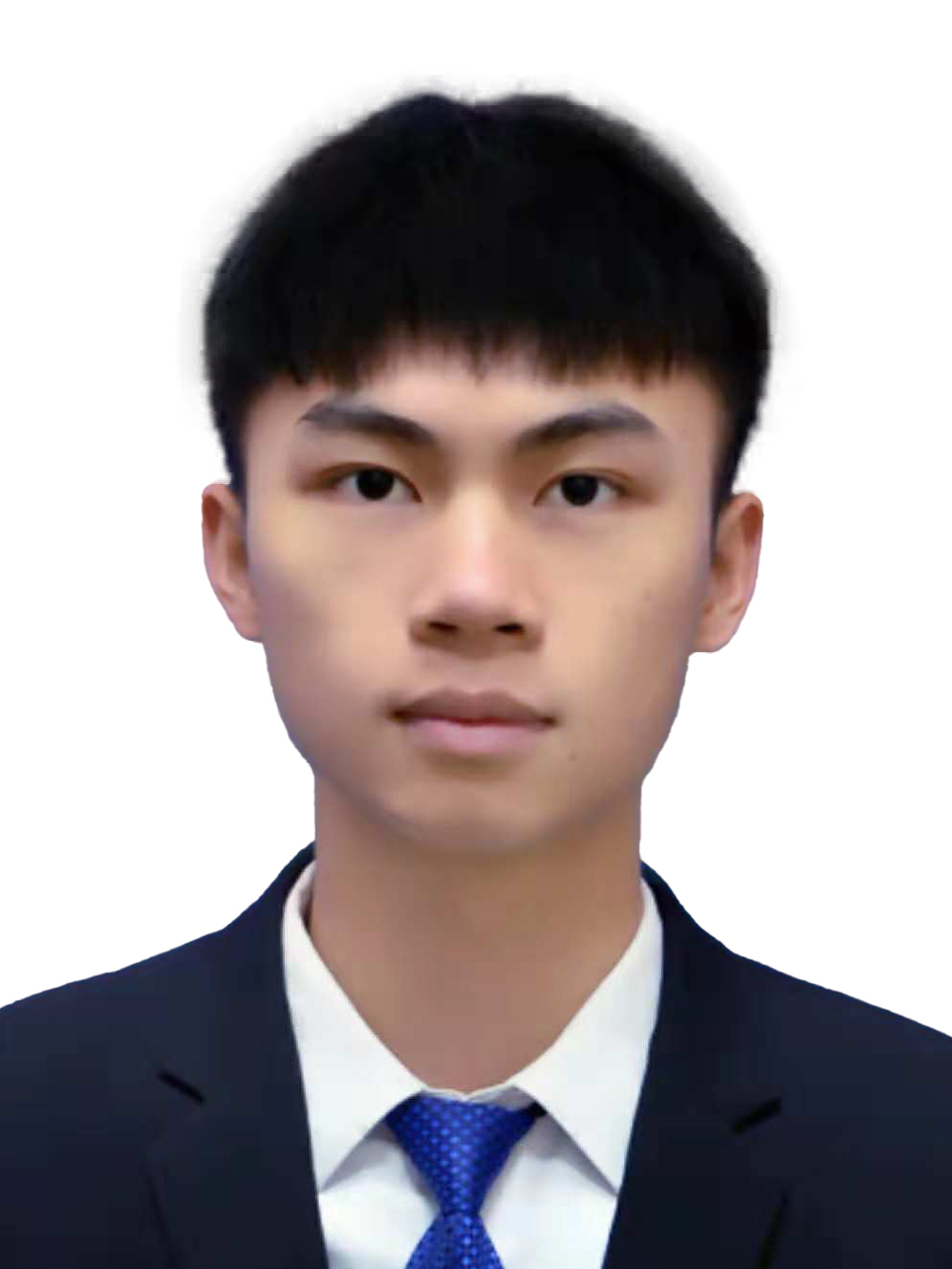}}]{Peizhao~Zhou}
received his B.S. degree from Shandong University of Science and Technology and his M.S. degree  from Northeast Normal University, respectively. He is currently pursuing his Ph.D. at Nankai University. His research interests include privacy-preserving data mining, analytics, and queries.
\end{IEEEbiography}

\begin{IEEEbiography}[{\includegraphics[width=1in,height=1.25in,clip,keepaspectratio]{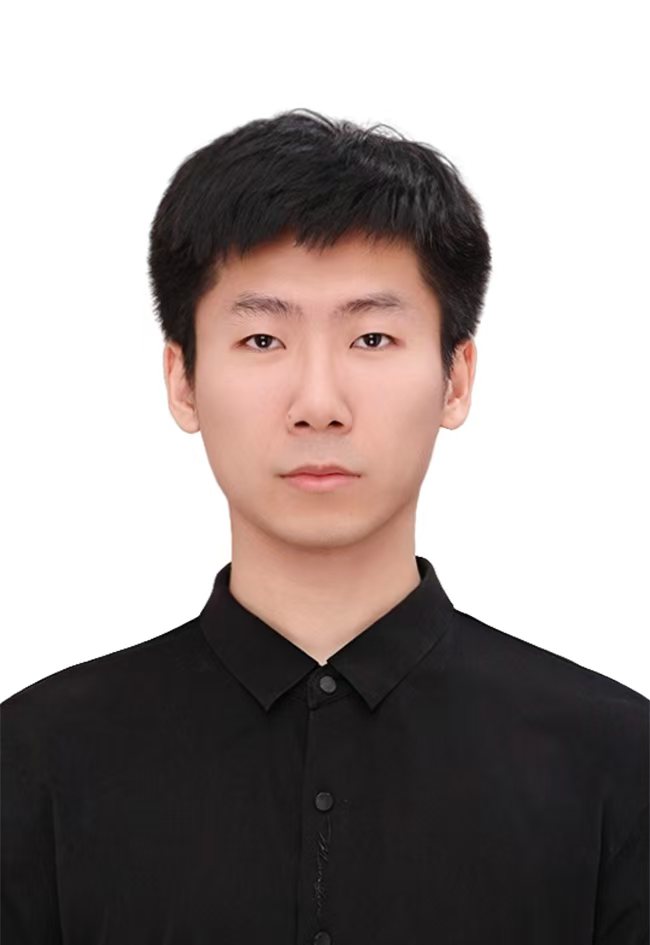}}]{Yuan~Wang}
received his B.S. degree  and M.S. degree from Nankai University, China,  in 2021 and 2025, respectively. His research focuses on artificial intelligence security and the security of large language models
\end{IEEEbiography}

\begin{IEEEbiography}[{\includegraphics[width=1in,height=1.25in,clip,keepaspectratio]{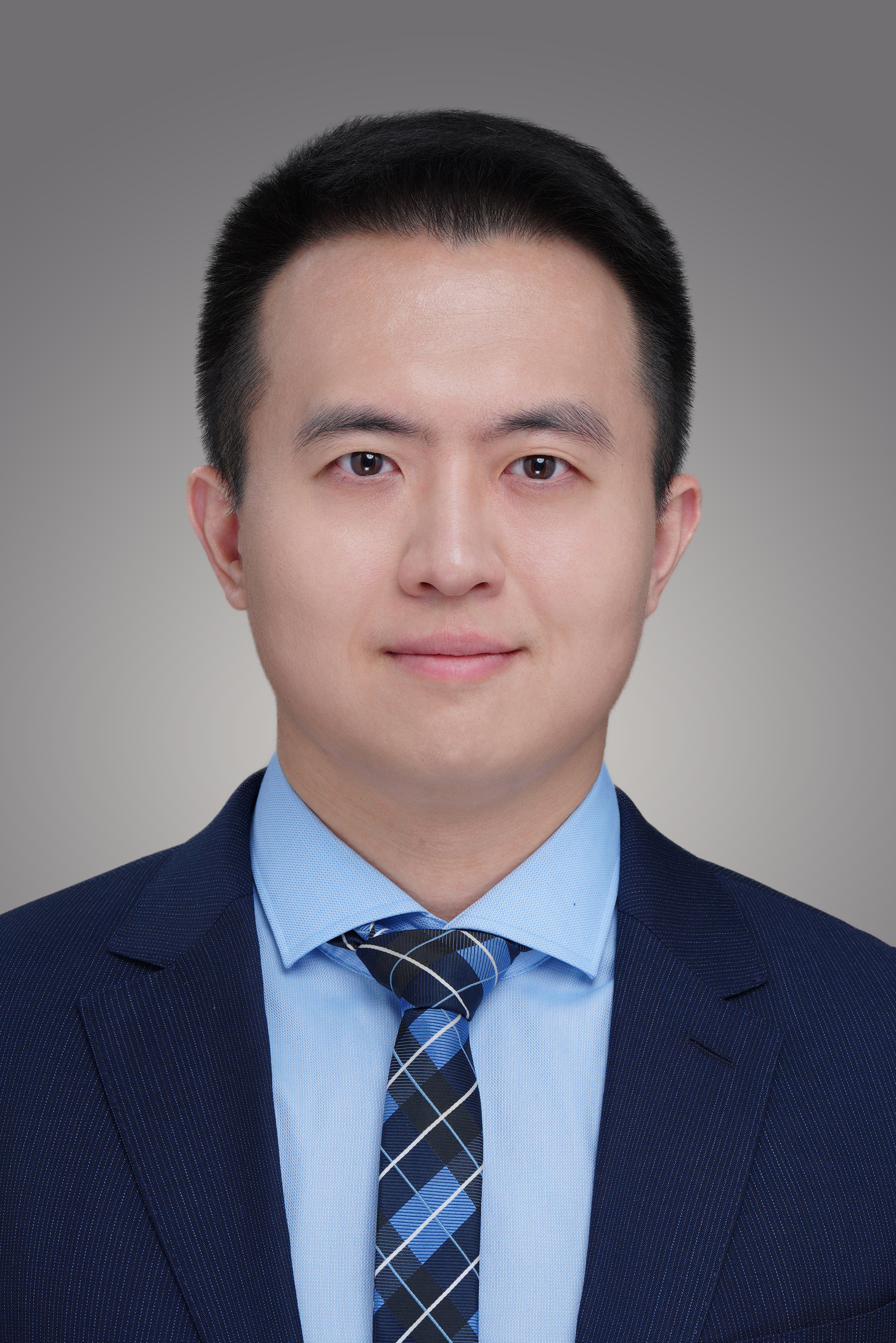}}]{Tong~Li}
received his B.S. and M.S. from Taiyuan University of Technology and Beijing University of Technology, in 2011 and 2014, respectively, both in Computer Science \& Technology. He got his Ph.D degree in information security from Nankai University in 2017. His research interests include applied cryptography, privacy-preserving computation, and secure machine learning.
\end{IEEEbiography}

\begin{IEEEbiography}[{\includegraphics[width=1in,height=1.25in,clip,keepaspectratio]{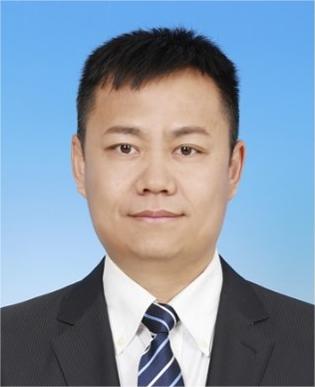}}]{Zheli~Liu}
received his B.S. and M.S. degrees in computer science and a Ph.D. degree in computer application from Jilin University, China, in 2002, 2005, and 2009, respectively. After a postdoctoral fellowship with Nankai University, he joined the College of Cyber Science, Nankai University, in 2011. He is currently a professor at Nankai University. His research interests include applied cryptography and data privacy protection.
\end{IEEEbiography}
\end{document}